\newcommand{\prob}{\mathbb{P}}
\newcommand{\Nitem}{M}

\newcommand{\nitem}{m}
\newcommand{\Nround}{T}
\newcommand{\nround}{t}

\documentclass[msom,nonblindrev]{informs3}
\usepackage{hyperref}
\usepackage{url}
 \usepackage{natbib}
 \bibpunct[, ]{(}{)}{,}{a}{}{,}%
 \def\bibfont{\scriptsize}%
\usepackage{bm}
\usepackage{zref-totpages}
\usepackage{wrapfig} 

\usepackage{pgfplots}
\usepgfplotslibrary{fillbetween}

\usepackage{booktabs} 
\usepackage{booktabs} 
\usepackage[ruled]{algorithm2e} 
\usepackage{subcaption}

\OneAndAHalfSpacedXI

\RequirePackage{amssymb,amsmath,ifthen,url,graphicx,color,array,theorem}
\TheoremsNumberedThrough     

\begin{document}


\TITLE{Learning in Repeated Multi-Unit Pay-As-Bid Auctions}

\ARTICLEAUTHORS{
\AUTHOR{Rigel Galgana}
\AFF{Operations Research Center, Massachusetts Institute of Technology, \EMAIL{rgalgana@mit.edu} \URL{}}
\AUTHOR{Negin Golrezaei}
\AFF{Sloan School of Management, Massachusetts Institute of Technology,  \EMAIL{golrezaei@mit.edu} \URL{}}
}

\ABSTRACT{

\textbf{Problem definition:} Motivated by Carbon Emissions Trading Schemes, Treasury Auctions, Procurement Auctions, and Wholesale Electricity Markets, which all involve the auctioning of homogeneous multiple units, we consider the problem of learning how to bid in repeated multi-unit pay-as-bid auctions. In each of these auctions, a large number of (identical) items are to be allocated to the largest submitted bids, where the price of each  of the winning bids is equal to the bid itself. In this work, we study the problem of optimizing bidding strategies from the perspective of a single bidder.

\textbf{Methodology/results:} Effective bidding in pay-as-bid (PAB) auctions is complex due to the combinatorial nature of the action space. We show that a utility decoupling trick enables a polynomial time algorithm to solve the offline problem where competing bids are known in advance. Leveraging this structure, we design efficient algorithms for the online problem under both full information and bandit feedback settings that achieve an upper bound on regret of $O(\Nitem \sqrt{\Nround \log \Nround})$ and $O(\Nitem \Nround^{\frac{2}{3}} \sqrt{\log \Nround})$ respectively, where $\Nitem$ is the number of units demanded by the bidder and  $\Nround$ is the total number of auctions. {\color{black}We accompany these results with a regret lower bound of $\Omega(M\sqrt{T})$ for the full information setting and $\Omega (M^{2/3}T^{2/3})$ for the bandit setting. We also present additional findings on the characterization of PAB equilibria.} 

\textbf{Managerial implications:} While the Nash equilibria of PAB auctions possess nice properties such as winning bid uniformity and high welfare \& revenue, they are not guaranteed under no regret learning dynamics. Nevertheless, our simulations suggest that these properties hold anyways, regardless of Nash equilibrium existence. Compared to its uniform price counterpart, the PAB dynamics converge faster and achieve higher revenue, making PAB appealing whenever revenue holds significant social value---e.g., ETS and Treasury Auctions.

\noindent


}
\maketitle





\section{Introduction}

Homogeneous multi-unit auctions, a subset of combinatorial auctions, are extensively utilized to auction off large quantities of identical items. Examples include Carbon Emissions Trading Schemes \citep{CarbonTaxVsCapTrade2013, LessonsLearned2017, Kira2019}, US Treasury Auctions \citep{TreasuryAuction2005, TreasuryUnifOrDisc2000}, Procurement Auctions \citep{Procurement2006}, and Wholesale Electricity Markets \citep{WholesaleElectricityMarkets2008, BiddingElectricity2003, DesigningElectricityAuctions2006}. In these auctions, bidders submit bid vectors and are allocated goods and charged payments based on the auction's rules.

The uniform price and pay-as-bid (PAB) mechanisms, natural extensions of the second and first price sealed bid auctions respectively, are commonly used. Bidders are allocated units in descending order of their bids; they are then charged either the lowest winning bid (uniform price) or their own bid (PAB). This paper focuses on the PAB auction format due to the recent industry-wide shift towards first price auctions, spurred by demands for price transparency and simpler revenue management \citep{fpadisplay2021, googlefpa2019}.

Despite their prevalence, PAB auctions pose significant challenges for participants trying to bid effectively, as outlined by \citep{CombinatorialAuctionDesign2003}. Bidders must balance the likelihood of winning units against the potential costs, complicated by the requirement for monotone bid vectors. Bidding too conservatively reduces the chances of securing units in later rounds, whereas overly aggressive bids can lead to inflated payments. This dilemma also affects uniform price auctions, though its impact on bidder behavior, welfare, and revenue remains poorly understood over time.

Consider the context of pollution license auctions under Carbon Emissions Trading Schemes (ETS), which compel large-scale polluters to buy pollution licenses via auction. These licenses, capping permissible carbon emissions, serve dual purposes: they discourage inefficient pollution through payment penalties and fund investments in renewable energy and green technology \citep{gregor2023review, chinagreen2022weng}. However, often the auction prices are significantly lower than the estimated social cost of carbon \citep{epa2023socialcostofcarbon}, reducing the effectiveness of this deterrent and limiting revenue generation. This inefficacy can be attributed to suboptimal auction designs, a consequence of not fully understanding the bidding dynamics and equilibria.

In this paper, we address the issue of learning optimal bidding strategies in repeated multi-unit PAB auctions. As we will elaborate later, we develop efficient no-regret algorithms that simplify the bidding complexity associated with PAB auctions. Through simulating the market dynamics derived from these learning algorithms, we empirically analyze the equilibria of PAB auctions, which have been poorly understood prior to our research. Our empirical findings demonstrate that in the equilibria resulting from these market dynamics, bidders' winning bids converge to nearly the same value, thus addressing concerns regarding price fairness in PAB auctions \citep{TreasuryUnifOrDisc2000, Trilemma2020}.

We also consistently observe high revenue from these equilibria, especially when compared to its uniform price counterpart. In the context of carbon markets, this additional revenue can be invested into clean-up efforts and green technology; see. e.g., phase 4 of the European Union Emissions Trading System (EU-ETS) \citep{gregor2023review}.

\subsection{Technical Contributions.} {\color{black}Our results are presented in Table \ref{table: summary of main results}, offering a comparison to the established results for uniform price auctions by \cite{brânzei2023online}. Below, we discuss our findings on equilibrium characterization for PAB and the details provided in the table.}

\begin{table}[h!]
\centering \footnotesize
\begin{tabular}{|>{\centering\arraybackslash}m{1.4cm}||>{\centering\arraybackslash}m{1.7cm}|>{\centering\arraybackslash}m{1.7cm}|>{\centering\arraybackslash}m{4.5cm}|>{\centering\arraybackslash}m{5.7cm}|}
 \hline
  & Regret (Full) & Regret (Bandit) & Regret Lower Bound  (bandit) & Equilibrium Behavior \\ [0.5ex]
 \hline\hline
 PAB & $\bm{O(MT^{\frac{1}{2}})}$ & $\bm{O(MT^{\frac{2}{3}})}$ & $\bm{\Omega(\max(MT^{\frac{1}{2}}, M^{\frac{2}{3}}T^{\frac{2}{3}}))}$ & High revenue, near welfare optimal, near equal winning bids\\  \hline 
 Uniform & $O(M^\frac{3}{2}T^{\frac{1}{2}})$ & $O(M^{\frac{7}{4}}T^{\frac{3}{4}})$ & $\Omega(MT^{\frac{1}{2}})$ & Low revenue, welfare optimal, staggered winning bids\\ [1ex] 
 \hline
\end{tabular}
\caption{\textbf{Summary of our Main Results.} Our contributions are written in bold. We compare to the uniform price auction market dynamics as described in \cite{brânzei2023online}.}
\label{table: summary of main results}
\end{table}

{\color{black}\textbf{Equilibrium Characterization (Section \ref{sec: equilibria}).} We present  theoretical results characterizing equilibria in PAB, distinguishing between various notions of equilibrium, including Pure Nash Equilibrium (PNE), Coarse Correlated Equilibrium (CCE), and Correlated Equilibrium (CE). We demonstrate that when a PNE exists, the winning bids must be nearly uniform, differing by no more than a discretization factor of \(\delta\). Additionally, we present sufficient conditions under which PNE exists, further reinforcing the aforementioned uniform bidding property. We experimentally show how often these conditions are satisfied as a function of the number of items, the discretization factor, and the number of market participants. However, we show that the uniform bidding property does not necessarily hold for CCEs or CEs via optimizing for the gap between the winning bids using the linear programming characterizations of CCEs and CEs. Nevertheless, our no-regret learning algorithms consistently exhibit the uniform bidding winning bids property (Table \ref{table: learning dynamics full info}), regardless of PNE existence.}

\textbf{New Framework to Study Learning How to Bid in PAB Auctions (Section \ref{sec:repeated}).} Let there be $N$ bidders/agents with $\Nitem$-unit demand in a PAB auction with $\overline{\Nitem}$ supply. Each bidder $n$ is endowed with valuation vector $\bm{v}_n = (v_{n, 1},\ldots,v_{n, \Nitem}) \in [0, 1]^\Nitem$, where $v_{n, i}$ is the marginal value of the $i$-th unit for bidder $n$.
The bidder submits a bid vector $\bm{b}_n = (b_{n, 1},\ldots,b_{n, \Nitem}) \in \mathcal{B}^{\Nitem}$, where $\mathcal{B}$ is some discretization of $[0, 1]$ that represents the set of all possible bids. Agents then receive allocation $x_n \in [\Nitem]$ and  utility $\sum_{\nitem=1}^{x_n} (v_{n, \nitem} - b_{n, \nitem})$ according to the PAB auction rule; see Section \ref{sec:model}. Repeating this auction across $\Nround$ rounds, each agent's goal is to minimize their regret with respect to their hindsight, utility maximizing bid vector.  

\textbf{Dynamic Programming  Scheme for Hindsight Optimal Offline Solution (Section \ref{sec:offline}).} To design low-regret bidding algorithms, we crucially leverage the structure of the hindsight optimal offline solution. In the offline/hindsight problem, the bidder has access to the (historical) dataset of submitted bids by competitors and seeks to find the utility maximizing bid vector on that dataset. (See \cite{RW16, derakhshan2021beating,derakhshan2019lp, golrezaei2021boosted} for works that study similar problems from an auctioneer's perspective.)

We show that the optimal solution to the offline problem---which is our benchmark in computing the regret of our online learning algorithms---can be solved  using a polynomial time Dynamic Programming (DP) scheme.  To do so, we make the following key observation:  
 to win $\nitem$ units (or equivalently, slots), an agent $n$ must have at least $\nitem$ bids larger than the smallest $\nitem$ among the largest $\overline{\Nitem}$ bids of all other bidders. This observation allows us to devise a DP where in each step of the DP, we decide about the bid for one unit, while considering the externality that this bid will impose on the bids and utilities for other units. This externality is precisely the fundamental tradeoff of PAB auctions aforementioned: bidding too small decreases the probability of winning the current or any subsequent units, however, bidding too large increases the payment of the current and previous units. 

\textbf{Decoupled Exponential Weights Algorithm (Section \ref{sec: decoupled exp weights section}).} 
We present our first set of algorithms to learn in the online setting, in both the full information and bandit feedback regimes. We leverage our DP scheme to obtain decoupled rewards, or reward estimates in the bandit setting, for each unit-bid value pair. In particular, we can obtain an exact expression for the utility estimate for bidding $b_\nitem$ for unit $\nitem$ that is independent of $b_1,\ldots,b_{\nitem-1}$ or $b_{\nitem+1},\ldots,b_\Nitem$, subject to bid vector monotonicity. This allows us to mimic the exponential weights algorithm on the exponentially large bid space in polynomial time and space complexity. We show that our decoupled exponential weights algorithm (Algorithm~\ref{alg: Decoupled Exponential Weights}), which achieves polynomial time and space complexities, attains \(O(M^{\frac{3}{2}} \sqrt{T \log{T}})\) regret in the full information setting and \(O(M^{\frac{4}{3}} T^{\frac{2}{3}} \sqrt{\log{T}})\) regret in the bandit settings, respectively. {\color{black}For an extension of these algorithms to a setting with time-varying stochastic valuations, please refer to Section~\ref{sec: time varying}.}

\textbf{Online Mirror Descent Algorithm (Section \ref{sec:bandit}).}  
In this section, we introduce an alternative learning algorithm based on Online Mirror Descent (OMD) that enhances the regret upper bounds of our decoupled exponential weight algorithm by a factor of \(\sqrt{M}\), albeit with increased computational demands. Utilizing our DP scheme, specifically the graph it induces (see Section \ref{sec:offline}), we maintain probability measures over the DP graph nodes rather than its edges. This modification, informed by the observation that bid vector utility depends solely on the graph's nodes, results in more efficient regret bounds.

{\color{black}\textbf{Regret Lower Bound (Section~\ref{sec: lower bound}).} To complement our regret upper bound, we construct regret lower bounds for both full information and bandit settings: $\Omega(M\sqrt{T})$ and $\Omega(\max{M^\frac{2}{3} T^{\frac{2}{3}}, M\sqrt{T}})$. To construct the regret lower bound of $\Omega(M\sqrt{T})$, which is a valid lower bound for both settings, we construct two distributions over adversary bid vectors for which any learning agent is guaranteed to incur regret linear in $M$ when trying to learn the optimal bid under these distributions. Observe that the regret of our OMD algorithm matches the dependence of the regret lower bound on $M$. 

For the bandit setting, we present another regret lower bound of $\Omega(M^\frac{2}{3} T^{\frac{2}{3}})$, which matches the regret of both Decoupled Exponential Weights and OMD Algorithms in terms of dependence on $T$. To construct this lower bound, we build upon the well-known $\Omega(T^{\frac{2}{3}})$ regret lower bound for the posted price and first price auctions  in \cite{DemandCurve2003, ContextBanditsCrossLearning2019}. In particular, we consider a distribution over adversary bid vectors such that the marginal distribution of the adversary's bid for any unit is the equal revenue distribution. We then construct a family of hypotheses by perturbing each of these marginals at a different bid value. We make a key simplification that the set of possible perturbations between two units must be disjoint, limiting the impact of cross-learning between units and allowing for stronger dependence in $M$. We show that any learning agent must incur regret at least $\Omega(M^\frac{2}{3} T^{\frac{2}{3}})$ when distinguishing among all possible combinations of perturbations.}

\subsection{Experimental Results and Managerial Implications}\label{sec:insights}

Our experiments yield valuable practical insights for both auction designers and participants. These insights are primarily derived from simulations of PAB market dynamics using the no-regret learning algorithms designed in our paper. Additionally, we compare these results with the market dynamics of uniform price auctions using the algorithms presented  in \cite{brânzei2023online}. It is important to emphasize that conducting such systematic comparisons was previously challenging due to the inherent difficulty of characterizing equilibria in these auctions prior to our research.

\begin{enumerate}
    \item \textbf{Near-Uniform Bidding in PAB Auctions.} {As shown in Table~\ref{table: learning dynamics full info}, the market dynamics consistently lead to the convergence of the winning bids and the largest losing bids to nearly identical values across all bidders.} This finding partially addresses one of the main concerns regarding the fairness of the PAB auction. Specifically, although the payment for each unit \textit{can} differ across units and bidders, under a reasonable learning and bidding strategy, these payments tend to converge to similar values over the long term.
    \item \textbf{Simplified Bidding Interface is Sufficient for PAB but not for Uniform Price.} In a recent trend of bidding simplification  and automation (e.g.,  \cite{aggarwal2019autobidding, deng2023multi, susan2023multi, lucier2023autobidders}), auctioneers may find it easier to restrict bidders' demand expressiveness by requiring only a single price and quantity, rather than a vector of bids. As per our previous insight, the bid value convergence of the market dynamics suggest appropriateness of this simplified bidding interface. In contrast, we show that the market dynamics of the uniform price auction converge to a staggered bid vector (Table \ref{table: learning dynamics full uniform}), suggesting that the simplified bidding interface may significantly damage the uniform price auction's welfare, revenue, or bidders' utility.
    \item \textbf{PAB Obtains High Revenue but Slightly Lower Welfare than Uniform Price.} From the insights provided by Tables \ref{table: learning dynamics full info} and \ref{table: learning dynamics full uniform}, it is evident that the PAB auction surpasses the uniform price auction in terms of revenue generation. However, it slightly lags behind in welfare, though to a lesser extent. Consequently, auctioneers who prioritize revenue (resp. welfare) should favor the PAB (resp. uniform price) auction over uniform price (resp. PAB) auctions.
\end{enumerate}

\subsection{Other Related Works}

\textbf{Learning in Auctions.} Most of the recent learning-theory-flavored auction design research has focused on the perspective of the auctioneer setting reserve prices \citep{MorgensternR16, mohri2016learning, CaiD17, DudikHLSSV17,kanoria2014dynamic, DemandCurve2003, braverman2018noregretbuyer, golrezaei2019IC, golrezaei2018dynamic, golrezaei2021bidding},
or uniform price auctions \citep{LearningRevOptSPA2013, mohri2016learning, OptReserveMyopic2018, LearningToBidRevenueMaximizing2019, OSPABidding2020}. 
{\color{black} Relatively few works study learning from the bidder's perspective or the resulting market dynamics. The works most closely aligned with our own within this literature \citep{LearningBidOptimallyAdversarialFPA2020, OptimalNoRegretFPA2020, ContextBanditsCrossLearning2019} establish low-regret algorithms for learning to bid in first-price auctions. As the first-price auction is precisely the PAB auction with $M=1$ units, we precisely recover the regret upper bound of $O(T^{\frac{2}{3}})$ established for the bandit feedback setting as described in \cite{ContextBanditsCrossLearning2019}, which is also shown to be tight. The superior rate of $O(T^{\frac{1}{2}})$ achieved in \cite{OptimalNoRegretFPA2020} is incomparable to our bandit feedback $O(T^{\frac{2}{3}})$ regret as they assume a stronger feedback structure where the winning bid is announced at the end of every round. Note that under this feedback structure, unlike the bandit setting, bidders who are not winning the item learn about their competitors. Similarly, the $O(T^{\frac{1}{2}})$ rate in \cite{LearningBidOptimallyAdversarialFPA2020} is obtained in the full-information setting with time-varying valuations. While we obtain matching regret rates (see Section~\ref{sec: time varying}), our result assumes stochastic valuations and compares to an arbitrary valuation-to-bid mapping, whereas their result assumes adversarial valuations and compares to any 1-Lipschitz valuation-to-bid mapping.

Generalizing the aforementioned learning algorithms to the multi-unit setting, however, is difficult as the space of possible bid vectors is exponentially large. To combat this, we take from the expansive structured bandit literature, which includes linear bandits \citep{StochLinearOpt2008, LinStochBandits2011, ContextLinearPayoff2011, Lattimore2020}, combinatorial bandits \citep{CMAB2013, MinimaxCombinatorial2011, niazadeh2021online}, and convex uncertainty set bandits \citep{van2020optimal}. In particular, our algorithm most closely resembles existing algorithms exploiting both these combinatorial and linear aspects in episodic Markov Decision Processes \citep{OREPS2013} or cost minimization on graphs \citep{PathKernel2003}.
In this paper, we contribute to this line of work by proposing a novel regret minimization framework for PAB auctions. Moreover, we numerically analyze the corresponding multi-agent learning dynamics (also for the uniform price auction), which complements our equilibrium analysis in Section \ref{sec: equilibria}.
}

\textbf{PAB Mechanism.} There are several multi-unit auction formats that are commonly used in practice; e.g., uniform price \citep{TreasuryUnifOrDisc2000, LastAcceptedBid2020, ImprovedRevenuePPASPA2021, Kira2019}, PAB \citep{Homogeneous2020, FPACollusion2000, LargeMultiUnit2018}, Vickrey-Clarke-Groves (VCG) \citep{MechanismsMultiUnit2007, LonelyVickrey2006}, ascending price \citep{AscendingCramton1998, EfficientAscending2004}. The literature is divided as to which auction is appropriate for various settings. For example, while the PAB mechanism has desirable revenue and welfare guarantees compared to the uniform price auction \citep{Homogeneous2020}, the empirical revenue of the two auctions is often comparable \citep{Turkish2010}, and some argue that guess-the-clearing-price and other strategic behavior \citep{GermaryReserve2021} along with collusion \citep{WholesaleElectricityMarkets2008} can further damage its performance. Furthermore, there are ethical and fairness concerns regarding PAB auctions, as their discriminatory nature implies that agents pay unequally for the same unit. Despite this criticism, and other arguments for (and against) other auction formats \citep{CombinatorialAuctionDesign2003, CombinatorialAuctions2004, InefficiencyStandard2013, DemandReduction2014, Trilemma2020}, we focus on the PAB mechanism  due to the simplicity and transparency of its payment rule, as well as, its widespread use.

Regarding the equilibria, bidding dynamics, efficiency, and other key properties of multi-unit PAB mechanisms, there are some partial results in the economics literature.  For example,
the Bayesian optimal bidding strategy is known for the case of 2-unit demand and supply multi-unit auctions \citep{OptimalBidding1995}, for when valuations follow a class of parametric distributions \citep{Homogeneous2020}, or when the bidders have symmetric valuations \citep{DemandReduction2014}.  The PAB mechanism is also known to be smooth \citep{syrgkanis2012composable}, which yields a number of desirable guarantees on the price of anarchy of the auction \citep{roughgarden2015smooth}, even with the presence of an aftermarket \citep{babaioff2022aftermarkets}. An additional attractive property of the PAB mechanism is complete transparency of payments, as given one's allocation, an agent knows precisely how much they will pay. This is in stark contrast to the uniform price auction, where shill bids can inflate payments by artificially increasing demand \citep{Trilemma2020}.

{\color{black}In this work, we contribute to this line of research by demonstrating that while in PNEs the winning bids are uniform (up to a discretization factor), this is not necessarily the case under CCEs and  CEs. While our work is not the first to show optimality of uniform bidding or characterize an approximately uniform bidding PNE in PAB \citep{inefficiency2013}, to our knowledge, it is the first that does so with the constraint of no marginal overbidding---bidder's cannot submit bids for units larger than their marginal value for that unit.}

\textbf{Relationship to Uniform Price Auctions.} The uniform price auction is an alternative mechanism of allocating multiple homogeneous goods. Closely related to the PAB auction, bidders are allocated units in decreasing order of bids, but instead of charging each bidder their corresponding winning bid, each bidder instead pays the smallest winning bid. The  EU-ETS carbon license auctions use the uniform price auction, rather than the PAB auction, largely due to fairness considerations, as each agent pays the same amount per unit allocated \citep{EUETS}. However, our work shows that price fairness is not a concern in the long term when bidders learn how to bid and converge to a common price. In a study closely aligned with our research, \cite{brânzei2023online} investigated the problem of learning optimal bidding strategies in multi-unit uniform pricing auctions. Since the uniform price auction employs a distinct payment rule, the bid optimization problem necessitates different approaches compared to the PAB auction. We show more specifically that the payment structure of PAB allows for additional simplifications in the learning task, allowing faster regret rates compared to uniform price.

\section{Preliminaries}\label{sec:model}

\textbf{Notation.} We let $|\mathcal{S}|$ denote the size of set $\mathcal{S}$, and define $[k] = \{1,\ldots,k\}$ to be the set of the first $k$ positive integers. We define $\mathcal{S}^{+k}$ to be the set of non-increasing $k$-vectors of elements from set $\mathcal{S}$. Similarly, $\mathcal{S}^{-k}$ denotes the set of non-decreasing $k$-vectors of elements from set $\mathcal{S}$. We let $\Delta(\mathcal{S})$ denote the set of valid probability measures over set $\mathcal{S}$. We also say that the quantity $x$ is $\lesssim$, $\propto$, or $\gtrsim$ than $f(a_1,\ldots,a_k)$ some function of $k$ algorithm parameters if $x \in \mathcal{O}(f(a_1,\ldots,a_k))$, $x \in \Theta(f(a_1,\ldots,a_k))$, or $x \in \Omega(f(a_1,\ldots,a_k))$,  respectively. We use $\tilde{O}$, $\tilde{\Theta}$, and $\tilde{\Omega}$ to hide logarithmic factors.
Lastly, given a set of elements $a_1,\ldots,a_k$, we define $a_{s:r} = (a_s,\ldots,a_r)$ for $r \leq k$ to be the set of elements from $s$ to $r$ inclusive.

\textbf{Auction format: Pay-as-bid.} 
 Consider a scenario in which there are $N$ bidders and $\overline{M}$ identical units available for auction in a PAB  format. In this context, we assume that each agent $n$ within the set $[N]$ desires a maximum of $M$ units. (It should be noted that our results can be extended to a situation where each agent $n$ may demand a different maximum number of units, denoted by $M_n$, which are not necessarily identical.)
 
 Let $\bm{v}_n = (v_{n, \nitem})_{\nitem \in [M]} \in [0, 1]^{+M}$ represent agent $n$'s non-increasing marginal valuation profile. This implies that for any given $n$ in the set $[N]$, the following conditions hold: (i) {valuation monotonicity} $v_{n,1}\ge v_{n, 2}\ge \ldots \ge v_{n, \Nitem}$ and (ii) the total valuation of agent $n$ after receiving $\nitem$ units is given by $\sum_{k=1}^{m}v_{n, k}$.

 Let $\bm{b}_n = \{b_{n, \nitem}\}_{\nitem \in [\Nitem]} \in [0, 1]^{+\Nitem}$  represent the non-increasing bids submitted by bidder $n$. Here, $b_{n, \nitem}$ refers to the bid made by bidder $n$ for the $\nitem$-th slot/unit. Similar to $\bm{v}_n$, we have the following conditions for any $n \in [N]$: (i) {bid monotonicity}
$b_{n,1}\ge b_{n, 2}\ge \ldots \ge b_{n, \Nitem}$ and  (ii) individual rationality (IR) $b_{n, \nitem} \leq v_{n, \nitem}$ for all $m \in [M]$. 
It is important to note that the bid monotonicity condition is not an assumptions; it is implied by the auction rule that will be stated shortly. Consequently, the total payment made by bidder $n$ after receiving $\nitem$ units is given by $\sum_{k=1}^{m} b_{n, k}$. We define $\bm{b}_{-n} = (b_{-n,\nitem})_{\nitem \in [\overline{\Nitem}]} \in [0, 1]^{-\overline{M}}$  as the set of the $\overline{\Nitem}$ largest bids not belonging to agent $n$, arranged in increasing order (i.e., $b_{-n, 1} \leq \ldots \leq b_{-n, \overline{\Nitem}}$).

 The auction operates according to the following rules:
 In a PAB auction, all bids submitted across the $N$ bidders are arranged in descending order. The $\nitem$-th unit is assigned to the bidder with the $\nitem$-th highest bid, and they are charged the amount of their bid. We denote the allocation to agent $n$ as $x_n(\bm{b}_n) = x(\bm{b}_{n}, \bm{b}_{-n})$, and the (quasi-linear) utility as $\mu_n(\bm{b}_n) = \mu(\bm{b}_{n}, \bm{b}_{-n})$, where
\begin{align}
    \label{eq: allocation and utility definition}
    x(\bm{b}_{n}, \bm{b}_{-n}) = \sum_{\nitem=1}^\Nitem \textbf{1}_{b_{n, \nitem} \geq b_{-n, \nitem}} \quad \text{and} \quad \mu(\bm{b}_{n}, \bm{b}_{-n}) = \sum_{\nitem=1}^{x_n(\bm{b}_n)} (v_{n, \nitem} - b_{n, \nitem})\,.
\end{align}
respectively. Here, $b_{-n, \nitem}$ represents the $\nitem$-th smallest bid among the $\overline{M}$ largest bids of all other bidders except bidder $n$.
It should be noted that $x(\bm{b}_{n}, \bm{b}_{-n})$ denotes the number of units that agent $n$ receives in the auction. We assume that ties are broken in favor of higher indexed bidders and that this is public knowledge.

{\color{black}In the following sections, we study PAB in two settings. We first consider the one-shot setting, wherein we characterize properties of PNEs, CCEs, and CEs. Then, we define and discuss the repeated setting and derive corresponding no-regret learning algorithms.}

{\color{black}\section{Equilibria of Pay-as-Bid Auctions}
\label{sec: equilibria}

We provide partial results regarding the properties and structure of common solution concepts such as PNE and CCE. More specifically, we show near-uniformity of the winning bids and each bidder's submitted bids under PNEs and we show these properties do not extend to CCEs, or even CEs. We first define these equilibria  formally:

\begin{definition}[Pure Nash Equilibria]
    In the PAB setting, a Pure Nash Equilibrium (PNE) is defined as a collection of $N$ bid vectors $\bm{b}_n \in \mathcal{B}^{+M}$ for all $n \in [N]$ that satisfies
    \begin{align*}
        \mu_i(\bm{b}_i, \bm{b}_{-i}) \geq \mu_i(\bm{b}'_i, \bm{b}_{-i})\,.
    \end{align*}
\end{definition}

PNE can be interpreted as a collection of \textit{independent} strategies under which no individual has an incentive to deviate. While PNEs are the most commonly studied and desirable solution concept in much of the economics and computer science literature, they are often intractable computationally and, in the case of PNE, may not necessarily even exist \citep{PredictionLearningGames}. As such, we will study a broader class of solution concepts known as CCEs, which contains the set of PNEs. These solution concepts generalize PNEs by allowing for dependencies between bidders' strategies. To better understand this, we first formally define CCEs, which are well known to comprise the limit cycles of no-regret learning algorithms, and CEs, which are a class of solutions between PNEs and CCEs that comprise the time-averaged joint bid distributions limit cycles of no-swap-regret learning \citep{blum2007externalinternalregret}.

\begin{definition}[Coarse correlated equilibria and correlated equilibria]
    In the PAB setting, a coarse correlated equilibrium (CCE) is defined as a joint probability distribution $p: \bigotimes_{n=1}^N \mathcal{B}^{+M} \to [0, 1]$ over all $N$ agents' bid vectors $\bm{b}_{1:N}$ that satisfies the following set of linear constraints for all $n \in [N]$:
    \begin{align*}
       \mathbb{E}_{\bm{b}_{1:N}\sim p}[  \mu_n(\bm{b}_n; \bm{b}_{-n})] \geq \mathbb{E}_{\bm{b}_{1:N}\sim p}  [\mu_n(\bm{b}'_n; \bm{b}_{-n})] \quad \text{ for any } \bm{b}'_n \in \mathcal{B}^{+M}\,.
    \end{align*}
    A correlated equilibrium (CE) is similarly defined as a joint probability distribution $p: \bigotimes_{n=1}^N \mathcal{B}^{+M} \to [0, 1]$ over all $N$ agents' bid vectors $\bm{b}_{1:N}$ that satisfies the more stringent set of linear constraints for all $n \in [N]$:
   \begin{align*}
        \mathbb{E}_{\bm{b}_{1:N}\sim p}[  \mu_n(\bm{b}_n; \bm{b}_{-n}) | \bm{b}_n] \geq    \mathbb{E}_{\bm{b}_{1:N}\sim p}[\mu_n(\bm{b}'_n; \bm{b}_{-n})| \bm{b}_n] \quad \text{ for any }  \bm{b}_n, \bm{b}'_n \in \mathcal{B}^{+M}\,.
    \end{align*}
    
\end{definition}

CCEs (resp. CEs) can be interpreted as a joint distribution $p(\cdot)$ over all strategy profiles such that it is ex-ante (resp. ex-post) optimal for bidders to adhere to the strategy prescribed by a ``signal" drawn from this distribution. Beyond the well known fact that the class of CCEs (resp. CEs) comprise the limit cycle of no external (resp. internal) regret learning dynamics in repeated games \cite{PredictionLearningGames}, CCEs and CEs are of major interest they possess strong welfare guarantees in smooth games, such as PAB auctions \citep{syrgkanis2012composable, InefficiencyStandard2013, roughgarden2015smooth, PriceOfAnarchyInAuctions2017, feldman2017correlated}. 
Moreover, the linear constraint representation of CCEs and CEs is particularly useful as one may define a linear program with a corresponding linear objective function in order to test certain hypotheses; as we do in the proof of Lemma \ref{lem: CCE non uniform winning}.
Having formally defined PNEs, CCEs, and CEs, we are ready to state our main results regarding the equilibria of PAB, assuming an even discretization $\mathcal{B} = \{0, \delta,\ldots, 1-\delta, 1\}$.

\subsection{Pure Nash Equilibria}

We begin by showing that PNEs, if they exist, must have nearly uniform winning bids; at most one discretization factor $\delta$ (which is equal to $\frac{1}{|\mathcal{B}|}$ under an even discretization $\mathcal{B}$) apart.

\begin{lemma}
    \label{lem: PNE uniform bidding} Any PNE (if one exists) requires that the highest bid $b_{(1)}$ is at most \(\delta\) above the \(M\)'th largest bid, denoted by \(b_{(M)}\). Here, $\delta$ is the discretization factor.
\end{lemma}

Lemma \ref{lem: PNE uniform bidding}, detailed in Section \ref{sec:proof:lem: PNE uniform bidding}, establishes that any PNE (if existing) necessitates nearly uniform winning bids, differing by at most a discretization factor of $\delta$. Additionally, the subsequent lemma, proved in Section \ref{sec:proof:lem: near-uniform optimal bidding}, illustrates that uniform bidding is optimal for any bidder.

\begin{lemma}
    \label{lem: near-uniform optimal bidding} 
    Consider any fixed bids \(\bm{b}_{-n}\) with tie-breaking favoring higher-indexed bidders. Let $\tilde{b}_{n,m}$ denote the minimum bid required for bidder \(n\) to win the \(m\)'th item under the tie-breaking rule. Specifically, \(\tilde{b}_{n,m} = b_{-n,m} + \delta\) if \(b_{-n,m}\) belongs to a bidder with a higher index than \(n\), and \(\tilde{b}_{n,m} = b_{-n,m} \) otherwise. Assuming no overbidding, bidder \(n\)'s optimal bids takes the following form:  \(b_j= \tilde{b}_{n,m}\) for all \(j\in [m]\), and \(b_j \leq v_j\) for all \(j>m\) for all $m$ such that $\tilde{b}_{n,m} < v_m$, where for convenience, $\tilde{b}_{n,0} = 0$. Consequently, under any PNE, if one exists, bidder \(\bm{b}_n\) should take the aforementioned form. 
\end{lemma}

This result states that under any fixed $\bm{b}_{-n}$, the optimal bid vector for bidder $n$, under which they are allocated $m$ items, is such that they pay the minimum possible price required to win those $m$ items.
We now present our main result of this section, where we provide sufficient conditions to fully characterize an approximately welfare-efficient PNE.

\begin{theorem}[Existence of an Approximately Efficient PNE]
    \label{thm: PNE existence}
     Define the clearing price $c = \lfloor v_{(M)} \rfloor_{\delta}$ to be the $M$'th largest valuation among all bidders, denoted by $v_{(M)}$,  rounded down to the nearest multiple of $\delta$, and similarly, define $c_{-n} = \lfloor v_{(-n, M)} \rfloor_{\delta}$ to be the rounded $M$'th largest valuation among all bidders except $n$. If $c = c_{-n}$ for all $n\in [N]$ and ties are broken in favor of higher indexed bidders, then there exists a PNE $(\bm{b}_1,\ldots,\bm{b}_N)$ where each bidder $n\in [N]$:  
    \begin{enumerate}
        \item submits bids of either all $c$ or all $c + \delta$ for units such that $v_{n,m} \geq c + \delta$,
        \item submit bids of $c$ for all units such that $v_{n,m} \in [c, c + \delta)$,
        \item submit bids smaller than $c$ for all other units.
    \end{enumerate}
    Moreover, this PNE is $M\delta$-approximately welfare optimal:
    \begin{align*}
        \sum_{m=1}^M v_{(m)} - \sum_{n=1}^N \sum_{m=1}^{x(\bm{b}_n, \bm{b}_{-n})} v_{n,m} \leq M\delta\,.
    \end{align*}
\end{theorem}
{\color{black}At a high level, the $c = c_{-n}$ constraint implies that bidders act as price-takers and cannot increase their utility by under-bidding for the units allocated to them under the PNE condition. In this competitive environment, there exists an approximately welfare-optimal PNE where all bidders submit bids close to the clearing price $c$, which represents the $M$'th largest valuation.

For a detailed proof, we refer the reader to the online appendix (arXiv:2307.15193v3). 
In Section \ref{sec:discuss:assumption}, we include discussion of this assumption as well as present an example where this assumption holds, but even a minor alteration leads to its breakdown. As demonstrated in Figure~\ref{fig: bid cycling main body}, this impacts bid convergence in our algorithm, resulting in cyclic bidding behavior for certain bidders. However, it's crucial to note that while the $c = c_{-n}$ condition suffices for existence of PNEs, PNE existence is broader, dependent on bidders acting as price-takers with minimal market influence. This holds true in markets where individual demands are much smaller than the total supply and the total demand exceeds the total supply.}

\subsection{Coarse Correlated and Correlated Equilibria}
Thus far, we have showed some structural results for PNEs in PAB auctions. Unfortunately, this structure does not extend to CCEs, or even the tighter class of CEs. 
\begin{lemma}
    \label{lem: CCE non uniform winning}
    There are CCEs (and CEs) of PAB where winning bids differ by more than one discretization factor, 
$\delta$, with a non-zero probability.
\end{lemma}

Because of this non-winning bid uniformity in CCEs, we cannot hope to recover such an approximate welfare optimality result as we did for PNE. Not all hope is lost, however, as we can still apply a smoothness-based argument to lower bound the expected welfare under any CCE of PAB \cite{syrgkanis2012composable, roughgarden2015smooth}. This result shows that in PAB auctions, the expected welfare under any CCE is at least $\frac{1}{2}(1 - \frac{1}{e})$ of the optimal welfare. Perhaps more importantly, despite the existence of CCEs with non-uniform winning bids, our no-regret learning algorithms consistently converge to CCEs which possess such uniformity (see Table~\ref{table: learning dynamics full info}). We leave such characterization of these CCEs as the result of our learning algorithms as future work.

}

\section{Learning to Bid in a Repeated Setting}
\label{sec:repeated}

In this section, we explore the repeated PAB setting where bidders participate over $T$ rounds, aiming to maximize their aggregate utility or minimize their regret. A limitation of the equilibria analysis from the previous section is its ambiguity on the attainability of these equilibria through natural learning dynamics like fictitious play or no-regret learning. For instance, while Lemmas~\ref{lem: near-uniform optimal bidding} and \ref{lem: PNE uniform bidding} depend on the existence of a PNE, which may not be present, bidders' learning behavior might still converge towards a specific subset of CCEs with certain  structures.\footnote{For example, when bidders behave according to a certain class of no-regret learning algorithms known as \textit{mean-based} algorithms, they may even \textit{last-iterate} converge to Mixed NE or even PNE of certain kinds of auctions \cite{braverman2018noregretbuyer, Deng_2022}.}
We respond to these insights by designing no-regret learning algorithms and empirically examining the dynamics and equilibria  they converge towards.

\subsection{Repeated Setting}

We consider a repeated setting where the PAB auction is conducted over $\Nround$ rounds. In this repeated setting, we will focus on the perspective of agent $n \in [N]$ and remove additional indexing when it is evident from the context. For each auction round, agent $n$ has a fixed valuation profile represented by 
$\bm{v} = (v)_{\nitem \in [\Nitem]} \in [0, 1]^{+\Nitem}$, and their
 bid vector in the $\nround$-th auction is denoted by $\bm{b}^\nround = (b^\nround_\nitem)_{\nitem \in [\Nitem]} \in [0, 1]^{+\Nitem}$.\footnote{In Section \ref{sec: time varying}, we extend our results to the case of time-varying valuations.} 
Similarly, $\bm{b}^\nround_{-} = (b^\nround_{-m})_{m \in [\overline{\Nitem}]} \in [0, 1]^{-\overline{\Nitem}}$ represents the competing bids in round $\nround$. In each round, agent $n$ receives $x^\nround_n(\bm{b}^\nround)$ units and earns a utility of $\mu^\nround_n(\bm{b}^\nround)$, where:
\begin{align}
    \label{eq: allocation and utility definition repeated}
    x^\nround_n(\bm{b}^\nround) = x(\bm{b}^\nround, \bm{b}^\nround_{-}) \quad \text{and} \quad \mu^\nround_n(\bm{b}^\nround) = \mu(\bm{b}^\nround, \bm{b}^\nround_{-})\,.
\end{align}
Recall that functions $x$ and $\mu$ are defined in Equation \eqref{eq: allocation and utility definition}. 
The goal of agent $n$ is to choose a sequence of bid vectors $(\bm{b}^\nround)_{\nround \in [\Nround]}$ that maximizes their total utility, given by $\sum_{\nround=1}^\Nround \mu^\nround_n(\bm{b}^\nround)$. However, the main challenge is that the vectors $\bm{b}^{\nround}_{-}$, representing the competing bids, are not known in advance. Instead, they are revealed in an online manner. Consequently, agents must learn how to bid optimally throughout the sequence of auctions, taking into account their previous allocations $H^{\nround-1} = (x^\tau(\bm{b}^\tau))_{\tau \in [\nround-1]}$ and their valuation profile $\bm{v}$. The performance of an agent's learning strategy is evaluated in terms of regret, which quantifies the difference between their expected utility using their learning strategy and the optimal utility achievable with perfect knowledge of the competing  bidding vectors in hindsight:
\begin{align}
\tag{Continuous Regret}
    \textsc{Regret} = \max_{\bm{b} \in [0, 1]^{+\Nitem}} \sum_{\nround=1}^\Nround \mu^\nround_n(\bm{b}) - \mathbb{E}\left[\sum_{\nround=1}^\Nround \mu^\nround_n(\bm{b}^\nround)\right]\,.
\end{align}
Here, $\bm{b}^{\nround}_-$ can be selected by an adaptive adversary; i.e. an adversary who can select $\bm{b}^\nround_-$ as a function of the entire auction history, which includes $\bm{b}^1,\ldots,\bm{b}^{\nround-1}$, but does not have access to the possible randomness when selecting $\bm{b}^\nround$.

{\color{black}The benchmark, $\max_{\bm{b} \in [0, 1]^{\Nitem}} \sum_{\nround=1}^\Nround \mu^\nround_n(\bm{b})$ used in the definition of continuous regret, is constructed considering all possible $\bm{b}^{\nround}_{-}$ for every round $\nround \in [\Nround]$, where in the hindsight optimal solution, the bid vector can be chosen from any vector $\bm{b} \in [0, 1]^{+\Nitem}$. However, as is common in the literature (see, e.g., \citep{LearningBidOptimallyAdversarialFPA2020, OptimalNoRegretFPA2020, ContextBanditsCrossLearning2019}) for learning to bid in first-price auctions (PAB auctions with $M=1$), characterizing the continuous regret as defined above requires discretizing the action space while considering discretization errors and characterizing the discrete regret, which we will define shortly, in order to bound the continuous regret. The level of discretization is then optimized to balance the discrete regret and discretization errors.

In light of this, let us confine ourselves to a discretization of $[0, 1]$, denoted by $\mathcal{B} = (B_1,\ldots,B_{|\mathcal{B}|})$, where $0 = B_1 < \ldots < B_{|\mathcal{B}|} = 1$. In such instances, we define an analogous version of regret:
\begin{align}
\tag{Discretized Regret}
    \textsc{Regret}_\mathcal{B} = \max_{\bm{b} \in \mathcal{B}^{+\Nitem}} \sum_{\nround=1}^\Nround \mu^\nround_n(\bm{b}) - \mathbb{E}\left[\sum_{\nround=1}^\Nround \mu^\nround_n(\bm{b}^\nround)\right]\,.
\end{align}}
In both  definitions of continuous and discretized regret, we henceforth implicitly assume that $b_m\le v_m$ for any $m\in [M]$; that is, we have  the bid vector $\bm{b}$ is subject to individual rationality and overbidding is not allowed. 
{\color{black}We  wish to derive a learning algorithm that achieves a  continuous regret  (i.e., $\textsc{Regret}$) that is polynomial  in $\Nitem$ and sub-linear in $\Nround$. To do so,  as stated above,  we consider the discretized setting, and bound $\textsc{Regret}_\mathcal{B}$; an upper bound on the continuous regret  will be obtained by accounting for the discretization errors, which we show in Section~\ref{sec: path kernels regret} is of order $O(\frac{MT}{|\mathcal{B}|})$.}

\subsection{Offline Setting} \label{sec:offline}

In this section, we provide an algorithm that efficiently computes the hindsight optimal bid vector $\max_{\bm{b} \in \mathcal{B}^{+\Nitem}} \sum_{\nround=1}^\Nround \mu^\nround_n(\bm{b})$ which is precisely our benchmark for regret. In order to compute agent $n$'s optimal fixed bidding strategy for the $\Nround$ rounds of PAB auctions when fixing all other bids $\bm{b}_{-n}^1,\ldots,\bm{b}_{-n}^T$, we design an algorithm based on path weight optimization. Recall the following optimization problem with bid space $\mathcal{B}$:
\begin{align}
\tag{Offline}
\label{eq:offline}
\max_{\bm{b} \in \mathcal{B}^{+\Nitem}} \sum_{\nround=1}^\Nround \mu^\nround_n(\bm{b})\,,
\end{align}
where $\mu^\nround_n(\bm{b})$ (defined in Equation \eqref{eq: allocation and utility definition repeated}) represents the utility of agent $n$ in round $\nround$ given bid vector $\bm{b}$ and competing bids $\bm{b}_-^t$. As mentioned earlier, the solution to this optimization problem serves as a benchmark for evaluating the performance of online learning algorithms in the repeated setting. Furthermore, it provides valuable insights for designing algorithms with polynomial time and space complexity for the repeated setting.

To solve  Problem \eqref{eq:offline}, we take advantage of the following decomposition: 
\begin{align}
    \label{eq: decomposition}
    \sum_{\nround=1}^\Nround \mu^\nround_n(\bm{b}) &= \sum_{\nround=1}^\Nround \sum_{\nitem=1}^\Nitem (v_\nitem - b_\nitem)\textbf{1}_{b_\nitem \geq b_{-\nitem}^t} = \sum_{\nitem=1}^\Nitem \sum_{\nround=1}^\Nround (v_\nitem - b_\nitem)\textbf{1}_{b_\nitem \geq b_{-\nitem}^t}\\ &:= \sum_{\nitem=1}^\Nitem \sum_{\nround=1}^\Nround w_\nitem^\nround(b_\nitem) := \sum_{\nitem=1}^\Nitem W^{\Nround+1}_\nitem(b_\nitem)\,,
    \label{def: def mu w W}
\end{align}
where 
$w^\nround_\nitem(b) = \textbf{1}_{b \geq b^{\nround}_{-\nitem}} (v_\nitem - b)$ represents the utility in the $\nround$-th auction for winning the $\nitem$-th item with bid $b$, and $W^{\Nround+1}_\nitem(b) = \sum_{\nround=1}^\Nround w^{\nround}_\nitem(b)$ represents the cumulative utility gained from winning the $\nitem$-th item with bid $b$ across the $\Nround$ auctions. \footnote{Here, in $w^\nround_\nitem(b)$, the same tie-breaking rule in Equation \eqref{eq: allocation and utility definition} is applied.} To solve Problem \eqref{eq:offline}, we develop a polynomial-time DP scheme utilizing these $w^\nround_\nitem(b)$ and $W^{\Nround+1}_\nitem(b)$. In particular, for any $\nitem \in [\Nitem]$ and any bid $b \in \mathcal B$, let $U_\nitem(b)$ be the optimal cumulative utility of the agents from units $\nitem, \nitem+1, \ldots, \Nitem$ over $\Nround$ auctions assuming that bids for unit $\nitem$ is less than or equal to $b$. We then have
\begin{align}
    U_\nitem(b) = \max_{b'\le b, b'\in \mathcal B}\left\{ W_{\nitem}^{\Nround+1} (b') + U_{\nitem+1}(b')\right\} \quad b\in \mathcal B, m\in [M] \quad \text{and} \quad U_{\Nitem+1}(b) = 0 \quad b \in \mathcal{B}\,.
\end{align}
Algorithm \ref{alg: Offline Full} uses the aforementioned  DP scheme to devise an optimal solution to Problem \eqref{eq:offline}. The following theorem, proven in Section~\ref{sec: offline proof}, shows the optimality of Algorithm \ref{alg: Offline Full}. 

\begin{theorem}\label{thm:offline}
Algorithm~\ref{alg: Offline Full} returns the optimal solution to Problem \eqref{eq:offline} {with time and space complexity of $O(\Nitem |\mathcal{B}|^2)$ and $O(\Nitem|\mathcal{B}|)$, respectively.}
\end{theorem}

\begin{algorithm}[t]
\footnotesize
	\KwIn{Valuation $\bm{v}$ for $\bm{v} \in [0, 1]^{+\Nitem}$, Other bids $\{\bm{b}^{\nround}_-\}_{\nround \in [\Nround]}$ for $\bm{b}^{\nround}_- \in \mathcal{B}^{-\overline{\Nitem}}$.}
	\KwOut{Optimal bid vector $\bm{b}^* = \text{argmax}_{\bm{b} \in \mathcal{B}^\Nitem} \mu^{ \Nround}_n(\bm{b})$ and its corresponding utility.}
	Let $W^{\Nround+1}_\nitem(b) \gets \sum_{\nround=1}^\Nround \textbf{1}_{b \geq b^{\nround}_{-\nitem}} (v_\nitem - b)$,  $b \in \mathcal{B}, \nitem \in [\Nitem]$,   define 
        $U_{\Nitem+1}(b) \gets 0$, $b \in \mathcal{B}$, and set $b^*_{0} = \max(\mathcal B)$\;
        \textbf{for} $m \in [M,\ldots,1], b \in \mathcal{B}: U_\nitem(b) \gets \text{max}_{b' \in \mathcal{B}; b' \leq b}W^{\Nround+1}_\nitem(b') + U_{\nitem+1}(b')$\;
        \textbf{for} $m \in [1,\ldots,M]: b_m^\star \gets \arg\max_{b \le b_{m-1}^*} U_\nitem(b)$\;
        \textbf{Return} $U_1(\max(\mathcal{B}))$ and $\bm{b}^{*} =(b_{m}^*)_{m\in [M]}$.
	\caption{\textsc{Offline}$(\bm{v}, \{\bm{b}^{\nround}_-\}_{\nround \in [\Nround]})$ \label{alg: Offline Full}}	
\end{algorithm}

This algorithm to solve the offline bid optimization problem enables us to compute the hindsight optimal utility, which serves as a benchmark for evaluating the effectiveness of our online learning algorithms. It is worth mentioning that we can represent our DP algorithm as an equivalent graph with $\Nitem$ layers, with $|\mathcal{B}|$ nodes in each. More precisely, we define the (offline) DP graph as follows:
\begin{enumerate}
    \item \textbf{DP nodes.} There are $\Nitem$ layers, each with    $|\mathcal{B}|$ nodes in each, denoted by $\{(\nitem, b)\}_{\nitem \in [\Nitem], b \in \mathcal{B}}$.
    \item \textbf{DP edges.} In this graph, there are only (directed) edges between two consecutive layers, i.e., from layer $m$ to layer $m+1$ for any $m \in [M-1]$. 
    In particular, 
    node $(\nitem, b)$ only has an edge to node $(\nitem+1, b')$ for $b' \leq b$ and if $b' \leq v_{m+1}, b \leq v_m$. 
    \item \textbf{DP weights.} We define the weight of node $(\nitem, b)$ to be $W_{\nitem}^{T+1}(b) = \sum_{\tau=1}^{\Nround} \textbf{1}_{b \geq b^{\tau}_{-\nitem}} (v_\nitem - b)$.
\end{enumerate}

\begin{figure}
    \centering
    \includegraphics[scale=0.24]{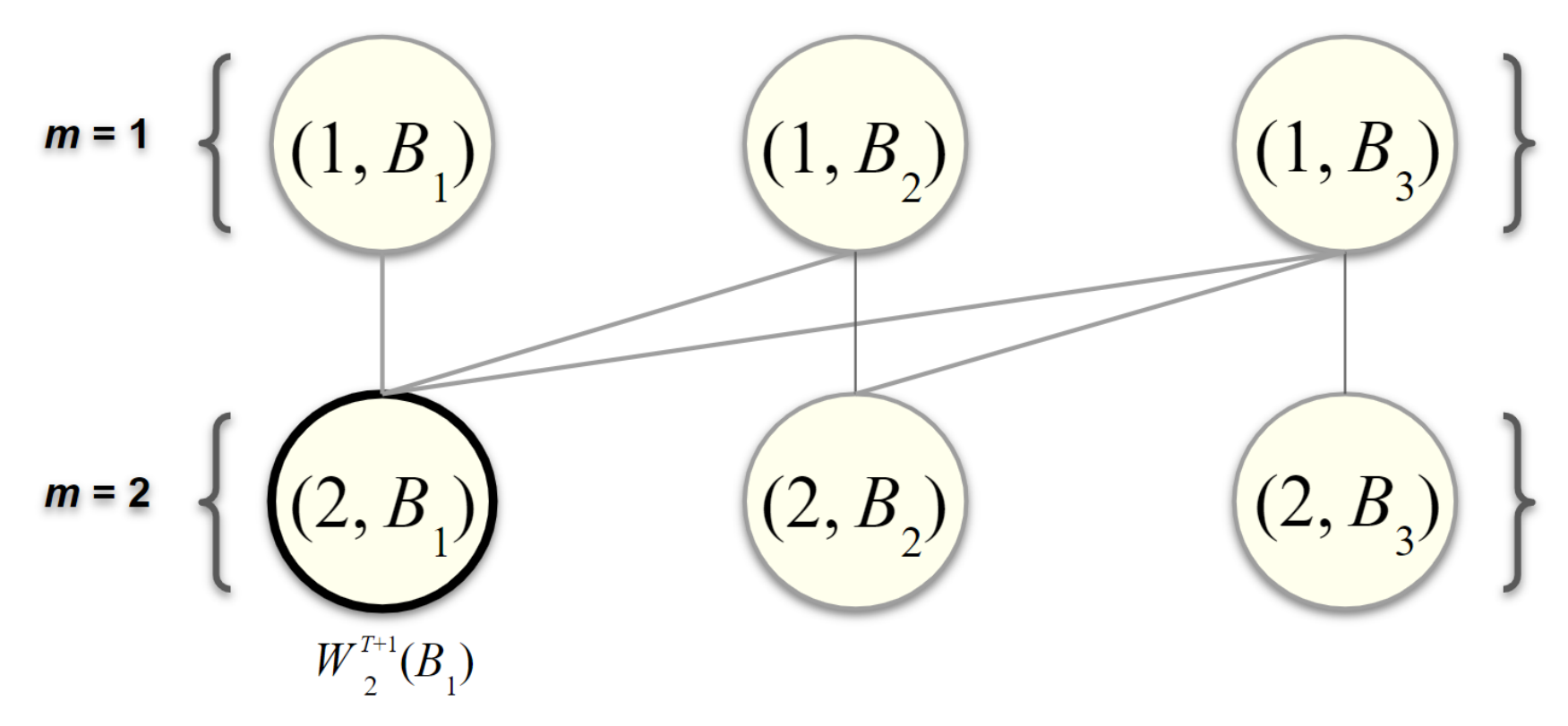}
    \caption{{DP graph for Problem (\ref{eq:offline})}: Bid optimization problem cast as a graph problem, for $\nitem = 2$ and $|\mathcal{B}| = 3$, $B_1< B_2< B_3$. Node $(\nitem, b)$, the node in the $\nitem$'th layer with bid value $b$, has  weight $W_\nitem^{T+1}(b)$.   \label{fig: valid_bids_graph}}
\end{figure}

Now, we extend our DP formulation for solving for the hindsight optimal to the online setting. We consider two feedback structures: (i) full information and (ii) bandit. In the full information setting, the agent's allocation and the values of $\bm{b}^\nround_-$ are revealed after the end of each round, whereas in the bandit setting, only the agent's allocation is revealed. As the bidder can only see information up to round $t$, we can define the DP graph at time $t$, as opposed to $T$, by setting $W_{\nitem}^t(b) = \sum_{\tau=1}^{\nround-1} \textbf{1}_{b \geq b^{\tau}_{-\nitem}} (v_\nitem - b)$. This allows us to construct algorithms for the full information and bandit settings by taking advantage of the structure of the DP graph to enhance efficiency and optimize storage of necessary computations.

\subsection{Online Setting: Decoupled Exponential Weights}

\label{sec: decoupled exp weights section}

In this section, we present our first algorithm for learning in the online setting. In particular, we construct a decoupled version of the Exponential Weights algorithm which circumvent the large space and time complexity of maintaining and updating the sampling distributions of all possible bid vectors. Our algorithms instead sequentially sample a singular bid value from each layer of our DP graph such that the probability of sampling a particular vector of bids $\bm{b}^\nround$ is precisely equal to the probability of the exponential weights algorithm selecting $\bm{b}^\nround$.\footnote{In Section~\ref{sec: time varying}, we present a generalization of our decoupled exponential weights algorithm to the setting with time varying valuations, where the valuations are drawn from some known, finite support distribution $F_{\bm{v}}$.} 
 
In the following sections, we describe our algorithm for both full information and bandit settings. It's important to note that while our decoupled exponential weights algorithm achieves regret suboptimal by a factor of \(O(\sqrt{M})\), we introduce an alternative regret-optimal algorithm based on Online Mirror Descent (OMD) in a subsequent section. Despite this, the decoupled exponential weights algorithm remains practical, as it avoids the need to solve a convex optimization problem at each time step \(t \in [T]\). 

\subsubsection{Full Information Setting}\label{sec:full}

Now let us focus on learning optimal bidding in an online fashion with full information feedback. One straightforward approach in this context is to apply the exponential weights algorithm \citep{DBLP:journals/iandc/LittlestoneW94} to the entire set of bid vectors. This algorithm guarantees per-round rewards within the range of $[-\Nitem, \Nitem]$. However, the challenge lies in the exponentially large bid space $\mathcal{B}^{+\Nitem}$. Tracking and updating weights for all possible bid vectors naively would lead to a non-polynomial time and space complexity. Although this approach achieves a small regret of $O(\Nitem^\frac{3}{2} \sqrt{T \log |\mathcal{B}|})$, we need a more efficient solution with a polynomial time and space complexity.

To do so, we leverage the DP scheme developed in Section \ref{sec:offline}. By utilizing the DP graph and the information it provides about bid vector utilities, we can effectively mimic the exponential weights algorithm without explicitly tracking weights for every bid vector. 
In Algorithm~\ref{alg: Decoupled Exponential Weights}, instead of 
 associating weights to each  possible bid vector, we  associate weights with each $(\nitem, b)$ pair for any $m\in [M]$ and $b\in \mathcal B$. These weights are then updated via variables $S_m^t(b)$ for $b\in \mathcal B, m\in[M]$, which are inspired by the DP scheme. For any round $t$, we define 
 \begin{align*} S^t_\nitem(b) &= \exp(\eta W_\nitem^\nround(b)) \sum_{\bm{b}'_{\nitem+1:\Nitem} \in \mathcal{B}^{+(\Nitem-\nitem)}, b'_{\nitem+1} \leq b'_\nitem = b} \exp(\eta \sum_{\nitem'=\nitem+1}^\Nitem W_{\nitem'}^\nround(b'_{\nitem'}))\\
 &= \exp(\eta W_\nitem^\nround(b)) \sum_{b' \in \mathcal{B}; b' \leq b} S_{\nitem+1}(b')\,.\end{align*}
 Here, $\eta> 0$ is the learning rate of the algorithm and we  recall that  $W_m^t(b)= \sum_{\tau=1}^{t-1} w^{\tau}_\nitem(b)$ is  cumulative utility gained across the first $t-1$ auctions from the winning the $\nitem$'th item with bid $b$, respectively. Computing $S_m^t(\cdot)$ is done in step $\textsc{Compute}-S_\nitem$ of the algorithm.  In step $\textsc{Sample}-\bm{b}$, the bid vector is then sampled  according to $S_m^t$'s subject to bid monotonicity.  To disallow overbidding, we initialize weights  $W^0_m(b) = -\infty$ for all $m \in [M], b \in \mathcal{B}$ such that $b > v_m$.

The concept of utility decoupling shares similarities with solutions used in combinatorial bandits, tabular reinforcement learning \citep{CMAB2013, OREPS2013}, and problems such as shortest path algorithms involving weight pushing or path kernels \citep{PathKernel2003, Hedging2010}. These methods are employed to solve variants of the shortest path or maximum weight path problems, where costs or weights are associated with edges rather than nodes. By exploiting the graph structure and the linearity of utilities with respect to the weights of each edge, these algorithms efficiently compute path weights based on edge weights, similar to how our algorithm computes path weights based on node weights. In addition to investigating a fundamentally different problem,   the key distinction is that our approach  considers weights associated with nodes instead of edges. In our setting, the reward associated with selecting bid $b'$ in slot $\nitem+1$ is independent of selecting bid $b \geq b'$ in slot $\nitem$. This allows us to get an improved regret bound and save a factor of $|\mathcal{B}|$ in terms of time and space complexity. Instead of storing and updating weights for $O(\Nitem |\mathcal{B}|^2)$ possible $(\nitem, b, b')$ slot-value-next value triplets, we only need to handle $O(\Nitem |\mathcal{B}|)$ possible $(\nitem, b)$ unit-bid pairs. The following statement is the main result of this section. 

\begin{theorem}[Decoupled Exponential Weights: Full Information] \label{thm:full}
    With $\eta = \Theta(\sqrt{\frac{\log |\mathcal{B}|}{MT}})$, Algorithm \ref{alg: Decoupled Exponential Weights} achieves (discretized) regret $O(\Nitem^\frac{3}{2} \sqrt{ \Nround \log |\mathcal{B}|})$, with total time and space complexity polynomial in $\Nitem$, $|\mathcal{B}|$, and $\Nround$. {\color{black}Optimizing for discretization error of $O(\frac{MT}{|\mathcal{B}|})$ from restricting the bid space to $\mathcal{B}$, we obtain a continuous regret of $O(\Nitem^\frac{3}{2} \sqrt{\Nround \log \Nround})$.}
\end{theorem}

Note that both the time and space complexity scale polynomially in $\Nitem$ and $|\mathcal{B}|$. {\color{black}For the special case of $M=1$ corresponding to the standard first price auction, we obtain the regret of $\tilde{O}(\sqrt{\Nround})$.}

 \begin{algorithm}[t]
 \footnotesize
	\KwIn{Learning rate $0<\eta < \frac{1}{M}$, $\bm{v} \in [0, 1]^{+\Nitem}$.}
	\KwOut{The aggregate utility $\sum_{\nround=1}^\Nround \mu_n^\nround(\bm{b}^{\nround})$}
	$W_\nitem^0(b) \gets 0$ for all $\nitem \in [\Nitem], b \in \mathcal{B}$ such that $b \leq v_m$; else $W_\nitem^0(b) \gets -\infty$\;
        $b_{0}^t \gets \max \mathcal B$, and $S_{M+1}^t (\min \mathcal{B})=1$ for any $t\in[T]$\;
 
	\For{$\nround \in [1,\ldots,\Nround]$:}{
            \textbf{Recursively Computing Exponentially Weighted Partial Utilities $\bm{S}^t$}\;
            \textbf{for} $m \in [M,\ldots,1], b \in \mathcal{B}: S^t_\nitem(b) \gets \exp(\eta W_\nitem^\nround(b)) \sum_{b' \leq b} S_{\nitem + 1}^\nround(b')$. \label{eq: compute s} \hspace{0mm} $\backslash \backslash$ $\textsc{Compute}-S_\nitem$\;
        \textbf{Determining the Bid Vector $\bm{b}^\nround$ Recursively}\;
        \textbf{for} $m \in [1,\ldots,M], b \leq b_{m-1}^t: b_\nitem^\nround \gets b$ with probability $ \frac{S^t_\nitem(b)}{\sum_{b' \leq b_{\nitem-1}^t} S^t_{\nitem}(b')};$ \hspace{1mm} $\backslash \backslash$ $\textsc{Sample}-\bm{b}$\;
            Observe $\bm{b}^{\nround}_-$ and receive reward $\mu_n^\nround(\bm{b}^{\nround})$\;
        \textbf{Update Weight Estimates} \;
        \textbf{for} $\nitem \in [\Nitem], b \in \mathcal{B}: W^{\nround+1}_{\nitem}(b) \gets W^{\nround}_{\nitem}(b) +  (v_\nitem - b)\textbf{1}_{b \geq b^t_{-m}}$ if $b \leq v_m$; else $W^{\nround+1}_{\nitem}(b) \gets -\infty$\;
        }

        \textbf{Return} $\sum_{\nround=1}^\Nround \mu_n^{\nround}(\bm{b}^{\nround})$
	\caption{\textsc{Decoupled Exponential Weights - Full Information}}
	\label{alg: Decoupled Exponential Weights}
\end{algorithm}

\subsubsection{Bandit Feedback Setting}

\label{sec: decoupled exp weights, bandit feedback}

We extend Algorithm~\ref{alg: Decoupled Exponential Weights} for the bandit feedback setting. In the bandit feedback setting, the bidder's allocation and utility are not available for all possible bid vectors, unlike in the full information setting. Instead, the agent only observes their utility for the submitted bid vector. To handle this, we use inverse probability weighted (IPW) node weight estimates $\widehat{w}^t_m(b)$ instead of the node weights $w^t_m(b)$ in Algorithm~\ref{alg: Decoupled Exponential Weights}. This adaptation results in a regret of $O(M^{\frac{3}{2}}\sqrt{|\mathcal{B}| T \log |\mathcal{B}|})$, as shown in Theorem~\ref{thm:decoupled exp - bandit feedback}. This regret includes an additional factor of $\sqrt{|\mathcal{B}|}$ compared to the full information setting.

The structure of Algorithm \ref{alg: Decoupled Exponential Weights - Path Kernels} is similar to that of Algorithm~\ref{alg: Decoupled Exponential Weights}. Both algorithms maintain node weight estimates, 
compute the sum of exponentiated partial bid vector estimated utilities recursively, and sample bids for each unit recursively proportional to these summed exponentiated utilities.
Specifically, Algorithm \ref{alg: Decoupled Exponential Weights - Path Kernels} samples bid vectors with probabilities proportional to the sum of the cumulative \textit{estimated} utility $\widehat{W}_{\nitem}^\nround(b) = \sum_{\tau=1}^{t-1} \widehat{w}_{\nitem}^\tau(b)$ over each unit-bid value pair, where   $\widehat{w}_{m}^\nround(b) = 1 - \frac{1 - (v_m - b)\textbf{1}_{b \geq b_{-m}^{\nround}}}{q_m^\nround(b)}\textbf{1}_{b_m^\nround = b}$ and $q_m^\nround(b)$ is the (unconditional) probability that bid $b$ is chosen for unit $m$ at time $t$.

Note that this bid vector utility estimator is a slightly different estimator {than the one} used in the standard $\textsc{Exp3}$ algorithm (See Chapter 11 of \cite{Lattimore2020}). In particular, the standard IPW estimator $\widehat{w}_m^\nround(b) = \frac{v_m-b}{q_m^\nround(b)} \textbf{1}_{b = b_m^t > b_{-m}^t}$, while unbiased, can be unboundedly large when $q_m^\nround(b)$ approaches 0, whereas our proposed estimator is bounded above by 1. As a consequence, we have that $\widehat{\mu}_n^\nround(\bm{b})$ is upper bounded by $M$, and therefore $ \eta \widehat{\mu}^\nround(\bm{b}) = \eta \sum_{m=1}^M \widehat{w}_m^\nround(b)$ is upper bounded by 1 for $\eta < \frac{1}{M}$, which we crucially use in the proof. 

 \begin{algorithm}[t]
 \footnotesize
	\KwIn{Learning rate $0 < \eta < \frac{1}{M}$, $\bm{v} \in [0, 1]^{+\Nitem}$}
	\KwOut{The aggregate utility $\sum_{\nround=1}^\Nround \mu_n^\nround(\bm{b}^{\nround})$}
	$\widehat{W}_\nitem^0(b) \gets 0$ for all $\nitem \in [\Nitem], b \in \mathcal{B}$ such that $b \leq v_\nitem$; else $\widehat{W}_\nitem^0(b) \gets -\infty$.\;
        $b_{0}^t \gets \max \mathcal B$, and $\widehat{S}_{M+1}^t (\min \mathcal{B})=1$ for any $t\in[T]$\;
 
	\For{$\nround \in [1,\ldots,\Nround]$:}{
            \textbf{Recursively Computing Exponentially Weighted Partial Utilities $\bm{S}^t$}\;
            \textbf{for} $m \in [M,\ldots,1], b \in \mathcal{B}: \widehat{S}^t_\nitem(b) \gets \exp(\eta \widehat{W}_\nitem^\nround(b)) \sum_{b' \leq b} \widehat{S}_{\nitem + 1}^\nround(b')$ \hspace{0mm} $\backslash \backslash$ $\textsc{Compute}-\widehat{S}_\nitem$\;
        \textbf{Determining the Bid Vector $\bm{b}^\nround$ Recursively}\;
        \textbf{for} $m \in [1,\ldots,M], b \leq b_{m-1}^t: b_\nitem^\nround \gets b$ with probability $\frac{\widehat{S}^t_\nitem(b)}{\sum_{b' \leq b_{\nitem-1}^t} \widehat{S}^t_{\nitem}(b')}; $ \hspace{1mm} $\backslash \backslash$ $\textsc{Sample}-\bm{b}$\;
        Observe $\bm{b}^{\nround}_-$ and receive reward $\mu_n^\nround(\bm{b}^{\nround})$\;
        \textbf{Recursively Computing Probability Measure $\bm{q}$}\;
        $q^t_1(b) \gets \frac{\widehat{S}^\nround_m(b)}{\sum_{b' \in \mathcal{B}} \widehat{S}^\nround_m(b')}$ for all $b \in \mathcal{B}$\;
        \textbf{for} $m \in [2,\ldots,M], b \in \mathcal{B}: q_\nitem^\nround(b) \gets \sum_{b' \geq b} \frac{q_{\nitem-1}^t(b')\widehat{S}^\nround_{\nitem}(b)}{\sum_{b" \geq b'} \widehat{S}^\nround_\nitem(b")}$ for all $b \in \mathcal{B}$\;
        \textbf{Update Weight Estimates}\;
        \textbf{for} $m \in [M], b \in \mathcal{B}: \widehat{W}^{\nround+1}_{\nitem}(b) \gets \widehat{W}^{\nround}_{\nitem}(b) + (1 - \frac{1 - (v_m - b)\textbf{1}_{b \geq b^t_m}}{q^t_m(b)} \textbf{1}_{b^t_m = b})$ if $b \leq v_m$; else $\widehat{W}^{\nround+1}_{\nitem}(b) \gets -\infty$\;
        }
        \textbf{Return} $\sum_{\nround=1}^\Nround \mu_n^{\nround}(\bm{b}^{\nround})$
	\caption{\textsc{Decoupled Exponential Weights - Bandit Feedback}}
	\label{alg: Decoupled Exponential Weights - Path Kernels}
\end{algorithm}

    The primary difference in the implementation of Algorithm~\ref{alg: Decoupled Exponential Weights - Path Kernels} as compared to Algorithm~\ref{alg: Decoupled Exponential Weights} is that we require additional steps in order to obtain unbiased node weight estimates $\widehat{W}^{\nround+1}_{\nitem}(b) = \sum_{\tau=1}^{t} \widehat{w}_m^\tau(b)$ which we compute using an IPW estimator. In order to do this, we must compute $q_m^t(b)$---the  probabilities of selecting bid $b$ at slot $m$.

\begin{theorem}[Decoupled Exponential Weights: Bandit Feedback] \label{thm:decoupled exp - bandit feedback}
    With $\eta = \Theta(\sqrt{\frac{\log |\mathcal{B}|}{M|\mathcal{B}|T}})$ such that $\eta < \frac{1}{M}$, Algorithm \ref{alg: Decoupled Exponential Weights - Path Kernels} achieves (discretized) regret $O(\Nitem^{\frac{3}{2}} \sqrt{ |\mathcal{{B}|}\Nround \log |\mathcal{B}|})$, with total time and space complexity polynomial in $\Nitem$, $|\mathcal{B}|$, and $\Nround$. {\color{black}Optimizing for discretization error of order $O(\frac{MT}{|\mathcal{B}|})$ from restricting the bid space to $\mathcal{B}$, we obtain a continuous regret of $O(M^{\frac{4}{3}}T^{\frac{2}{3}} \sqrt{\log \Nround})$.}
\end{theorem}

{\color{black}For the case of $M=1$---the standard first price auction---the (continuous) regret scales with $\tilde{O}(T^{\frac{2}{3}})$, which matches the (tight) regret bound of \cite{ContextBanditsCrossLearning2019} for the bandit setting.}

\subsection{Online Setting: Online Mirror Descent} 
\label{sec:bandit}

We propose our second online learning algorithm which achieves faster regret rates than our decoupled exponential weights algorithms at the cost of additional computation. Instead of mimicking the exponential weights algorithm, we reformulate the problem as online linear optimization over node probabilities in our DP graph. We solve for these probabilities using OMD and construct a corresponding policy that sequentially samples bids based on these probabilities.

 \begin{algorithm}[t]
 \footnotesize
	\KwIn{Learning rate $\eta > 0$, Valuation $\bm{v} \in [0, 1]^{+\Nitem}$ } 
	\KwOut{The aggregate utility $\sum_{\nround=1}^\Nround \mu_n^\nround(\bm{b}^\nround)$.}
	$\pi_0((m, b), b') \gets \frac{1}{|\{b" \in \mathcal{B}: b" \leq b, b" \leq v_{m+1}\}|}$ for all $\nitem \in [\Nitem], b \geq b' \in \mathcal{B}, b' \leq v_{m+1}$. Let $\bm{q}^0 \in [\Nitem] \times \mathcal{B} \to [0, 1]$ be the corresponding unit-bid value pair occupancy measure\;
	\For{$\nround \in [\Nround]$:}{
            \textbf{Determining the Bid Vector $\bm {b}^t$ recursively.} Set $b_1$ to $b \in \mathcal{B}$ with probability  $q^t_1(b)$\;
            \textbf{for} $m \in [1,\ldots,M-1], b \in \mathcal{B}: b_{m+1} \gets b$ with probability $\pi^t((m, b_\nitem), b)$\;
            Receive reward $\mu^\nround_n(\bm{b}^\nround) = \sum_{\nitem=1}^\Nitem w_\nitem^\nround(b^\nround_\nitem)$ and observe $w_\nitem^\nround(b^\nround_\nitem)$ where $w_\nitem^\nround(b) = (v_\nitem - b)\textbf{1}_{b \geq b^\nround_{-\nitem}}$\;
            \textbf{Update Reward Estimates}\;
            \textbf{for} $\nitem \in [\Nitem], b \in \mathcal{B}: \widehat{w}_\nitem^\nround(b) \gets \frac{w_\nitem^\nround(b)}{q^{\nround-1}_\nitem(b)} \textbf{1}_{b = b^{\nround}_\nitem}$ if \emph{Bandit Feedback}, $\widehat{w}_\nitem^\nround(b) \gets w_\nitem^\nround(b)$ if \emph{Full Information}\;
            \textbf{Determining Probability Measure $\bm{q}^t$ over any unit-bid value pair $(m, b)$.} Set \vspace{-2mm}
            \begin{align}
                \bm{q}^\nround \gets \text{argmin}_{\bm{q} \in \mathcal{Q}} \eta\langle \bm{q}, -\widehat{\bm{w}}^\nround\rangle + D(\bm{q} || \bm{q}^{\nround-1}) \label{eq:q_t}\,
            \end{align} 
            where $\mathcal{Q}$ is as in Equation ~\eqref{eq:Q} and $D(\bm{q} || \bm{q}') = \sum_{\nitem \in [\Nitem], b \in \mathcal{B}} q_\nitem(b)\log \frac{q_\nitem(b)}{q'_\nitem(b)} - (q_\nitem(b) - q'_\nitem(b))$.
            
            \textbf{Convert $\bm{q}^t$ to Policy $\bm{\pi}^t$.}
            Compute any feasible solution $\bm{\pi}^t$ to constraints $q^t_m(b) = \sum_{b' \geq b} q^t_m(b') \pi^t((m-1, b'), b)$ and $\sum_{b" \leq b} \pi^t((m, b), b") = 1$ for all $m \in [M], b \in \mathcal{B}$.
        }
        \textbf{Return $\sum_{\nround=1}^\Nround \mu_n^\nround(\bm{b}^\nround)$.} 
	\caption{\textsc{OMD - Bid Optimization in Multi-Unit Pay as Bid Auctions}}
	\label{alg: OMD}
\end{algorithm}

\textbf{Algorithm Ideas. }Recall that $\bm{q}: [M] \times \mathcal{B} \to [0, 1]$ denotes the node probabilities; i.e., $q_m^t(b)$ is the probability of selecting bid $b$ for unit $m$ at round $t$. Utilizing PAB's reward structure, we can rewrite the aggregate utility---and consequently, the regret---as a function of $\{\bm{q}^t, \bm{w}^t\}_{t \in [T]}$
\begin{align}
    \mathbb{E}_{\bm{b} \sim \bm{\pi}}\left[ \sum_{t=1}^T \mu^{\nround}_n(\bm{b}) \right] = \sum_{t=1}^T \sum_{\nitem=1}^\Nitem \mathbb{E}_{b_\nitem \sim \bm{q}^t_\nitem}\left[ w_\nitem^\nround(b_\nitem) \right] = \sum_{t=1}^T \sum_{\nitem=1}^\Nitem  \sum_{b \in \mathcal{B}} q^t_\nitem(b) w_\nitem^\nround(b) = \sum_{t=1}^T \langle \bm{q}^t, \bm{w}^\nround \rangle\,.
    \label{eq: Loss of policy}
\end{align}
We make the important observation that the rewards are \textit{linear} in the $\bm{q}^t$'s, which immediately hints towards the use of online linear optimization (OLO). It is also straightforward to verify linearity\footnote{While showing that the linearity of set $\mathcal{Q}$ is trivial, we must also show for completeness that any valid policy over our DP graph $\bm{\pi}$ yields node probabilities $\bm{q} \in \mathcal{Q}$. We leave the proof of this in the online appendix, Lemma~\ref{lem: QSpace Equivalence}.} of the feasible region $\mathcal{Q}$ of $\bm{q}$, as they only contain non-negativity, probability mass, and stochastic dominance constraints:
\begin{align}
    \mathcal{Q} = \Big\{\bm{q} \in [0, 1]^{M \times |\mathcal{B}|}:  \sum_{b \in \mathcal{B}} q_m(b) = 1, \sum_{b \leq b'} q_{m+1}(b) \geq \sum_{b \leq b'} q_m(b) \forall b, b' \in \mathcal{B}, m \in [M]\Big\}\,.
    \label{eq:Q}
\end{align}
Now, we make a second important observation that there is no explicit dependence of the regret or feasible region $\mathcal{Q}$ on the policy $\bm{\pi}$ that generated the node probabilities $\bm{q}$. Indeed, one of our primary contributions is the insight that our algorithm's regret guarantee is agnostic w.r.t. the exact policy $\bm{\pi}$, conditional upon having the same node probabilities $\bm{q}$. Unlike existing algorithms that optimize over policies $\bm{\pi}^t$ over a graph's node-edge pairs \cite{OREPS2013, PathKernel2003} which is of size $O(M|\mathcal{B}|^2)$ in our setting, our approach of optimizing over node probabilities saves an additional factor of $|\mathcal{B}|$ in both regret and computation. Using negentropy regularization as in Equation \eqref{eq:q_t} of Algorithm~\ref{alg: OMD}, we achieve continuous regret \textit{linear} in $M$:

\begin{theorem}[Online Mirror Descent: Bandit Feedback] \label{thm: OMD}    With $\eta = \Theta(\sqrt{\frac{\log |\mathcal{B}|}{|\mathcal{B}|\Nround}})$, Algorithm~\ref{alg: OMD} achieves (discretized) regret $O(\Nitem \sqrt{|\mathcal{B}| \Nround \log |\mathcal{B}|})$, with total time and space complexity polynomial in $\Nitem$, $|\mathcal{B}|$, and $\Nround$. {\color{black}Optimizing for discretization error of order $O(\frac{MT}{|\mathcal{B}|})$ from restricting the bid space to $\mathcal{B}$, we obtain a continuous regret of $O(\Nitem \Nround^{\frac{2}{3}})$.} 
\end{theorem}

We defer the proof to the online appendix (arXiv:2307.15193v3). Under full information, we recover the regret of Algorithm \ref{alg: Decoupled Exponential Weights} by replacing the node weight estimates with the true weights.

\begin{corollary}[Online Mirror Descent: Full Information] \label{cor}
    With $\eta = \Theta(\sqrt{\frac{\log |\mathcal{B}|}{T}})$, Algorithm \ref{alg: OMD} achieves (discretized) regret $O(\Nitem \sqrt{ \Nround \log |\mathcal{B}|})$, with total time and space complexity polynomial in $\Nitem$, $|\mathcal{B}|$, and $\Nround$. {\color{black}Optimizing for discretization error of order $O(\frac{MT}{|\mathcal{B}|})$ from restricting the bid space to $\mathcal{B}$, we obtain a continuous regret of $O(\Nitem \sqrt{\Nround \log \Nround})$.}
\end{corollary}

\section{Regret Lower Bound}
\label{sec: lower bound}
{\color{black}In this section, we complement our upper bound results on regret with corresponding lower bounds. We introduce two bounds: for the full information setting, we establish a regret lower bound of \(\Omega(\Nitem \sqrt{\Nround})\), which is also a valid lower bound for the bandit setting. This lower bound demonstrates that even in a full-information, stochastic setting, continuous regret should scale linearly with \(\Nitem\), reflecting the linear dependency observed in our OMD-based algorithms (Algorithm \ref{alg: OMD}) that achieve regret \(O(\Nitem \sqrt{\Nround \log \Nround})\). Furthermore, our regret lower bound aligns with the continuous regret upper bound for the OMD algorithm in the full information setting, up to logarithmic factors. This matching bound indicates  that no algorithm, with or without discretization, can achieve smaller regret. For the bandit setting, we additionally present a regret lower bound of \(\Omega(\Nitem^{\frac{2}{3}} \Nround^{\frac{2}{3}})\), matching the regret lower bounds of our two algorithms in terms of \(\Nround\)'s dependence for the bandit setting.}

\begin{theorem}[Regret Lower Bound for the Full Information Setting]\label{thm:lower}
    In the full information setting, the continuous regret of learning to bid in PAB is $\Omega(\Nitem \sqrt{\Nround})$. This implies an equivalent lower bound in the bandit setting.
\end{theorem}

{\color{black}To establish lower bounds, we devise a stochastic adversary whose bidding distribution makes it challenging for the bidder to optimize bids, resulting in $\Omega(M\sqrt{T})$ regret. We define bid vectors $\bm{b}'_{-}$ and $\bm{b}"_{-}$ and construct two adversary bid distributions $F$ and $G$. Under $F$, $\prob(\bm{b}_-^\nround = \bm{b}'_{-}) = \frac{1}{2} + \delta$, and under $G$, $\prob(\bm{b}_{-}^\nround = \bm{b}'_{-}) = \frac{1}{2} - \delta$. We analyze the expected utilities for bidding vectors under $F$ and $G$ and derive expressions for expected utilities. By evaluating these expressions, we compute a per-step incurred regret of $\Theta(\Nitem \delta)$. Employing Le Cam's method and relating regret under $(F+G)/2$ to the Kullback-Leibler divergence between $F$ and $G$, we establish a regret bound of $\Omega(\Nitem \sqrt{\Nround})$. This bound holds when $\delta$ is chosen to be $\Omega(\frac{1}{\sqrt{T}})$. 

The main novelty in the proof is the structure of  the highest  competing bids $\bm{b}'_{-}$ and $\bm{b}"_{-}$, where we set
$\bm{b}'^{-} = (0,\ldots,0,c,\ldots,c)$, with $k$ values of 0 and $\Nitem - k$ values of $c$. Additionally, we define  $\bm{b}"^{-} = (c,\ldots,c)$ as the $\Nitem$-vector of bids at $c$. We then show that  for certain choices of $c$ and $k$, the expected utility-maximizing bid vector under ${\bm{b}_{-}^\nround} \sim F$ is $(0,\ldots,0)$ and under ${\bm{b}_{-}^\nround} \sim G$ is $(c,\ldots,c)$. This leads to the desired results after choosing $c$ and $k$ carefully.}

{\color{black}

\begin{theorem}[Regret Lower Bound for the Bandit Setting] \label{thm:lower_bound_bandit}
Assuming a time horizon of \(T\), a number of units \(M\), and that \(T \geq M^2\), the regret of learning to bid in PAB auctions is lower bounded by \(\Omega(M^{\frac{2}{3}} T^{\frac{2}{3}})\).
\end{theorem}

Theorems \ref{thm:lower} and \ref{thm:lower_bound_bandit}, which shown in Section \ref{sec:proof:lower:bandit}, together provide the regret lower bound of \(\Omega(\max\{ M^{\frac{2}{3}} T^{\frac{2}{3}}, M\sqrt{T} \})\) for the bandit setting. 
To construct the regret lower bound of \(\Omega(M^{\frac{2}{3}} T^{\frac{2}{3}})\) for the bandit setting, we begin with a base hypothesis that specifies a distribution over competing bids. Under this base hypothesis, the per-unit utility of unit \(m\) for each possible bid value is the same and equals some constant \(c_m\), which decays with \(m\). For each unit \(m\), we then construct \(O\left(\frac{|\mathcal{B}|}{M}\right)\) different hypotheses by perturbing the base distribution for one of the bids in a certain partition of \(\mathcal{B}\), denoted by \(\mathcal{B}_m\). See Figure \ref{fig:partitions} for an illustration of the partitions.

In fact, under the alternative hypothesis, the perturbed bid is optimal, and hence any learning algorithm should be able to identify the perturbed bid for any unit \(m\). The main challenge in such a construction is to ensure no ``cross-learning" occurs across units. That is, by learning the utility for a given unit, we do not learn anything about any other units. In our construction, the way we construct the alternative hypotheses ensures that their perturbations are disjoint across units, allowing us to prevent any cross-learning between units.
  
\begin{figure}
    \centering
    \includegraphics[width=0.8\textwidth]{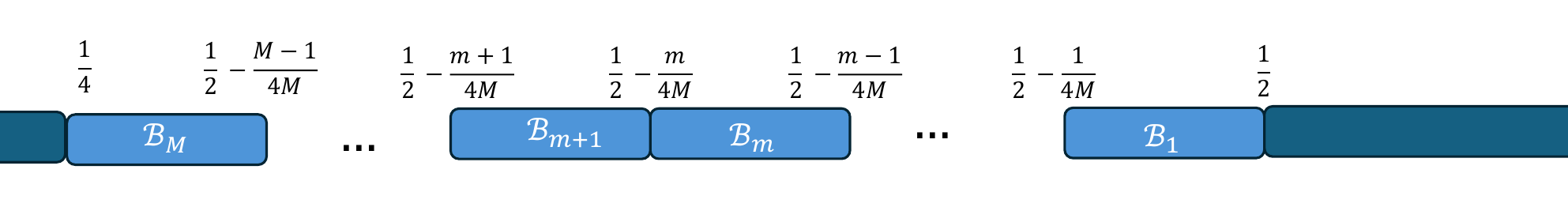}
    \caption{A schematic illustration of how we partition the set of bids in  $\mathcal B\subset [0,1]$ to construct alternative hypotheses. {Under one of these hypotheses denoted by $(j_1,\ldots,j_M)$, we set the marginal distribution of each $b_{-m}$ to be such that all bids $b \in \mathcal{B}_m$ yield $c_m$ expected utility, except for $b_m^{j_m}$---the $j_m$'th largest bid in $\mathcal{B}_m$---which yields $c_m + \gamma_m$ expected utility where $\gamma_m>0$.}}
    \label{fig:partitions}
\end{figure}
}

\section{Experiments}
\label{sec: experiments}

{\color{black} We run two experiments; the first exclusively to empirically verify the regret guarantees of our algorithms for PAB, and the second to compare between the market dynamics of no regret learning in PAB and uniform price auctions. In first experiment, we simulate the market dynamics induced by our decoupled exponential weights algorithms under full information (Algorithm~\ref{alg: Decoupled Exponential Weights}) and bandit feedback (Algorithm~\ref{alg: Decoupled Exponential Weights - Path Kernels}). We omit our OMD algorithm (Algorithm~\ref{alg: OMD}) due to the prohibitively large computational cost of running convex optimization algorithms for each bidder at each time step.\footnote{We do, however, compare the three algorithms in the stochastic setting in Section~\ref{sec:stochastic}.} Moreover, as we observe in our simulations (Figure \ref{fig:regret-plot}), the decoupled exponential weights algorithms achieve the same linear regret in $M$ as our OMD algorithm, despite the weaker theoretical regret guarantee of order $O(M^{3/2})$.

Our second suite of experiments explores the impact of $M, \mathcal{B}, N$, and $T$ on welfare, revenue, regret, and other statistics in the PAB and uniform price settings with full information. We run 50 trials of each possible parameterization for each experiment, where in each trial, all marginal valuations are drawn from a Unif(0, 1) distribution (and then sorted).

{In Section~\ref{sec: additional experiments}, we run additional experiments showing the evolution of the aforementioned quantities over the course of the market dynamics for both PAB and uniform price, as well as comparing the PAB and uniform price equilibria under bandit feedback.} \footnote{In our implementation of all algorithms, we impose that $b_{n, m}^t \leq v_{n, m}$ for all $t$, which prevents the learning dynamics from converging to an overbidding equilibrium, such as in \cite{inefficiency2013}.} }

{\color{black}\textbf{First Experiment: Regret as a Function of $M$ and $T$.} To better understand the impact of $M$ and $T$ on the continuous regret, we run the repeated auction setting with varying $T, M$. Recall that the discrete regret of PAB is of order $O(M^\frac{3}{2}\sqrt{T \log |\mathcal{B}|})$ and $O(M^\frac{3}{2}\sqrt{|\mathcal{B}| T \log |\mathcal{B}|})$ in the full information and bandit feedback settings respectively. As such, we set $|\mathcal{B}|$ to balance the discretization error and discrete regret: $|\mathcal{B}| = \max(5, \sqrt{\frac{T}{M}})$ and $\eta = \sqrt{\frac{\log|\mathcal{B}|}{MT}}$ (for the full information, PAB setting) and $|\mathcal{B}| = \max(5, M^\frac{1}{3}T^{\frac{1}{3}})$ and $\eta = \sqrt{\frac{\log|\mathcal{B}|}{M|\mathcal{B}|T}}$ (for the bandit, PAB setting). Recall that the continuous regret is the sum of the discrete regret and the corresponding discretization error. Here, we plot the $\log-\log$ sum of {discrete} regrets across all $N=3$ bidders, normalized by $NT$ to obtain the per-bidder, per-round average regret; see Figure \ref{fig:regret-plot}. Running a linear regression on the median, we find that the slopes w.r.t. $T$ of the median regret are approximately $-\frac{1}{2}$ and $-\frac{1}{3}$ for the full information and bandit settings, respectively. This confirms the $\sqrt{T}$ and $T^{\frac{2}{3}}$ regret dependence for Theorem~\ref{thm:full} and \ref{thm:decoupled exp - bandit feedback}. Interestingly, the slopes w.r.t. $M$ of the median regret is approximately equal to $1$, suggesting that the $M^\frac{3}{2}$ dependence in Theorems~\ref{thm:full} and \ref{thm:decoupled exp - bandit feedback} is not tight by a factor of $\sqrt{M}$. \footnote{For the values of $|\mathcal{B}|$ chosen to balance the discretization error (which is of order $O(\frac{MT}{|\mathcal{B}|})$) and \textit{theoretical} discretized regret upper bound, the discretization error would scale with $O(M^\frac{3}{2})$ and would eventually dominate the discretized regret plotted in Figure~\ref{fig:regret-plot}. However, after observing the $O(M)$ regret scaling, we scale $|\mathcal{B}| = \Theta(\sqrt{T})$ under full information (resp. $|\mathcal{B}| = \Theta(T^{\frac{1}{3}})$ under bandit) to improve the continuous regret to $O(M\sqrt{T})$ (resp. $O(MT^{\frac{2}{3}})$).} }

\begin{figure}
    \centering
    \includegraphics[scale=0.30, trim={0 0 130 55},clip]{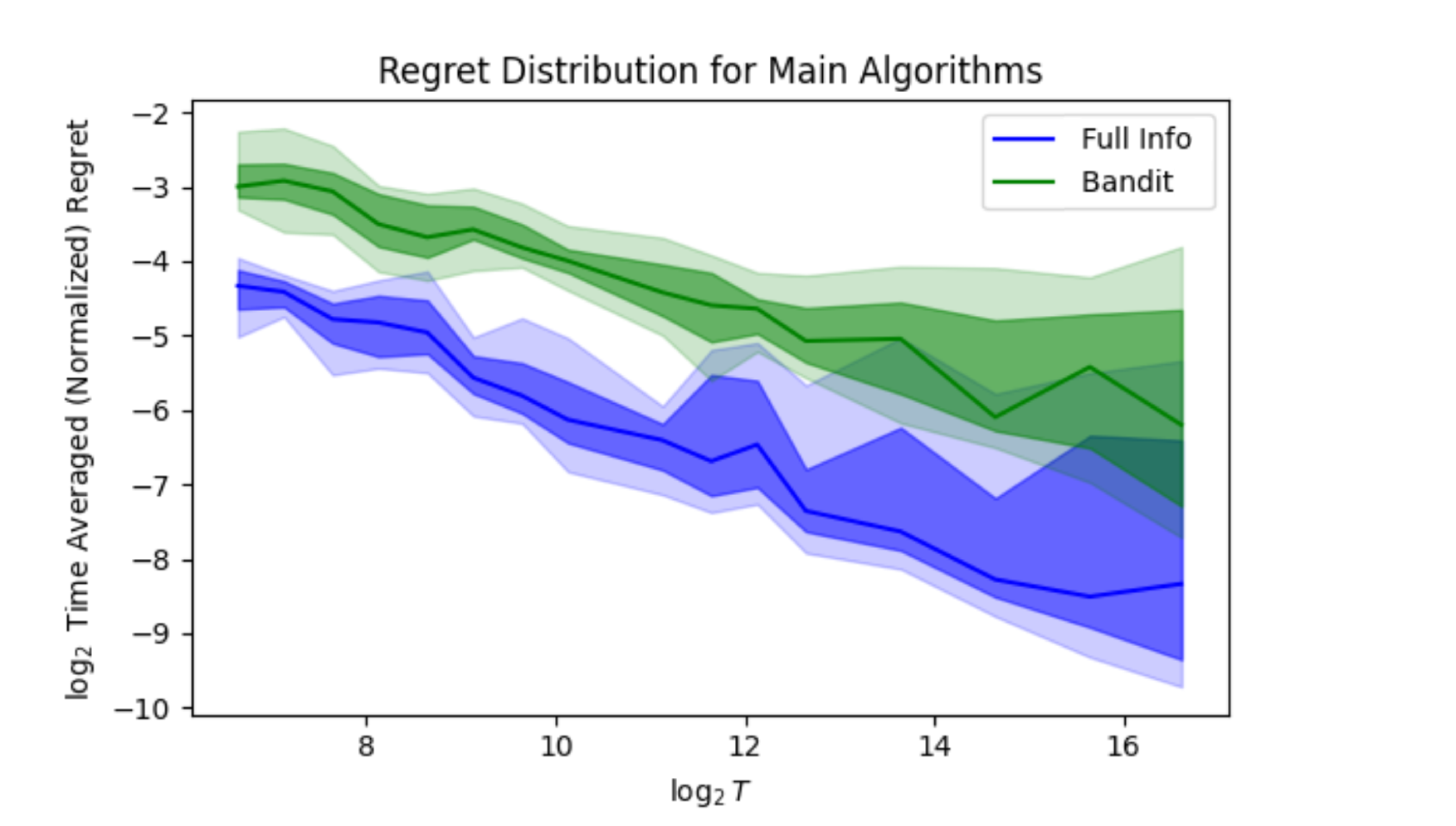}
    \includegraphics[scale=0.26, trim={0 0 0 55},clip]{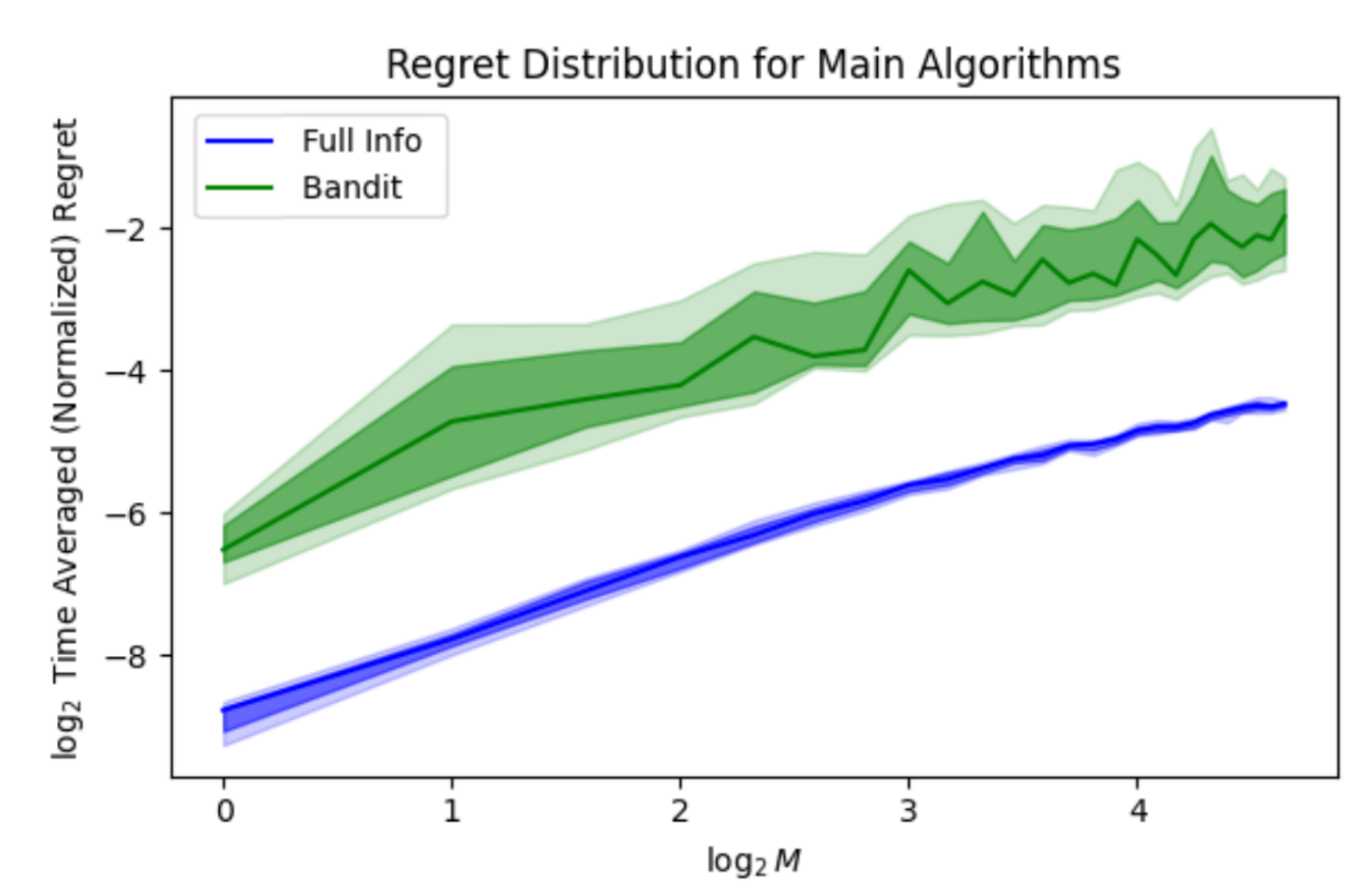}
    \caption{We compare the time-averaged aggregate discretized regret across all agents for the PAB, under both full information and bandit feedback, for both varying $T$ (left) and $M$ (right). When varying $T$ (resp. $M$), we fix $M = 5$ (resp. $T = 25000$) and derive $|\mathcal{B}|$ and $\eta$ according to Theorems~\ref{thm:full} and ~\ref{thm:decoupled exp - bandit feedback}.}
    \label{fig:regret-plot}
\end{figure}

{\color{black}\textbf{Second Experiment: Impact of $M, |\mathcal{B}|, N$ on Equilibrium Behavior for PAB and Uniform Price.} We compare the impacts of $M$, $|\mathcal{B}|$, and $N$ on welfare, revenue, regret, competitive ratio (which we will define shortly), and run-times\footnote{We ran all experiments on a Centos 6, 2x8 cores Intel Xeon 2.0 GHz, with 64 GB RAM.} of the market dynamic simulations for both PAB and uniform price. For the repeated auction parameters, we set $N \in \{3, 5\}$ bidders, $\overline{M} = M \in \{1, 5, 10\}$ items, with all valuations drawn from $\text{Unif}(0, 1)$ and then sorted. The bid space is either $\mathcal{B} =\{\frac{i}{10}\}_{i \in [10]}$ or $\mathcal{B} = \{\frac{i}{20}\}_{i \in [20]}$. We run the market dynamics using the decoupled exponential weights algorithm under full information (Algorithm~\ref{alg: Decoupled Exponential Weights} and the full information algorithm as described in \cite{brânzei2023online} with $T=10^5$, $\eta = \sqrt{\frac{\log |\mathcal{B}|}{MT}}$).}

\begin{table}[ht]
    \centering
    \begin{minipage}[t]{0.48\linewidth} 
        \caption{Summary of the Bidding Dynamics (PAB, Full Information, Algorithm~\ref{alg: Decoupled Exponential Weights})}
        \label{table: learning dynamics full info}
        \footnotesize
        \begin{tabular}{lcccc}
        \toprule
         & \multicolumn{2}{c}{$|\mathcal{B}| = 10$} & \multicolumn{2}{c}{$|\mathcal{B}| = 20$} \\
        \cmidrule(lr){2-3} \cmidrule(lr){4-5}
        Metric & $N = 3$ & $N = 5$ & $N = 3$ & $N = 5$ \\
        \midrule
        \multicolumn{5}{c}{\textbf{For \( M = 1 \)}} \\
        Regret  & .02/.03 & .02/.02 & .03/.04 & .02/.03 \\
        Welfare Gap & .13/4.9 & .34/3.2 & .29/1.3 & .5/2.3 \\
        Revenue Gap & 32/61 & 12/49 & 30/63 & 18/33 \\
        CR Gap & .31/2.4 & 1.2/3.5 & .36/2 & .86/3.9 \\
        Runtime (sec) & 206 & 253 & 299 & 571 \\
        \midrule
        \multicolumn{5}{c}{\textbf{For \( M = 5 \)}} \\
        Regret  & .17/.19 & .12/.13 & .21/.23 & .14/.16 \\
        Welfare Gap  & 1.1/3.2 & 1.1/3.1 & .92/2.3 & 1/2.2 \\
        Revenue Gap  & 35/47 & 19/26 & 32/46 & 18/27 \\
        $b_{(1)}/b_{(M)}$ & 1.19/1.34 & 1.12/1.17 & 1.07/1.18 & 1.06/1.08 \\
        CR Gap  & .35/.69 & .75/1.2 & .46/.76 & .93/1.4 \\
        Runtime (sec) & 1370 & 1833 & 2419 & 4122 \\
        \midrule
        \multicolumn{5}{c}{\textbf{For \( M = 10 \)}} \\
        Regret  & 0.37/0.40 & 0.24/0.26 & 0.45/0.49 & 0.29/0.32 \\
        Welfare Gap  & 1.8/3.7 & 1.3/2.7 & 1.3/2.3 & 1.1/1.8 \\
        Revenue Gap  & 31/40 & 18/22 & 30/39 & 17/22 \\
        $b_{(1)}/b_{(M)}$ & 1.16/1.23 & 1.14/1.16 & 1.09/1.16 & 1.06/1.08 \\
        CR Gap  & 0.46/0.64 & 0.84/1.14 & 0.59/0.73 & 1.00/1.49 \\
        Runtime (sec) & 2512 & 3811 & 4679 & 8330 \\
        \bottomrule
        \end{tabular}
    \end{minipage}\hfill
    \begin{minipage}[t]{0.48\linewidth} 
        \caption{Summary of the Bidding Dynamics (Uniform Price, Full Information \cite{brânzei2023online})}
        \label{table: learning dynamics full uniform}
        \footnotesize
        \begin{tabular}{lcccc}
        \toprule
         & \multicolumn{2}{c}{$|\mathcal{B}| = 10$} & \multicolumn{2}{c}{$|\mathcal{B}| = 20$} \\
        \cmidrule(lr){2-3} \cmidrule(lr){4-5}
        Metric & $N = 3$ & $N = 5$ & $N = 3$ & $N = 5$ \\
        \midrule
        \multicolumn{5}{c}{\textbf{For \( M = 1 \)}} \\
        Regret & .03/.04 &	.02/.03&	.03/.04&	.02/.03\\
        Welfare Gap  & .3/4 & .2/2.9 & .3/1.6 & .3/1.1 \\
        Revenue Gap  & 30/58 & 20/50 & 50/74 & 20/40 \\
        CR Gap & .4/1.2&	.8/1.4	&.6/2.7&	1/2.3\\
        Runtime (sec) & 181&	231	&316&	538 \\
        \midrule
        \multicolumn{5}{c}{\textbf{For \( M = 5 \)}} \\
        Regret & .19/.23	&.14/.18	&.21/.25	&.15/.17\\
        Welfare Gap  & .3/2.4 & .6/1.8 & .5/2.2 & .4/1.4 \\
        Revenue Gap  & 44/56 & 17/32 & 38/52 & 19/30 \\
        $b_{(1)}/b_{(M)}$ & 1.76/2.23 & 1.23/1.49 & 1.59/1.95 & 1.24/1.47 \\
        CR Gap & .3/.6&	.8/1.7&	.4/.6	&.8/1.3\\
        Runtime (sec) & 1932&	2850&	6976&	10883 \\
        \midrule
        \multicolumn{5}{c}{\textbf{For \( M = 10 \)}} \\
        Regret & .49/.66	&.41/1.3&	.5/.6	&.4/.5\\
        Welfare Gap  & .3/1.5 & .5/1.1 & 1/2.6 & .3/.6 \\
        Revenue Gap  & 37/57 & 19/26 & 40/50 & 19/23 \\
        $b_{(1)}/b_{(M)}$ & 1.86/2.39 & 1.33/1.42 & 1.81/2.17 & 1.28/1.36 \\
        CR Gap & .4/.9	&1.2/4.6&	.5/.7&	1.2/2.0\\
        Runtime (sec) & 3509&	6132&	13118&	20913\\
        \bottomrule
        \end{tabular}
    \end{minipage}
\end{table}

{\color{black}We report our results in Tables~\ref{table: learning dynamics full info} and ~\ref{table: learning dynamics full uniform}. In these tables, the regret is defined as the sum of {discrete} regrets across all bidders, normalized by $NT$, the number of bidders times the number of rounds. The welfare gap and revenue gap rows indicate the difference between the maximum welfare and the time-averaged welfare and revenue, respectively, normalized by the maximum welfare defined as the sum of the $M$ largest valuations. The ratio $b_{(1)}/b_{(M)}$ in the tables represents the ratio of the largest to the smallest winning bid, which is undefined for $M = 1$. This ratio is computed by taking the mean of the per-round bid ratios over the final 10\% of the time steps, as we are only concerned with long run bid convergence. The CR Gap in the tables represents the time-averaged competitive ratio subtracted from 1; that is, it is 1 minus the sum over the difference of each bidder's realized and hindsight optimal utilities, normalized by the sum over all bidders' hindsight optimal utility. Each entry of Tables \ref{table: learning dynamics full info} and \ref{table: learning dynamics full uniform} (except those regarding the runtime) denotes the median value across 50 trials with $T = 10^5$. The first number in each entry corresponds to the median, whereas the second corresponds to the 90th percentile for the first four rows. The runtime row denotes the median duration each trial was run for in seconds. Our main observations are as follows:}
{\color{black}
\begin{enumerate}
    
    \item \textbf{Regret Grows with $M, |\mathcal{B}|$.} Consistent with the regret guarantees from the aforementioned theorems, as well as Figure~\ref{fig:regret-plot}, the regret increases as a function of $M$ and $|\mathcal{B}|$. Intuitively, there are $O(M|\mathcal{B}|)$ weights that need to be learned in our algorithm, thus, requiring more exploration for larger $M$ and $|\mathcal{B}|$ leading to larger regret. Interestingly, the regret for uniform price only increases in $M$, despite requiring $O(M|\mathcal{B}|^2)$ weights to learn.

    \item \textbf{Welfare Worsens in $M$, Improves in $|\mathcal{B}|$ for PAB; not for Uniform Price. } For PAB, the welfare gap increases for larger $M$, but for lower $|\mathcal{B}|$, though this trend is not present under uniform price. This phenomenon occurs due to the deterministic tie-breaking rule. As $M$ increases and $|\mathcal{B}|$ decreases, there is an increased likelihood in PAB that low tie-break priority bidders with winning valuations will lose items to high tie-break priority bidders with marginally lower valuations. When $|\mathcal{B}|$ is small, the increase in price will disincentivize these bidders from trying to win these items by increasing their bid. In contrast, under uniform price, only the lowest winning bid affects payments, so these low tie-break priority bidders can still win their higher value items. As such, we find that the uniform price marginally outperforms PAB in welfare.

    \item \textbf{Revenue Improves in $|\mathcal{B}|$ and $N$. } As $|\mathcal{B}|$ increases, the auctioneer can extract larger payments for bidders, as seen in the decrease in revenue gap for larger $|\mathcal{B}|$, though it is more pronounced under PAB and less so for uniform price. The revenue gaps, in both PAB and uniform price, benefit from this increase $N$, as the increased competition decreases the impact of strategic bid shading. We find that PAB consistently outperforms uniform price in terms of revenue, and more importantly, to a larger degree than which uniform price outperforms PAB in welfare.
    
    \item \textbf{The Largest to Smallest Winning Bid Ratio Approximately Converges in PAB; not Uniform Price.} We see that the ratio $b_{(1)}/b_{(M)}$ is around 1.1 or 1.2 for PAB, but around 1.4 to 2 for uniform price, with this ratio decreasing in $N$ and $|\mathcal{B}|$. As $N$ increases, the increased competition turns bidders into price takers, yielding a tightened spread on the winning bid values. As for $|\mathcal{B}|$, in the case of PAB, we notice that the bid ratio is close to 1, suggesting that the difference between the largest and smallest bid ratio is approximately one discretization factor $\delta = \frac{1}{|\mathcal{B}|}$. With the revenue gap being approximately $\frac{1}{3}$ and $\frac{1}{5}$ for $N = 3, 5$ respectively, the clearing price would have been around $\frac{2}{3}$ and $\frac{4}{5}$. Taking the ratio of the two closest values in $\mathcal{B}$ to these yields 1.17 and 1.14 for $|\mathcal{B}| = 10$, or 1.08 and 1.07 for $|\mathcal{B}| = 20$, which is consistent with our empirical findings. The effect of $|\mathcal{B}|$ is less pronounced, but still present for the uniform price auction.

    \item \textbf{Competitive Ratio grows in $N$ despite Regret Shrinking in $N$.} We observe that the regret is generally smaller for larger $N$. One justification for this phenomenon is that individual bidders have a smaller impact on the market clearing price, and thus, have less of an incentive to strategically shade their bids leading to faster convergence towards equilibrium behavior. Taking a broader view, this trend suggests a sub-linear dependence of social regret on $N$, offering a possible improvement over recent literature characterizing improved convergence rates for multi-agent games with more players \cite{syrgkanis2015fast, foster2016robustnessconvergence}. Perhaps more interestingly, the competitive ratio \textit{grows} in $N$, despite its similarity to regret, with the exception that it is normalized by the hindsight optimal utility. This suggests that this hindsight optimal utility---i.e., the optimal consumer surplus---decreases as $N$ increases, showing that the bidders become more of price takers as $N$ increases.

    \item \textbf{Run-time Increases in $M, N, |\mathcal{B}|$.} The run-time of our algorithms is approximately linear in each of $M$ and $|\mathcal{B}|$, as predicted in Section~\ref{sec: path kernels regret}. The algorithm for uniform price is approximately linear in $M$ and quadratic in $|\mathcal{B}|$, consistent with \cite{brânzei2023online}. As there are $N$ agents running the learning algorithm in each trial, the run-time scales linearly in $N$.
\end{enumerate}
}
\section{Extension: Time Varying Valuations}

\label{sec: time varying}
{\color{black}
We extend Algorithms \ref{alg: Decoupled Exponential Weights} and \ref{alg: Decoupled Exponential Weights - Path Kernels} to the time varying valuations setting. In particular, we assume that the valuations $\bm{v}$ are no longer fixed, and instead, in every round $t$, $\bm{v}^\nround$ is independently  drawn from some known distribution $F_{\bm{v}}$ with discrete, finite support $\mathcal{V}$. 
This contextual setting requires a  stronger  benchmark oracle in comparison to our original setup with a fixed valuation. The new benchmark oracle, which we will formalize shortly, possesses knowledge of the hindsight optimal bid vector for each context. That is, under this benchmark, we have the optimal mapping from any context (valuation vector) to an action (bid vector). 
Consequently, our current definitions of $\textsc{Regret}$ and $\textsc{Regret}_{\mathcal{B}}$ need to be updated to accommodate these contextual factors:
\begin{align}
\tag{Continuous Contextual Regret}
    \textsc{Regret}(F_{\bm{v}}) = \max_{\bm{b}: \mathcal{V} \to [0, 1]^{+\Nitem}} \sum_{\nround=1}^\Nround \mathbb{E}_{\bm{v} \sim F_{\bm{v}}}[\mu^\nround_n(\bm{b}(\bm{v}); \bm{v})] - \mathbb{E}\left[\sum_{\nround=1}^\Nround \mu^\nround_n(\bm{b}^\nround; \bm{v}^\nround)\right]\,.
\end{align}
Here, $\mu^\nround_n(\bm{b}; \bm{v})$ denotes the utility of bidder $n$ by submitting bid vector $\bm{b}$ with valuations $\bm{v}$ at round $t$ where the competing bids are $\bm{b}_{-}^{t}$. Observe that in the benchmark of $ \textsc{Regret}(F_{\bm{v}})$, i.e., $\max_{\bm{b}: \mathcal{V} \to [0, 1]^{+\Nitem}} \sum_{\nround=1}^\Nround \mathbb{E}_{\bm{v} \sim F_{\bm{v}}}[\mu^\nround_n(\bm{b}(\bm{v}); \bm{v})]$, we abuse notation and define valuation-to-bid vector mapping $\bm{b}: \mathcal{V} \to [0, 1]^{+M}$.  
We have an equivalent definition of discretized contextual regret:
\begin{align}
\tag{Discretized Contextual Regret}
    \textsc{Regret}_\mathcal{B}(F_{\bm{v}}) = \max_{\bm{b}:\mathcal {V}\to
    \mathcal{B}^{+\Nitem}} \sum_{\nround=1}^\Nround \mathbb{E}_{\bm{v} \sim F_{\bm{v}}}[\mu^\nround_n(\bm{b}(\bm{v}); \bm{v})] - \mathbb{E}\left[\sum_{\nround=1}^\Nround \mu^\nround_n(\bm{b}^\nround; \bm{v}^\nround)\right]\,.
\end{align}
An agent's goal is to minimize their contextual regret with respect to their valuation distribution $F_{\bm{v}}$. Using naive contextual bandit algorithms would lead to a large regret, as the regret of these algorithms  scales  with the square root of the number of contexts. However, we make an observation that we have \emph{complete cross-learning} over these contexts as in \cite{ContextBanditsCrossLearning2019}. That is,
 whenever the agents chooses bid $\bm b$ in round $t$ while having context/value $\bm v$ and receives reward $\mu_{n}^{t}({\bm b}; \mathbf v)$, they also learn the value of $\mu_{n}^{t}({\bm b}; \mathbf v')$ for any  contexts/values $\bm v'$
 This is because of the functional form of $\mu_n(\bm{b}; {\bm v}) = \sum_{\nitem=1}^{x_n(\bm{b}, {\bm b}_{-})} (v_{n, \nitem} - b_{n, \nitem})$.

As such, we borrow  from the results described in \cite{ContextBanditsCrossLearning2019}; specifically those explaining the cross-learning-across-contexts generalizations of the $\textsc{EXP3}$ algorithm in the stochastic contexts (valuations) and adversarial rewards setting (adversarial competing bids). 
We assume that the agent has access to their valuation distribution. Moreover,  as stated earlier, we assume that the support of this valuation distribution is finite; i.e., $|\mathcal{V}| < \infty$. This scenario occurs often in practice where bidders' valuations depend naturally on some natural events. For example, investors may prescribe a `low' or `high' value to certain assets depending on various market indices. 

We generalize the $\textsc{EXP3-CL}$ algorithm described in \cite{ContextBanditsCrossLearning2019} to our PAB setting, specifically Algorithm~\ref{alg: Decoupled Exponential Weights - Path Kernels}, and  achieve exactly the same regret rates as our non-contextual variants, albeit requiring an additional $O(|\mathcal{V}|)$ factor of memory and computation.

In order to make the generalization more clear, at a high level, the $\textsc{EXP3-CL}$ algorithm on a set of $K$ arms and $C$ contexts with full cross-learning constructs a reward estimator $\hat{r}(k; c) = \frac{r(k; c)}{\sum_{c}\prob(c)\prob(k^t = k | c^t = c)}\textbf{1}_{k^t = k}$ for each arm $k$ and context $c$ pair. Here, the term $\sum_{c}\prob(c)\prob(k^t = k | c^t = c)$ is the expected probability that arm $k^t = k$ was selected under context $c^t = c$, where in the summation we take expectation over the stochasticity of contexts $c$. This estimator mirrors that of standard $\textsc{EXP3}$ using the IPW estimator, except that the IPW is averaged over the context distribution. 

To apply this to our setting, we wish to mimic the behavior of the $\textsc{EXP3-CL}$ algorithm with our decoupled exponential weights algorithm. This can be done by running the $\textsc{EXP3-CL}$ estimator on all of the nodes $b \in \mathcal{B}$ within each layer $m \in [M]$. In particular, we use the following estimator  $\widehat{w}_m^t(b; \bm{v}) = 1 - \frac{1 - w_m^t(b; \bm{v})}{Q_m^t(b)} \textbf{1}_{b_m^t = b}$, where the  normalizer $Q_m^t(b) = \sum_{\bm{v} \in \mathcal{V}} \prob(\bm{v}^t = \bm{v}) q_m^t(b; \bm{v})$ in this estimator 
is the expected probability of selecting bid $b$, where the expectation is taken with respect to all valuation vectors $\bm{v} \in \mathcal{V}$. This procedure, formally described in Algorithm~\ref{alg: Decoupled Exponential Weights - Time Varying Known Finite} and analyzed in the online appendix (arXiv:2307.15193v3), yields the following regret upper bound:

\begin{theorem}[Time Varying Valuations - Decoupled Exponential Weights] \label{thm: time varying known finite}
    Under bandit feedback (resp. full information feedback), Algorithm~\ref{alg: Decoupled Exponential Weights - Time Varying Known Finite}, with appropriately chosen $\eta$, achieves contextual continuous regret $\textsc{Regret}(F_{\bm{v}})$ of order $O(\Nitem^\frac{4}{3} \Nround^{\frac{2}{3}} \sqrt{\log \Nround})$ (resp. $O(\Nitem^\frac{3}{2} \sqrt{\Nround \log \Nround})$ with total time time and space complexity polynomial in $M$, $|\mathcal{B}|$, $|\mathcal{V}|$, and $\Nround$.
\end{theorem}

\begin{algorithm}[t]
\footnotesize
	\KwIn{Learning rate $0 < \eta < \frac{1}{M}$, Valuation Distribution $F_{\bm{v}}$}
	\KwOut{The aggregate utility $\sum_{\nround=1}^\Nround \mu_n^\nround(\bm{b}^{\nround}; \bm{v}^\nround)$}
	$\widehat{W}_\nitem^0(b; \bm{v}) \gets 0$ for all $\nitem \in [\Nitem], b \in \mathcal{B}, \bm{v} \in \mathcal{V}$ such that $b \leq v_m$; else $\widehat{W}_\nitem^0(b; \bm{v}) \gets -\infty$.\; 
	\For{$\nround \in [1,\ldots,\Nround]$:}{
            \textbf{Observe Valuation Vector $\bm{v}^t \sim F_{\bm{v}}$}\;
            $b_{0}^t \gets \max \mathcal B$, and $\widehat{S}_{M+1}^t (\min \mathcal{B}; \bm{v}^t)=1$ for any $t\in[T]$\;
            \textbf{Recursively Computing Exponentially Weighted Partial Utilities $\bm{S}^t$}\;
            \textbf{for} $m \in [M,\ldots,1], b \in \mathcal{B}: \widehat{S}^t_\nitem(b; \bm{v}^t) \gets \exp(\eta \widehat{W}_\nitem^\nround(b; \bm{v}^t)) \sum_{b' \leq b} \widehat{S}_{\nitem + 1}^\nround(b'; \bm{v}^t)$ \hspace{0mm} $\backslash \backslash$ $\textsc{Compute}-\widehat{S}_\nitem$\;
        \textbf{Determining the Bid Vector $\bm{b}^\nround$ Recursively}\;
        \textbf{for} $m \in [1,\ldots,M], b \leq b_{m-1}^t: b_\nitem^\nround \gets b$ with probability $\frac{\widehat{S}^t_\nitem(b; \bm{v}^t)}{\sum_{b' \leq b_{\nitem-1}^t} \widehat{S}^t_{\nitem}(b'; \bm{v}^t)}; $ \hspace{1mm} $\backslash \backslash$ $\textsc{Sample}-\bm{b}$\;
        Observe $\bm{b}^{\nround}_-$ and receive reward $\mu_n^\nround(\bm{b}^{\nround}; \bm{v}^t)$\;
        $Q_m^t(b) \gets 0$ for all $m \in [M], b \in \mathcal{B}$\;
        \For{$\bm{v} \in \mathcal{V}$:}{
            \textbf{Recursively Computing Probability Measure $\bm{q}$ Under $\bm{v}\in {\mathcal V}$}\;
            $\widehat{S}^t_{M+1}(b; \bm{v}) \gets 1$ for all $m \in [M], b \in \mathcal{B}$\;
            \textbf{for} $m \in [M,\ldots,1], b \in \mathcal{B}: \widehat{S}^t_\nitem(b; \bm{v}) \gets \exp(\eta \widehat{W}_\nitem^\nround(b; \bm{v})) \sum_{b' \leq b} \widehat{S}_{\nitem + 1}^\nround(b'; \bm{v})$\;
            $q^t_1(b; \bm{v}) \gets \frac{\widehat{S}^\nround_m(b; \bm{v})}{\sum_{b' \in \mathcal{B}} \widehat{S}^\nround_m(b'; \bm{v})}$ for all $b \in \mathcal{B}$\;
            \textbf{for} $m \in [2,\ldots,M], b \in \mathcal{B}: q_\nitem^\nround(b; \bm{v}) \gets \sum_{b' \geq b} \frac{q_{\nitem-1}^t(b'; \bm{v})\widehat{S}^\nround_{\nitem}(b; \bm{v})}{\sum_{b" \geq b'} \widehat{S}^\nround_\nitem(b"; \bm{v})}$ for all $b \in \mathcal{B}$\;
            $Q_m^t(b) \gets Q_m^t(b) + \prob(\bm{v}^t = \bm{v})q_m^t(b; \bm{v})$
        }
        \textbf{Update Weight Estimates}\;
        \textbf{if} $\textsc{Bandit Feedback}$, \textbf{for} $m \in [M], b \in \mathcal{B}, \bm{v} \in \mathcal{V}$\; 
        $\widehat{W}^{\nround+1}_{\nitem}(b; \bm{v}) \gets \widehat{W}^{\nround}_{\nitem}(b; \bm{v}) + (1 - \frac{1 - (v - b)\textbf{1}_{b \geq b^t_m}}{Q_m^t(b)} \textbf{1}_{b^t_m = b})$ if $b \leq v$; else $\widehat{W}^{\nround+1}_{\nitem}(b; \bm{v}) \gets -\infty$\;
        \textbf{if} $\textsc{Full Information}$, \textbf{for} $m \in [M], b \in \mathcal{B}, \bm{v} \in \mathcal{V}$\;
        $\widehat{W}^{\nround+1}_{\nitem}(b; \bm{v}) \gets \widehat{W}^{\nround}_{\nitem}(b; \bm{v}) + (v - b)\textbf{1}_{b \geq b^t_m}$ if $b \leq v$; else $\widehat{W}^{\nround+1}_{\nitem}(b; \bm{v}) \gets -\infty$\;
        }
        \textbf{Return} $\sum_{\nround=1}^\Nround \mu_n^{\nround}(\bm{b}^{\nround}; \bm{v}^\nround)$
	\caption{\textsc{Decoupled $\textsc{EXP3-CL}$ - Time Varying Valuations}}
	\label{alg: Decoupled Exponential Weights - Time Varying Known Finite}
\end{algorithm}}
\section{Concluding Remarks}

We have provided low-regret learning algorithms for PAB auctions in the full information and bandit settings with corresponding polynomial time and space complexities. In particular, we utilize our DP formulation and its equivalent graph representation to decouple the utility associated with bidding $b_\nitem = b$ for all $\nitem \in [\Nitem], b \in \mathcal{B}$. We derived two algorithms, one that mimics the exponential weights algorithm and another based on OMD, both of which allowed us to achieve polynomial (in $\Nitem$, $|\mathcal{B}|$, and $\Nround$) regret upper bounds, as well as time and space complexities, despite the combinatorially large bid space.
Furthermore, our experimental simulations highlight the convergence of winning bids in PAB auctions, coupled with higher revenue generation compared to uniform price auctions. However, we note a slight lag in welfare compared to uniform price auctions. Our findings provide actionable guidance for auction design and bidder strategy in multi-unit auctions, contributing to a deeper understanding of bidding dynamics and equilibria.

There are several intriguing avenues for future research that can be explored based on the current work. A promising direction is to leverage the structure induced by bid monotonicity in PAB auctions. 
Recent advancements in a simpler single-unit setting have demonstrated the efficacy of cross-learning between bids under certain feedback structures \citep{OptimalNoRegretFPA2020, LearningBidOptimallyAdversarialFPA2020}. It would be intriguing to apply cross-learning techniques in our multi-unit setting and can explore whether they can enhance our regret bounds. Furthermore, inspired by our numerical results---where we show that the winning bids in PAB market dynamics converge to the same value---we can explore the design of online learning algorithms for the setting where bidders are restricted to a simplified bidding interface, wherein they are only allowed to submit a single  price and quantity for the units demanded rather than an entire vector of bids. 
Another potential research direction is to study conditions under which last iterate convergence for learning in PAB holds. As observed in Figure \ref{fig: bid cycling main body}, the learners can converge to a PNE, though proving this was only recently done by \cite{Deng_2022} for the class of \textit{mean-based} algorithms (defined in \cite{braverman2018noregretbuyer}) for the single unit, first price auction under certain conditions.

\section{Acknowledgments}

R.G. and N.G. were supported in part by the Young Investigator Program (YIP) Award from the Office of Naval Research (ONR) N00014-21-1-2776 and the MIT Research Support Award.

\bibliographystyle{ACM-Reference-Format}
\footnotesize{
\bibliography{ref.bib}
}
\newpage
\newpage 

\section{Appendix}
{\color{black}\subsection{Proof of Lemma \ref{lem: PNE uniform bidding}}\label{sec:proof:lem: PNE uniform bidding}
We prove this by contradiction. Suppose buyer \(n\) places any bid strictly greater than \(b_{(M)} + \delta\). Assume the bids of all other bidders are fixed. Any bid \(b\) by bidder \(n\) that equals or exceeds \(b_{(M)} + \delta\) still secures allocation, as \(b \geq b_{(M)} + \delta > b_{(M)}\), irrespective of tie-breaking. Therefore, bidder \(n\) can reduce all such bids to exactly \(b_{(M)} + \delta\) without losing any items, thereby lowering their total payment. Consequently, the largest bid cannot exceed \(b_{(M)} + \delta\).

\subsection{Proof of Lemma \ref{lem: near-uniform optimal bidding}}\label{sec:proof:lem: near-uniform optimal bidding}
    Consider any bid vector $\bm{b} \in \mathcal{B}^{+M}$. Assuming that bidder $n$'s  is allocated  $m$ under $\bm{b}$, then $b_m \geq \tilde{b}_{n,m}$, as $\tilde{b}_{n,m}$ is the smallest bid required for bidder $n$ to win the $m$'th item. Define $\bm{b}': b'_j = \min(\tilde{b}_{n,m}, b_j)$ which is $\bm{b}$ except setting the first $m$ bids to $\tilde{b}_m$.
    By monotonicity of $\bm{b}_{-n}$ and $\bm{b}$---i.e., $b_{-1} \leq \ldots \leq b_{-m} \leq \tilde{b}_{n,m} \leq b_m \leq \ldots \leq b_1$---the allocation is the same under $\bm{b}$ and under $\bm{b}': b'_j = \min(\tilde{b}_{n,m}, b_j)$. As the allocations are equal and the payment is smaller under $\bm{b}'$, which belongs to the set $\{\bm{b} \in \mathcal{B}^{+M}: b_j = (\tilde{b}_{n,m}) \forall j \in [m], b_{j} \leq v_j \forall j > m\}$, then $\bm{b}'$ must yield at least as much utility as $\bm{b}$. Thus, any optimal bid vector that is allocated $m$ units, conditioning on $\bm{b}_{-n}$ must have uniform $m$ highest bids at value $\tilde{b}_{n,m}$. Because we prohibit overbidding, we only consider $\bm{b}$ that win $m$ items for $m$ satisfying $\tilde{b}_{n, m} < v_m$.
    
{\color{black}
\subsection{Assumptions in Theorem \ref{thm: PNE existence}} \label{sec:discuss:assumption}
The validity of the PNE characterization in Theorem \ref{thm: PNE existence} depends on the premise that there is sufficient competition, ensuring that the clearing price $c$ remains constant even upon removing any single bidder. However, this assumption may not always hold true. To illustrate this, we initially provide an example where the assumption is met. However, a minor modification to the example leads to the failure of this assumption. In Figure~\ref{fig: bid cycling main body}, we illustrate how this impacts the convergence of our learning algorithms' final iterations.

In this figure, we plot the bid values over the course of the market dynamics induced by our full information decoupled exponential weights algorithm (Algorithm~\ref{alg: Decoupled Exponential Weights}) with $N=3$, $M=4$, $|\mathcal{B}|=10$, and $T=10^5$. In the left figure, we let $\bm{v}_1 = \bm{v}_2 = \bm{v}_3 = [1-\epsilon, 1-\epsilon]$, and thus there exists a PNE via satisfying the $c = c_{-n} = \lfloor 1-\epsilon \rfloor_{\delta=0.1} = 0.9$ condition from Theorem~\ref{thm: PNE existence}. Moreover, this PNE is characterized by all bidders submitting bids of $0.9$ for the first two units, which is precisely what the market dynamics converge to.\footnote{{We do not present theoretical guarantees on last iterate convergence for PAB. For a discussion in the $M=1$ unit setting, please refer to \cite{Deng_2022}.}} In the right figure, we assume that  bidder one only demands one unit at $1-\epsilon$, instead of two; that is $\bm v_1= [1-\epsilon, 0]$. The value of other bidders remain the same. Here, we do not satisfy the $c = c_{-n}$ condition as $ c_{-1} = 1-\epsilon \neq 0 = c_{-2} = c_{-3}$. In fact, one can verify that there exists no PNE in this auction and we observe cyclic bidding behavior from the 2nd and 3rd bidders.

To explain why the learning dynamics do not converge in cases where PNE is absent, we outline several observations:

\begin{enumerate}
    \item In the auction in Figure~\ref{fig: bid cycling main body} with no PNE, $\bm{v}_1 = [1-\epsilon, 0]$, so bidder 1 can submit at most one non-zero bid. Similarly, $\bm{v}_2 = \bm{v}_3 = [1-\epsilon, 1-\epsilon]$, so bidders 2 and 3 can submit at most two non-zero bids. Bidder 2 and bidder 3 can guarantee an allocation of at least 1 since among all other bidders, the demand is at most 3. However, they can guarantee an allocation of 2 if they can slightly outbid the 2nd largest bid among the other two bidders.
    \item In Figure~\ref{fig: bid cycling main body}, we observe that bidder 1's bid converges to 0.5. In contrast, bidder 2's bids cycle between 0.1 and 0.5, and bidder 3's bids cycle between 0 and 0.4. What happens is that bidder 3 will try to win one unit cheaply by bidding for both units at price $0$. Bidder 2 realizes that they can slightly outbid them at a price of 0.1, yielding utility $1.8 - 2\epsilon$. Then bidder 3 matches their bid, and because they are given tie-break priority, they win two units, yielding utility $1.8 - 2\epsilon$. Bidder 2, now only winning one unit, realizes that they can win 2 units by increasing their bid to 0.2, which bidder 3 then matches, and so on.
    \item Once bidder 2 submits two bids of 0.5, which yields a utility of $1 - 2\epsilon$, bidder 3 can either submit two bids of 0.5, which yields an allocation of 2 and a utility of $1 - 2\epsilon$, or they can submit a bid of 0, which only yields an allocation of 1 but a utility of $1 - \epsilon$. Thus, they submit bids of 0 as the best response to bidder 1 submitting a bid of 0 and bidder 2 submitting two bids of 0.5, after which the cycle repeats.
\end{enumerate}

\begin{figure}
    \centering
    \includegraphics[scale=0.5, trim={0, 0, 0, 0}, clip]{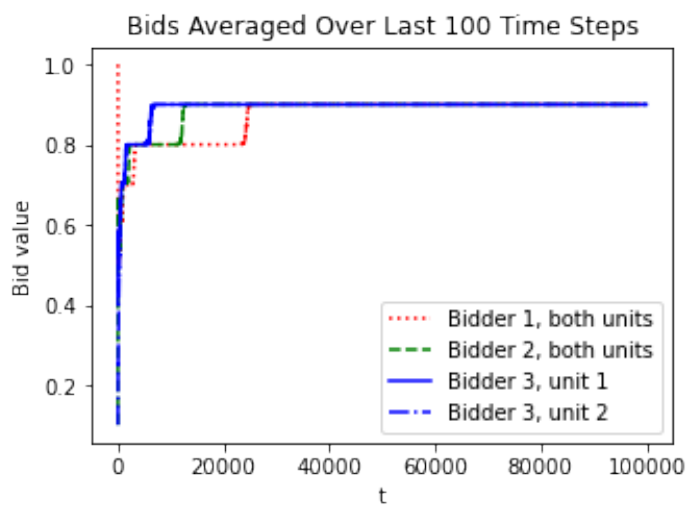}
    \includegraphics[scale=0.5, trim={0, 0, 0, 0}, clip]{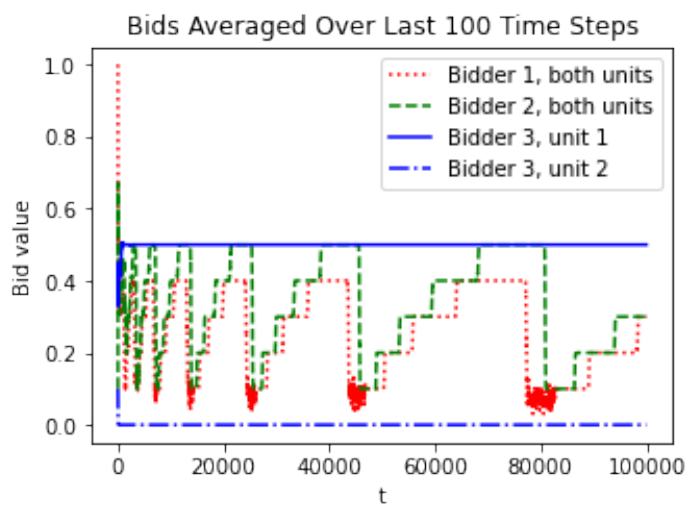}
    \caption{\textbf{Market Dynamics (Non) Last Iterate Convergence. } We plot the bid values over the course of the market dynamics induced by our full information decoupled exponential weights algorithm (Algorithm~\ref{alg: Decoupled Exponential Weights}) with $N=3,M=4,|\mathcal{B}|=10,T=10^5$. In the left figure, we let $\bm{v}_1 = \bm{v}_2 = \bm{v}_3 = [1-\epsilon, 1-\epsilon]$, and thus there exists a PNE via satisfying the $c = c_{-n} = \lfloor 1-\epsilon \rfloor_{\delta=0.1} = 0.9$ condition from Theorem~\ref{thm: PNE existence}. Moreover, this PNE is characterized by all bidders submitting bids of $0.9$ for the first two units, which is precisely what the market dynamics converge to. In the right figure, we assume that bidder $3$ demands one unit at $1-\epsilon$, instead of two. The value of other bidders remain the same. Here, we do not satisfy the $c = c_{-n}$ condition as $c = c_{-1} = 1-\epsilon \neq 0 = c_{-2} = c_{-3}$. In fact, one can verify that there exists no PNE in this auction and we observe cyclic bidding behavior from the 2nd and 3rd bidders.}
    \label{fig: bid cycling main body}
\end{figure}

Despite this fragility of the $c = c_{-n}$ assumption in Theorem~\ref{thm: PNE existence}, this is merely a sufficient condition and not necessary for the existence of the PNE. At its core, this $c = c_{-n}$ assumption reflects the requirement that strategic bid shading cannot offset a decrease in allocation with a decrease in payment. For example, if we instead set $v_{1,2} = 0.8+\epsilon$ in Figure~\ref{fig: bid cycling main body}, a PNE still exists at $\bm{b}_1 = [0.9, 0.8]$ and $\bm{b}_2 = \bm{b}_3 = [0.8, 0.8]$ despite $ c_{-1} = 1 -\epsilon \neq 0.8 + \epsilon = c_{-2} = c_{-3}$. For practical understanding, PNE existence is satisfied when bidders are price-takers with minimal influence over the market price---a condition typically met in markets where the demand each bidder $n$ places on units $M$ is substantially smaller than the overall supply $\overline{M}$, and the total supply is less than the aggregate demand, i.e., $\sum_{n \in [N]} \overline{M}_n \gg M$.
}
}

{\color{black}\subsection{Proof of Lemma \ref{lem: CCE non uniform winning}}
    Consider the following $M=2$ item, $N=2$ bidder auction where bidder 1 and 2's valuations are $[1 + \epsilon, 1 + \epsilon], [1, 0]$ respectively, for some small $\epsilon > 0$. Let $\mathcal{B} = \{0, \frac{1}{|\mathcal{B}|}, \frac{2}{|\mathcal{B}|},\ldots, 1\}$ be an even discretization of $[0,1]$ with the discretization factor of $\delta =  \frac{1}{|\mathcal{B}|}$,  
    and assume that ties are broken in favor of bidder 2. 
    In this setting, we can show the existence of CCEs where the winning bids are \textit{never} uniform, even up to a discretization factor.
    
    To that end, we solve a linear program using the linear constraints formulation of CCEs and maximize the probability of drawing a joint action profile with non-uniform winning bids. We first define $b_{(m)} = b_{(m)}(\bm{b}_1,\ldots,\bm{b}_N)$ as the $m$'th largest bid among bid vectors $\bm{b}_1,\ldots,\bm{b}_N$. Now, we maximize the probability that the gap between $b_{(1)}$ and $b_{(M)}$ is more than one discretization factor $\delta = \frac{1}{|\mathcal{B}|}$ apart by maximizing the objective function $\max \sum_{\bm{b}_1,\ldots,\bm{b}_n} \bm{p}(\bm{b}_1,\ldots,\bm{b}_n) \mathbf{1}_{b_{(1)} > b_{(M)} + \delta}$. Under CCE constraints, the objective function is equal to 1 for all $|\mathcal{B}| \in [5, 10]$, meaning that there exist CCEs entirely supported on joint bidding profiles which have winning bids more than $\delta$ apart. Under CE constraints, the objective function is increasing in $|\mathcal{B}|$ but stabilizes around $0.6$ at $|\mathcal{B}| = 8$. In other words,  when $|\mathcal{B}|$ grows large, there exists a CCE (resp. CE) under which there is a 1.0 (resp 0.6) probability over joint strategy profiles that violate the winning bid uniformity. $\blacksquare$
    }

\subsection{Proof of Theorem \ref{thm:offline}: Offline Bid Optimization Algorithm}
\label{sec: offline proof}

We give a proof of correctness of the offline bid optimization algorithm used to compute the hindsight optimal bid vector across $T$ rounds of PAB auctions. Our proof shows that the variables $U_m(b)$ are path weights of the optimal partial bid vector with weights $W_m^{T+1}(b)$. Thus, $U_1(b)$ is the optimal bid vector and $b^*_m$'s can be used to back out the optimal bid vector recursively in polynomial time.

\begin{proof}{Proof of Theorem \ref{thm:offline}}
   By definition, we have that $U_\nitem(b)$ is given by:
    \begin{align*}
        \max_{ b\ge b_\nitem\geq\ldots\geq b_\Nitem} \sum_{\nitem' = \nitem}^\Nitem W^{\Nround+1}_{\nitem'}(b_{\nitem'}) = \max_{ b\ge b_\nitem\geq\ldots\geq b_\Nitem}  W^{\Nround+1}_{\nitem}(b_\nitem ) +  \sum_{\nitem' = \nitem+1}^\Nitem W^{\Nround+1}_{\nitem'}(b_{\nitem'}) = \max_{ b' \leq b} W^{\Nround+1}_\nitem(b') + U_{\nitem+1}(b')\,.
    \end{align*}
    Since we have that $U_{\Nitem}(b ) = W^{\Nround+1}_{\Nitem} (b )$ trivially correct from the base case, and the optimality of $U_\nitem(b)$ follows from induction. Consequently, optimality of $b_\nitem^*$ follows from induction. The base case trivially holds as $b_1^* = \text{argmax}_{b \in \mathcal{B}} U_1(b)$. The recursive case also follows straightforwardly by definition of $bb^*_\nitem = \text{argmax}_{b \leq b^*_{\nitem-1}} U_\nitem(b)$. As $b^*_{\nitem-1}$ was optimal by the induction hypothesis, then $b_\nitem^*$ must also be optimal.
    
    We finish this proof by discussing the time and space complexity of Algorithm \ref{alg: Offline Full}.  {Table $U$ is of size $O(\Nitem |\mathcal{B}|)$, with each entry requiring taking a maximum over $O(|\mathcal{B}|)$ terms, yielding time and space complexities of $O(\Nitem |\mathcal{B}|^2)$ and $O(\Nitem|\mathcal{B}|)$ respectively.}

\end{proof}

\subsection{Proof of Theorems \ref{thm:full} and \ref{thm:decoupled exp - bandit feedback}: Decoupled Exponential Weights Algorithm}
\label{sec: path kernels regret}

\begin{proof}{}
    We give the proofs of correctness, complexity analysis, and regret analysis for the decoupled exponential weights algorithms for both the full information (Algorithm~\ref{alg: Decoupled Exponential Weights}) and the bandit setting (Algorithm~\ref{alg: Decoupled Exponential Weights - Path Kernels}). Our proof comes in 5 parts. We first prove correctness of the bandit version of our algorithm. In particular, we show that defining the node and bid vector utility estimates to be $\widehat{w}_m^\nround(b) = 1 - \frac{1-(v_m-b)1_{b \geq b_{-m}^\nround}}{q_m^t(b)}1_{b = b^\nround_m}$ and $\widehat{\mu}^\nround(\bm{b}) = \sum_{\nitem=1}^\Nitem \widehat{w}_\nitem^\nround(b_m)$, our algorithm samples bid vector $\bm{b}^\nround$ with probability proportional to $\exp(\eta \sum_{\tau=1}^{t-1} \widehat{\mu}^\tau(\bm{b}))$ via the same recursive sampling procedure as in Algorithm~\ref{alg: Decoupled Exponential Weights}. In the second part and third parts, we derive a corresponding regret upper bound and obtain the time and space complexities of our algorithm with bandit feedback. In the fourth part, we optimize the continuous regret w.r.t. the selection of $\mathcal{B}$. In the fifth part, we show how to extend our algorithm and results to the full information setting.

    \textbf{Part 1: Algorithm Correctness.} In this part of the proof, we argue that our choice of estimator is unbiased and that Algorithm~\ref{alg: Decoupled Exponential Weights - Path Kernels} samples bid vectors with the same probability that the exponential weights algorithm would have, given the same node utility estimates $\widehat{w}_m^t(b)$. To show unbiasedness of $\widehat{w}_m^t(b)$, we have:
    \begin{align*}
        \mathbb{E}\left[\widehat{w}_m^{t}(b)\right] = \mathbb{E}\left[1 - \frac{1-w^t_m(b)}{q^t_m(b)} \textbf{1}_{b^t_m = b}\right] = \mathbb{E}\left[1 - \frac{\textbf{1}_{b^t_m = b}}{q^t_m(b)} + \frac{\textbf{1}_{b^t_m = b} \cdot w^t_m(b)}{q^t_m(b)}\right] = w^t_m(b)\,.
    \end{align*}
    Now, it remains to show that our sampling procedure $\textsc{Sample}-\bm{b}$ w.r.t. $\widehat{S}^\nround_m$ indeed samples bid vectors $\bm{b}$ with the same probability as the exponential weights algorithm under weights $\widehat{\mu}^\nround_n(\bm{b})$. In particular, we want to show that our algorithm samples bid vectors $\bm{b}^\nround$ with probability proportional to $\exp(\eta \sum_{\tau=1}^{t-1}\widehat{\mu}_m^\tau(b_m))$ for any $m\in [M]$. This follows from analyzing the dynamic programming variables that represent the sum of exponentiated (estimated) partial bid vector utilities, $\widehat{S}$.

    In exponential weights, the bidder selects at round $\nround+1$ some action $\bm{b}$ with probability $P^\nround(\bm{b})$ proportional to $\sum_{\tau=1}^{\nround} \widehat{\mu}_n^{\nround}(b)$. Using our representation of $\widehat{\mu}_n^{\nround}(\bm{b})$ as a function $\widehat{w}_\nitem^{\tau}(b)$, we have:
    \begin{align*}
        P^{\nround}(\bm{b}) = \frac{\exp(\eta \sum_{\tau=1}^{t-1} \widehat{\mu}_n^{\tau}(\bm{b}))}{\sum_{\bm{b}' \in \mathcal{B}^{+\Nitem}} \exp(\eta \sum_{\tau=1}^{t-1} \widehat{\mu}_n^\tau(\bm{b}'))} = \frac{\exp(\eta \sum_{\nitem=1}^\Nitem \widehat{W}_\nitem^\nround(b_\nitem))}{\sum_{\bm{b}' \in \mathcal{B}^{+\Nitem}} \exp(\eta \sum_{\nitem=1}^\Nitem \widehat{W}_\nitem^\nround(b'_\nitem))}\,.
    \end{align*}
    Hence, we wish to construct a sampler that samples $\bm{b}$ with the above probability. Defining $b_0 = \max_{b \in \mathcal{B}} b$, we begin by decomposing the denominator as follows:
    \begin{align*}
        &\sum_{\bm{b} \in \mathcal{B}^{+\Nitem}} \exp(\eta \sum_{\nitem=1}^\Nitem \widehat{W}_\nitem^\nround(b_\nitem)) = \sum_{b_1 \in \mathcal{B}, b_1 \leq b_0} \sum_{b_2 \in \mathcal{B}, b_2 \leq b_1} \ldots \sum_{b_\Nitem \in \mathcal{B}, b_\Nitem \leq b_{\Nitem-1}} \exp(\eta \sum_{\nitem=1}^\Nitem \widehat{W}_\nitem^\nround(b_\nitem))\\
        &= \sum_{b_1 \in \mathcal{B}, b_1 \leq b_0} \exp(\eta \widehat{W}_1^\nround(b_1)) \sum_{b_2 \in \mathcal{B}, b_2 \leq b_1}\exp(\eta \widehat{W}_2^\nround(b_2)) \ldots \sum_{b_\Nitem \in \mathcal{B}, b_\Nitem \leq b_{\Nitem-1}} \exp(\eta \widehat{W}_\Nitem^\nround(b_\Nitem))\,.
    \end{align*}
   Recall a key object $\widehat{S}^t_\nitem(b)$, which is the sum of exponentially weighted utilities of partial bid vectors $\bm{b}'_{\nitem:\Nitem} \in \mathcal{B}^{+(\Nitem - \nitem + 1)}$ over slots $\nitem,\ldots,\Nitem$ subject to $b_\nitem = b$.
    \begin{align*}
        \widehat{S}^t_\nitem(b) = \exp(\eta \widehat{W}_\nitem^\nround(b)) \sum_{\bm{b}'_{\nitem+1:\Nitem}:b'_{\nitem+1} \leq b'_\nitem = b} \exp(\eta \sum_{\nitem'=\nitem+1}^\Nitem \widehat{W}_{\nitem'}^\nround(b'_{\nitem'})) = \exp(\eta \widehat{W}_\nitem^\nround(b)) \sum_{b' \in \mathcal{B}; b' \leq b} \widehat{S}^t_{\nitem+1}(b')\,.
    \end{align*}
    With the trivial base case $\widehat{S}^t_\Nitem(b) = \exp(\eta \widehat{W}_\Nitem^\nround(b))$, we can recover all of the exponentially weighted partial utilities $\{\widehat{S}^t_{\nitem}(b)\}_{\nitem \in [\Nitem], b \in \mathcal{B}}$ given $\bm{W}^\nround$. Once we have computed $\{\widehat{S}^t_{\nitem}(b)\}_{\nitem \in [\Nitem], b \in \mathcal{B}}$, we can sample $\bm{b}$ according to its exponentially weighted utility $\exp(\eta \widehat{\mu}^{\nround}_\nitem(\bm{b}))$ by sequentially sampling each $b_1,\ldots,b_\Nitem$.
    
    Let $P_{\textsc{D}}^t(\bm{b})$ be the probability that our Algorithm \ref{alg: Decoupled Exponential Weights - Path Kernels} returns bid vector $\bm{b} \in \mathcal{B}^{+\Nitem}$ in round $t$. Recall that we sample $\bm{b}$ by setting $b_\nitem^\nround$ to $b \in \mathcal B, b \le b_{m-1}^t$ with probability $ \frac{\widehat{S}^t_\nitem(b)}{\sum_{b' \leq b_{\nitem-1}^t} \widehat{S}^t_{\nitem}(b')}$. Hence, the probability of selecting $\bm{b}$ is the product of $\nitem$ conditional probability mass functions (pmf's) and we have 
    \begin{align*}
        P_{\textsc{D}}^\nround(\bm{b}) = \prod_{\nitem=1}^\Nitem \frac{\widehat{S}^t_\nitem(b_\nitem)}{\sum_{b' \leq b_{\nitem-1}} \widehat{S}^t_{\nitem}(b')} = \left(\prod_{\nitem=1}^{M-1} \frac{\exp(\eta \widehat{W}_\nitem^\nround(b_\nitem)) \sum_{b \leq b_{\nitem}} \widehat{S}^t_{\nitem+1}(b)}{\sum_{b' \leq b_{\nitem-1}} \widehat{S}^t_{\nitem}(b')}\right)\left(\frac{\exp(\eta \widehat{W}_M^t(b_M))}{\sum_{b' \leq b_{M-1}} \widehat{S}_M^t(b')}\right)\,.
    \end{align*}
    Moving the $\exp(\eta \widehat{W}_m^t(b_m))$ outside of the product, we obtain:
    \begin{align*}
        P_{\textsc{D}}^\nround(\bm{b}) &= \left(\prod_{\nitem=1}^{\Nitem-1} \exp(\eta \widehat{W}_\nitem^\nround(b_\nitem)) \right) \left(\prod_{\nitem=1}^{M-1} \frac{\sum_{b \leq b_{\nitem}} \widehat{S}^t_{\nitem+1}(b)}{\sum_{b' \leq b_{\nitem-1}} \widehat{S}^t_{\nitem}(b')}\right)\left(\frac{\exp(\eta \widehat{W}_M^t(b_M))}{\sum_{b' \leq b_{M-1}} \widehat{S}_M^t(b')}\right)\\
        &= \left(\prod_{\nitem=1}^{\Nitem} \exp(\eta \widehat{W}_\nitem^\nround(b_\nitem)) \right) \left(\frac{\sum_{b \leq b_{M-1}} \widehat{S}^t_{M}(b)}{\sum_{b' \leq b_{0}} \widehat{S}^t_{1}(b')}\right)\left(\frac{1}{\sum_{b' \leq b_{M-1}} \widehat{S}_M^t(b')}\right) = \frac{\prod_{\nitem=1}^{\Nitem} \exp(\eta \widehat{W}_\nitem^\nround(b_\nitem))}{\sum_{b' \leq b_{0}} \widehat{S}_1^t(b')}\,.
    \end{align*}
    We now rearrange the last expression to recover the exponential weights algorithm probabilities:
    \begin{align*}
        P_{\textsc{D}}^\nround(\bm{b}) = \frac{\prod_{\nitem=1}^\Nitem \exp(\eta \widehat{W}_\nitem^\nround(b_\nitem))}{\sum_{b \leq b_0} S_1^t(b)} = \frac{\exp(\eta \sum_{\nitem=1}^\Nitem \widehat{W}_\nitem^\nround(b_\nitem))}{\sum_{\bm{b}' \in \mathcal{B}^{+\Nitem}} \exp(\eta \sum_{m=1}^M \widehat{W}_m^t(b'_m))} = P^\nround(\bm{b})\,.
    \end{align*}

    \textbf{Part 2: Regret Analysis.} We are now ready to derive the regret upper bound on Algorithm~\ref{alg: Decoupled Exponential Weights - Path Kernels}. First, we show that the bid vector utility estimators $\widehat{\mu}^\nround(\bm{b})$ are both unbiased and have a finite upper bound. To show the upper bound, we take expectation with respect to the bid vectors selected by our algorithm and observe that 
    \begin{align*}
        \mathbb{E}[\widehat{\mu}^\nround(\bm{b})] &= \sum_{m=1}^\Nitem \mathbb{E}[\widehat{w}_m^t(b_m)]= \sum_{m=1}^M \mathbb{E}\left[1 - \frac{1-(v_m-b_m)1_{b_m > b_{-m}^\nround}}{q_m^t(b_m)}1_{b_m = b^\nround_m}\right]
    \end{align*}
    As we are considering the expectation ex-post, keeping the $b_{-m}^t$'s fixed, we have independence between the two indicator functions and we obtain:
    \begin{align*}
        = \sum_{m=1}^M \mathbb{E}\left[1 - \frac{1-w^\nround_m(b_m)}{q_m^t(b_m)}1_{b_m = b^\nround_m}\right]= M - \sum_{m=1}^M \frac{1 - w^\nround_m(b_m)}{q_m^t(b_m)}\mathbb{E}\left[1_{b_m = b^\nround_m}\right]= \mu^\nround(\bm{b})\,.
    \end{align*}

    As for the finite upper bound, we have that $\widehat{\mu}^\nround(\bm{b}) = \sum_{m=1}^\Nitem \widehat{w}_m^t(b_m)$ is the sum over $M$ node utility estimators, each of which is upper bounded by 1. Hence, $\widehat{\mu}^\nround(\bm{b}) \leq \Nitem$ for all $\bm{b} \in \mathcal{B}^{+\Nitem}$. Now, we make the following claim:
    \begin{lemma}\label{lem:regret_bound}
        Let $\widehat{\mu}^\nround(\bm{b}) = \sum_{m=1}^M (1 - \frac{1-(v_m-b_m)1_{b_m > b_{-m}^\nround}}{q_m^t(b_m)}1_{b_m = b^\nround_m})$ be our bid vector utility estimate as discussed. Then, any algorithm which samples bid vectors $\bm{b}$ with probability proportional to $\exp(\eta \sum_{\tau=1}^{t-1} \widehat{\mu}^\nround(\bm{b}))$ at round $\nround$ for $\eta \leq \frac{1}{M}$ has regret upper bound
        \begin{align}
            \label{eq: ExpWeights Analysis 2}
            \textsc{Regret}_{\mathcal{B}} \lesssim \eta^{-1}M\log|\mathcal{B}| + \eta \sum_{\nround=1}^\Nround \sum_{\bm{b}} \prob(\bm{b}^\nround =\bm{b}) \mathbb{E}[(\sum_{m=1}^M \widehat{w}^\nround_m(b_m))^2]\,.
        \end{align}
    \end{lemma}
    \proof{Proof of Lemma \ref{lem:regret_bound}}
        We will closely follow the analysis of the $\textsc{Exp3}$ algorithm as presented in Chapter 11.4 of \cite{Lattimore2020}. In particular, we follow their regret analysis until Equation 11.13. Define $\Phi^t = \sum_{\bm{b} \in \mathcal{B}^{+\Nitem}} \exp(\eta \sum_{\tau=1}^{t} \widehat{\mu}^\tau(\bm{b}))$ to be the \textit{potential} at round $\nround$. As per our initial conditions in Algorithm~\ref{alg: Decoupled Exponential Weights - Path Kernels}, we have $\widehat{\mu}^0(\bm{b}) = 0$, and consequently, $\Phi^0 = |\mathcal{B}^{+\Nitem}|$. While it is not immediately apparent how the potentials $\Phi^\nround$ relate to the regret, we begin by upper bounding $\exp(\eta \sum_{t=1}^{T} \widehat{\mu}^t(\bm{b}))$ for a fixed $\bm{b}'$:
        \begin{align}
            \label{eq: Potentials}
            \exp(\eta \sum_{t=1}^{T} \widehat{\mu}^t(\bm{b}')) \leq \sum_{\bm{b} \in \mathcal{B}^{+\Nitem}} \exp(\eta \sum_{t=1}^{T} \widehat{\mu}^t(\bm{b})) = \Phi^T = \Phi^0 \prod_{t=1}^T \frac{\Phi^t}{\Phi^{t-1}}\,.
        \end{align}
        Now, we upper bound each $\frac{\Phi^t}{\Phi^{t-1}}$:
        \begin{align*}
            \frac{\Phi^t}{\Phi^{t-1}} = \sum_{\bm{b} \in \mathcal{B}^{+M}} \frac{\exp(\eta \sum_{\tau=1}^{t} \widehat{\mu}^\tau(\bm{b}))}{\Phi^{t-1}} = \sum_{\bm{b} \in \mathcal{B}^{+M}} \frac{\exp(\eta \sum_{\tau=1}^{t-1} \widehat{\mu}^\tau(\bm{b}))}{\Phi^{t-1}} \exp(\eta \widehat{\mu}^t(\bm{b})) = \sum_{\bm{b} \in \mathcal{B}^{+M}} \prob(\bm{b}^\nround = \bm{b}) \exp(\eta \widehat{\mu}^t(\bm{b}))\,,
        \end{align*}
        where in the last equality, we used the condition that our algorithm samples bid vector $\bm{b}$ with probability proportional to $\exp(\eta \sum_{\tau=1}^{t-1} \widehat{\mu}^\nround(\bm{b}))$ at round $\nround$. In order to continue the chain of inequalities, we note that for $\eta \leq \frac{1}{M}$, we have that the quantity $\eta \widehat{\mu}^t(\bm{b})$ is upper bounded by 1 as $\eta \widehat{\mu}^t(\bm{b}) \leq \eta M \leq 1$. In the first inequality, we used the fact that $\widehat{\mu}^\nround(\bm{b}) \leq \Nitem$. Now, we can apply the inequalities $\exp(x) \leq 1 + x + x^2$ and $1 + x \leq \exp(x)$ for all $x \leq 1$, with $x= \eta \widehat{\mu}^t(\bm{b})$ and $x = \eta \prob(\bm{b}^\nround = \bm{b}) \widehat{\mu}^t(\bm{b})$, respectively,  to obtain:
        \begin{align*}
            \frac{\Phi^t}{\Phi^{t-1}} &\leq \sum_{\bm{b} \in \mathcal{B}^{+M}} \prob(\bm{b}^\nround = \bm{b}) \exp(\eta \widehat{\mu}^t(\bm{b})) \leq \sum_{\bm{b} \in \mathcal{B}^{+M}} \prob(\bm{b}^\nround = \bm{b}) \left[1 + \eta\widehat{\mu}^t(\bm{b}) + \eta^2 \widehat{\mu}^t(\bm{b})^2 \right]\\
            &= 1+ \sum_{\bm{b} \in \mathcal{B}^{+M}} \prob(\bm{b}^\nround = \bm{b}) \left[\eta\widehat{\mu}^t(\bm{b}) + \eta^2 \widehat{\mu}^t(\bm{b})^2 \right] \leq \exp(\sum_{\bm{b} \in \mathcal{B}^{+M}} \prob(\bm{b}^\nround = \bm{b}) \left[\eta\widehat{\mu}^t(\bm{b}) + \eta^2 \widehat{\mu}^t(\bm{b})^2 \right])\,.
        \end{align*}
        Combining this with Equation~\eqref{eq: Potentials} and then taking logarithms, we obtain:
        \begin{align*}
            \eta \sum_{t=1}^{T} \widehat{\mu}^t(\bm{b}') \leq \log \Phi^0 + \eta \sum_{t=1}^{T} \sum_{\bm{b} \in \mathcal{B}^{+M}} \prob(\bm{b}^\nround = \bm{b}) \widehat{\mu}^t(\bm{b}) + \eta^2 \sum_{t=1}^T \sum_{\bm{b} \in \mathcal{B}^{+M}}\prob(\bm{b}^\nround = \bm{b}) \widehat{\mu}^t(\bm{b})^2\,.
        \end{align*}
        Dividing both sides by $\eta$, applying the upper bound on $\Phi^0$, and rearranging, we obtain that for any $\bm{b}' \in \mathcal{B}^{+M}$:
        \begin{align*}
            \sum_{t=1}^{T} \widehat{\mu}^t(\bm{b}') - \sum_{t=1}^{T} \sum_{\bm{b} \in \mathcal{B}^{+M}} \prob(\bm{b}^\nround = \bm{b})  \widehat{\mu}^t(\bm{b}) \lesssim \eta^{-1} \Nitem \log |\mathcal{B}| + \eta \sum_{t=1}^T \sum_{\bm{b} \in \mathcal{B}^{+M}}\prob(\bm{b}^\nround = \bm{b}) \widehat{\mu}^t(\bm{b})^2\,.
        \end{align*}
        Taking expectations of $\widehat{\mu}^\nround(\bm{b})$ with its definition in terms of $\widehat{w}^\nround_m(b_m)$, we obtain the right hand side of the lemma:
        \begin{align*}
            \mathbb{E}\left[\sum_{t=1}^{T} \widehat{\mu}^t(\bm{b}') - \sum_{t=1}^{T} \sum_{\bm{b} \in \mathcal{B}^{+M}} \prob(\bm{b}^\nround = \bm{b})  \widehat{\mu}^t(\bm{b})\right] \lesssim \eta^{-1}M\log|\mathcal{B}| + \eta \sum_{\nround=1}^\Nround \sum_{\bm{b}} \prob(\bm{b}^\nround =\bm{b}) \mathbb{E}[(\sum_{m=1}^M \widehat{w}^\nround_m(b_m))^2]\,.
        \end{align*}
        Replacing $\sum_{\bm{b} \in \mathcal{B}^{+M}} \prob(\bm{b}^\nround = \bm{b})  \widehat{\mu}^t(\bm{b})$ with $\mathbb{E}[\widehat{\mu}^\nround(\bm{b}^\nround)]$ and recalling that the bid vector utility estimates $\widehat{\mu}^\nround$ were unbiased, we have:
        \begin{align*}
            \sum_{t=1}^{T} \mu^t(\bm{b}') - \sum_{t=1}^{T} \mathbb{E}[\mu^t(\bm{b}^\nround)] \lesssim \eta^{-1}M\log|\mathcal{B}| + \eta \sum_{\nround=1}^\Nround \sum_{\bm{b}} \prob(\bm{b}^\nround =\bm{b}) \mathbb{E}[(\sum_{m=1}^M \widehat{w}^\nround_m(b_m))^2]\,.
        \end{align*}
        Notice that as this is true for any $\bm{b}'$, we can replace it with the bid vector that maximizes the true cumulative utility $\sum_{t=1}^T \mu^\nround(\bm{b}')$ and see that the left hand side becomes precisely $\textsc{Regret}_\mathcal{B}$, which completes the proof.
        \Halmos
    \endproof

    Now, it remains to show an upper bound on the second moment of the bid vector utility estimate $\mathbb{E}[(\sum_{m=1}^M \widehat{w}^\nround_m(b_m))^2]$. A crude attempt would be to say that $\mathbb{E}[(\sum_{m = 1}^M \widehat{w}^\nround_m(b_m))^2] \leq \Nitem \sum_{m = 1}^M \mathbb{E}[\widehat{w}^\nround_m(b_m)^2]$:
    \begin{align}
        &\sum_{\nround=1}^\Nround \sum_{\bm{b}} \prob(\bm{b}^\nround = \bm{b}) \mathbb{E}[(\sum_{m = 1}^M \widehat{w}^\nround_m(b_m))^2] \leq \Nitem \sum_{\nround=1}^\Nround \sum_{\bm{b}} \prob(\bm{b}^\nround = \bm{b}) \sum_{m = 1}^M \mathbb{E}[\widehat{w}^\nround_m(b_m)^2] \label{eq: full info difference}\\
        &= \Nitem \sum_{\nround=1}^\Nround \sum_{\nitem = 1}^\Nitem \sum_{b \in \mathcal{B}} \mathbb{E}[\widehat{w}^\nround_m(b)^2] \sum_{\bm{b}: b_m = b} \prob(\bm{b}^\nround = \bm{b}) = \Nitem \sum_{\nround=1}^\Nround \sum_{\nitem = 1}^\Nitem \sum_{b \in \mathcal{B}} \mathbb{E}[\widehat{w}^\nround_m(b)^2] q^\nround_m(b) \notag\\
        &\leq \Nitem \sum_{\nround=1}^\Nround \sum_{\nitem = 1}^\Nitem \sum_{b \in \mathcal{B}} \frac{2}{q_m^t(b)} q_m^t(b) = O(\Nitem^2 |\mathcal{B}| \Nround)\,. \notag
    \end{align}
    Where the last inequality follows from:
    \begin{align*}
        \mathbb{E}[\widehat{w}_m^t(b)^2] = \mathbb{E}\left[\left( 1 - \frac{1-w_m^t(b)}{q_m^t(b)} \textbf{1}_{b_m^t = b} \right)^2 \right] = 1 - 2\mathbb{E}\left[\frac{1-w_m^t(b)}{q_m^t(b)}\textbf{1}_{b_m^t=b}\right] + \mathbb{E}\left[\left(\frac{1-w_m^t(b; \bm{v})}{q_m^t(b)}\right)^2\textbf{1}_{b_m^t=b}\right]\,.
    \end{align*}
    Evaluating the expectations with $\mathbb{E}\left[\textbf{1}_{b_m^t = b}\right] = q_m^t(b)$, we have:
    \begin{align*}
        \mathbb{E}[\widehat{w}_m^t(b)^2] = 1 - \left[2 - 2w_m^t(b)\right] + \left[\frac{(1 - w_m^t(b))^2}{q_m^t(b)}\right] = 2w_m^t(b) - 1 + \frac{1}{q_m^t(b)} \leq 1 + \frac{1}{q_m^t(b)} \leq \frac{2}{q_m^t(b)}\,.
    \end{align*}    
    Plugging this back into our upper bound yields stated regret bound for $\eta = \Theta(\sqrt{\frac{\log |\mathcal{B}|}{M|\mathcal{B}|T}})$ such that $\eta < \frac{1}{M}$:
    
    \begin{align*}
        \textsc{Regret}_\mathcal{B} \lesssim \eta^{-1}\Nitem \log |\mathcal{B}| + \eta\Nitem^2 |\mathcal{B}|T = O(\Nitem^{\frac{3}{2}}\sqrt{|\mathcal{B}| \Nround \log |\mathcal{B}|})\,.
    \end{align*}

    \textbf{Part 3: Complexity Analysis.} We note that the time and space complexity analysis is identical to that of Algorithm~\ref{alg: Decoupled Exponential Weights}, as the only additional computational work being done is computing the normalization terms $q_m^\nround(b)$, which requires $O(M|\mathcal{B}|)$ space and $O(MT|\mathcal{B}|)$ time respectively. Hence, discarding old tables, the total time and space complexities of Algorithm \ref{alg: Decoupled Exponential Weights - Path Kernels} are $O(\Nitem |\mathcal{B}| \Nround)$ and $O(\Nitem |\mathcal{B}|)$ respectively. As this is polynomial in $\Nitem, |\mathcal{B}|, \Nround$, we have proven the claim of polynomial space and time complexities.

    \textbf{Part 4: Selecting $\mathcal{B}$.} We claim that the sub-optimality due to the discretization is on the order of $\frac{MT}{|\mathcal{B}|}$. Assume that $\mathcal{B} \equiv \{\frac{i}{|\mathcal{B}|}\}_{i \in [|\mathcal{B}|]}$ is an even discretization of $[0, 1]$, and recall the continuous regret benchmark:
    \begin{align*}
        \textsc{Regret} = \max_{\bm{b} \in [0, 1]^{+\Nitem}} \sum_{\nround=1}^\Nround \mu^\nround_n(\bm{b}) - \mathbb{E}\left[\sum_{\nround=1}^\Nround \mu^\nround_n(\bm{b}^\nround)\right],
    \end{align*}
    where the maximum is taken over the entire space $[0, 1]^{+\Nitem}$ rather than $\mathcal{B}^{+\Nitem}$. Let $\bm{b}^*$ denote the maximizer of the continuous regret. Then, bidder $n$ could have obtained at least the same allocation by rounding up each bid in $\bm{b}^*$ to the next largest multiple of $\frac{1}{|\mathcal{B}|}$. Let this rounded bid vector be denoted by $\bm{b}^+$. As their allocation, thus value for the set of items received, does not decrease, and their total payment increases by a maximum of $\frac{M}{|\mathcal{B}|}$ at each round, then we have that $\mu_n^t(\bm{b}^+) \geq \mu_n^t(\bm{b}^*) - \frac{M}{|\mathcal{B}|}$. Let $\bm{b}_{\mathcal{B}}^* \in \mathcal{B}^{+\Nitem}$ denote the hindsight optimal utility vector returned by our offline dynamic programming (Algorithm~\ref{alg: Offline Full}), which serves as the regret benchmark in the definition of discretized regret. Noting that $\bm{b}^+ \in \mathcal{B}^{+\Nitem}$, we have that the total utility of bidding $\bm{b}^*_{\mathcal{B}}$ must be at least that of $\bm{b}^+$. Thus,
    \begin{align*}
        \sum_{\nround=1}^\Nround \mu^\nround_n(\bm{b}_{\mathcal{B}}^*) \geq \sum_{\nround=1}^\Nround \mu^\nround_n(\bm{b}^+)\geq \sum_{\nround=1}^\Nround \mu^\nround_n(\bm{b}^*) - \frac{MT}{|\mathcal{B}|}
    \end{align*}
    We balance this with the discretized regret $O(\Nitem^{\frac{3}{2}}\sqrt{|\mathcal{B}| \Nround \log |\mathcal{B}|})$ with $|\mathcal{B}| = M^{-\frac{1}{3}}T^{\frac{1}{3}}$. This yields continuous regret $\textsc{Regret} = O(M^{\frac{4}{3}}T^{\frac{2}{3}} \sqrt{\log \Nround})$.

    \textbf{Part 5: Extending to the Full Information Setting.} Thus far, we have only discussed the bandit feedback algorithm. Fortunately, the full information setting algorithm is exactly the same except for two differences: 1) we do not need to compute $\bm{q}$ and 2) we can replace the reward estimates $\widehat{\mu}^t(\bm{b})$ with the true rewards $\mu^t(\bm{b})$ in Equation~\ref{eq: full info difference}. The first difference can only serve to improve the time and space complexity of our algorithm. The second difference allows us to improve the bound on $\sum_{\nround=1}^\Nround \sum_{\bm{b}} \prob(\bm{b}^\nround = \bm{b}) \mathbb{E}[(\sum_{m = 1}^M \widehat{w}^\nround_m(b_m))^2]$ in the left hand sight of Equation~\ref{eq: full info difference} by replacing $\widehat{w}^\nround_m(b_m))$ with $w^t_m(b_m)$:
    \begin{align*}
        \sum_{\nround=1}^\Nround \sum_{\bm{b}} \prob(\bm{b}^\nround = \bm{b}) \mathbb{E}[(\sum_{m = 1}^M \widehat{w}^\nround_m(b_m))^2] = \sum_{\nround=1}^\Nround \sum_{\bm{b}} \prob(\bm{b}^\nround = \bm{b}) \mathbb{E}[(\sum_{m = 1}^M \widehat{w}^\nround_m(b_m))^2] \leq M^2 \sum_{\nround=1}^\Nround \sum_{\bm{b}} \prob(\bm{b}^\nround = \bm{b}) = M^2T
    \end{align*}
    
    Notice that this bound is a factor of $|\mathcal{B}|$ improvement over that in the bandit setting. Consequently, we obtain our stated regret bound of $O(M^\frac{3}{2} \sqrt{T \log |\mathcal{B}|})$ with the choice of $\eta = \Theta(\sqrt{\frac{\log |\mathcal{B}|}{MT}})$. Balancing this regret with the error term, which is of order $O\left(\frac{M|\mathcal{B}|}{T}\right)$, the optimal choice of $|\mathcal{B}|$ is given by $\Theta(\sqrt{\frac{T}{M}})$. This yields corresponding continuous regret of $O(M^{\frac{3}{2}}\sqrt{T \log T})$.

\end{proof}

\section{Regret Lower Bounds}

In this section, we prove our stated regret lower bounds for both the full information and bandit feedback settings.

\subsection{Proof of Theorem \ref{thm:lower}: Regret Lower Bound (Full Information)}

To construct our lower bounds, we construct a stochastic adversary whose distribution across their bids makes it difficult for the bidder to determine their optimal bid, and thus, occurs $\Omega(M\sqrt{T})$ regret while doing so. We define $\bm{b}'_- = (0,\ldots,0,c,\ldots,c)$, where there are $k$ and $\Nitem - k$ values of 0 and $c$ each. We additionally define $\bm{b}"_- = (c,\ldots,c)$ as the $\Nitem$-vector of bids at $c$. Restricting the adversary's bid vectors to be in $\{\bm{b}'_-, \bm{b}"_-\}$, we construct two adversary bid vector distributions $F$ and $G$ over $\{\bm{b}'_-, \bm{b}"_-\}^\Nround$ such that   under $F$, we have $\prob(\bm{b}_-^\nround = \bm{b}'_-) = \frac{1}{2} + \delta$  and $\prob(\bm{b}_-^\nround = \bm{b}"_-) = \frac{1}{2} - \delta$ 
 and under $G$, we have $\prob(\bm{b}_-^\nround = \bm{b}'_-) = \frac{1}{2} - \delta$ and $\prob(\bm{b}_-^\nround = \bm{b}"_-) = \frac{1}{2} + \delta$  for some $\delta \in [0, \frac{1}{2}]$ to be optimized over later.

 Assume that $\bm{v} = (1,\ldots,1)$, all tiebreaks are won for simplicity, and the competitors' bids over time are independent. Then, for certain choices of $c$ and $k$ (which we show below), the expected utility maximizing bid vector under $\{\bm{b}_-^\nround\}_{\nround \in [\Nround]} \sim F$ is $(0,\ldots,0)$ and under $\{\bm{b}_-^\nround\}_{\nround \in [\Nround]} \sim G$ is $(c,\ldots,c)$. 
 In particular, we can compute precisely the expected value of bidding $\bm{b}^\nround = \bm{b}$ for all $\nround \in [\Nround]$ under both $F$ and $G$. Note that as adversary bid values only take values in $\{0, c\}$ and bidder $n$ wins all tiebreaks, then the bidder only need consider bid vectors consisting only of all $0$ or $c$. Letting $\nitem$ denote the number of bids in $\bm{b}$ equal to $c$, we have:
    \begin{align*}
        \mathbb{E}_F\left[ \mu_n^\nround(\bm{b})\right] =  (\frac{1}{2} + \delta)\left((1 - c)m + \max(0, M - k - m)\right)  + (\frac{1}{2} - \delta)(1 - c)m  \,.
    \end{align*}
    Where $\mathbb{E}_F$ denotes the expectation with respect to the adversary bids drawn from $F$, namely $\{\bm{b}_-^\nround\}_{\nround \in [\Nround]} \sim F$ (and similarly for $\mathbb{E}_G$ below). In particular, we have that with probability $\frac{1}{2} + \delta$, the adversary will select bid $\bm{b}'_-$. We are then guaranteed to win $m$ units at a price of $c$, for a utility of $1 - c$ per unit. If $\nitem < k$, then $M - k - m$ of the items were won with price 0, for a utility of 1 per unit. With probability $\frac{1}{2}-\delta$, all of the adversary bids are $c$, and we obtain $\nitem$ units at a cost of $c$ each, which corresponds to utility $1 - c$. Similarly, we have: 
    \begin{align*}
        \mathbb{E}_G\left[\mu_n^\nround(\bm{b})\right] =  (\frac{1}{2} - \delta)\left((1 - c)m + \max(0, M - k - m)\right)  + (\frac{1}{2} + \delta)(1 - c)m \,.
    \end{align*}
    If we sample $\{\bm{b}^\nround_-\}_{\nround \in [\Nround]}$ according to the mixture $\frac{F+G}{2}$, this corresponds to the case where $\delta = 0$, i.e., the probability of observing either $\bm{b}'_-$ or $\bm{b}"_-$ is equal. We have for all $\bm{b}$:
    \begin{align*}
        \mathbb{E}_{(F+G)/2}[\mu_n^\nround(\bm{b})] = \frac{1}{2}((1 - c)m + \max(0, M-k-m)) + \frac{1}{2}(1 -c)m = (1-c)m + \frac{1}{2}\max(0, M-k-m) 
    \end{align*}
    Note that under $F$, the optimal occurs at the all 0's vector for $c > \frac{1}{2} - \delta$ and $(\frac{1}{2} + \delta)(M - k) > (1-c)m = 0$. Similarly, the optimal occurs at the all $c$'s vector for $c > \frac{1}{2} - \delta$ and $(\frac{1}{2} - \delta)(M - k) > (1 - c)M$. These obtain utilities of $(\frac{1}{2}+\delta)(M-k)$ and $M - Mc$ respectively. One choice of $c$ and $k$ is $\frac{2}{3}$ and $\frac{M}{3}$, with $0 < \delta < \frac{1}{6}$. Looking at the regret incurred each step of the algorithm by selecting any action $\bm{b}$, we have:
    \begin{align*}
        &\max_{\bm{b}'} \left(\mathbb{E}_F[ \mu_n^\nround(\bm{b}') - \mu_n^\nround(\bm{b})] \right) + \max_{\bm{b}'} \left(\mathbb{E}_G[ \mu_n^\nround(\bm{b}') - \mu_n^\nround(\bm{b})] \right) \\
        &\geq \max_{\bm{b}'} \mathbb{E}_F[\mu_n^\nround(\bm{b}')] + \max_{\bm{b}'} \mathbb{E}_G[\mu_n^\nround(\bm{b}')] - 2\max_{\bm{b}'} \mathbb{E}_{(F+G)/2}\left[  \mu_n^\nround(\bm{b}')\right]\\
        &\geq \mathbb{E}_F[ \mu_n^\nround((0,\ldots,0))] + \mathbb{E}_G[ \mu_n^\nround((c,\ldots,c))] - 2\max_{\bm{b}'} \mathbb{E}_{(F+G)/2}\left[  \mu_n^\nround(\bm{b}')\right]\\
        &= (\frac{1}{2}+\delta)(M-k) + \left(M - Mc\right) - 2\max_{\bm{b}'} \mathbb{E}_{(F+G)/2}\left[\mu_n^\nround(\bm{b}')\right]\\
        &\geq (\frac{1}{2}+\delta)(M-k) + \left(M - Mc\right) - \max_{m \in [M]} \left(2(1-c)m + \max(0, M-k-m) \right)
    \end{align*}
    Let $k = \frac{M}{3}$ and $c = \frac{2}{3}$, and the above expression simplifies to:
    \begin{align*}
        \frac{2M(1+\delta)}{3} - \max_{m \in [M]} \left(\frac{2m}{3} + \max(0, \frac{2M}{3}-m) \right) = \frac{2M(1+\delta)}{3} - \frac{2M}{3} = \frac{2M\delta}{3} = \Theta(M\delta)\,.
    \end{align*}     
    We can now invoke the useful lemma relating the regret under $(F+G)/2$ to the Kullback-Leilber divergence:
    \begin{lemma}[\cite{NonparametricEstimation2008} Theorem 2.2.]
        We have for any two discrete distributions $F$ and $G$:
        \begin{align}
            \mathbb{E}_{(F+G)/2} \left[\textsc{{Regret}}_\mathcal{B}(T)\right] = \Omega\left( \frac{\Delta}{2} \exp(-D_{\emph{KL}(F || G)}) \right)\,,
        \end{align}
        where $\Delta$ denotes the sum of the total regret incurred under $F$ or $G$.
    \end{lemma} 
    When $F$ and $G$ are independent Bernoulli processes with parameters $\frac{1}{2}+\delta$ and $\frac{1}{2}-\delta$ respectively, then $D_{\text{KL}}(F || G) \leq CT\delta^2$ for some constant $C$. Using $\Delta = \Theta(MT\delta)$ and $\delta = \Omega(\frac{1}{\sqrt{T}})$, we have that the previous lemma implies
    \begin{align}
        \textsc{Regret}_{\mathcal{B}}  = \Omega\left(\Nitem \sqrt{\Nround}\right)\,.
    \end{align}                 

\subsection{Proof of Theorem \ref{thm:lower_bound_bandit}: Regret Lower Bound (Bandit Feedback)}\label{sec:proof:lower:bandit}

  {\color{black}To construct the regret lower bound, we begin by establishing a base hypothesis concerning the distribution of the highest competing bids. Specifically, let \(H_{m}\) represent the marginal cumulative distribution function (CDF) of the \((M-m+1)\)'st highest competing bid, denoted by \(b_{-m}\):
    \[H_m(b) = \min(1,\frac{c_m}{1 - b}) \quad \text{for} \quad c_m = \frac{1}{2} - \frac{m}{4M}\] 
    and $b \in \mathcal{B}$. We note that larger $m$ implies smaller $c_m$, which implies a smaller CDF for the same value of $b$; i.e. $H_m(b) \geq H_{m+1}(b)$  for all $b \in \mathcal{B}$. Thus, this satisfies the stochastic dominance requirement on the highest other bids $b_{-1}\leq \ldots \leq b_{-M}$, and there exists a joint distribution over highest other bids such that their marginal distributions are precisely $H_m$ for all $m \in [M]$. 
        
    Assuming that the bidder's valuations are $\bm{v} = [1,\ldots,1]$ and they given tie-break priority, the expected utility for bidding $b$ for unit $m$ is equal to
    \[(1-b)\prob(b \geq b_{-m}) = (1-b)H_m(b) = c_m,\qquad b \in [0, 1-c_m]\,.\] That is, the expected utility from the $m$'th unit is precisely $c_m$ for any bid $b\in [0, 1-c_m]$. Here, the expected utility for any bid $b > 1 - c_m$ will be strictly smaller than $c_m$.

    Now, we construct a partitioning over the bid discretization space.  Assuming an even discretization $\mathcal{B} = \{\frac{i}{|\mathcal{B}|}\}_{i \in \mathcal{B}}$ with $|\mathcal{B}| = M^{\frac{1}{3}}T^{\frac{1}{3}}$, we partition $\mathcal{B}$ into $M$ disjoint buckets: 
    \[\mathcal{B}_m = \mathcal{B} \cap (\frac{1}{2} - \frac{m}{4M}, \frac{1}{2} - \frac{m-1}{4M}]\,.\]
    Here, each $\mathcal{B}_m$ is of size $\frac{|\mathcal{B}|}{4M}$, and notice that $\mathcal{B}_m$ lies entirely to the right of $\mathcal{B}_{m+1}$; i.e., $b > b'$ for any $b \in \mathcal{B}_m$ and  $b' \in \mathcal{B}_{m+1}$.  See Figure \ref{fig:partitions} for an illustration. We note that we are not restricting bids for the $m$'th unit to be within $\mathcal{B}_m$. We only construct this partition to generate a family of hypotheses the agent will have to differentiate between. 

Now, we construct the family of hypotheses \(\mathcal{H} = \{\times_{m=1}^M \mathcal{H}_m^{j_m}, j_m \in [|\mathcal{B}_m|], m\in [M]\}\) in a similar manner as \cite{DemandCurve2003}, where each hypothesis $(\mathcal{H}_m^{j_m})_{m\in [M]}$ can be described by \(M\) indices, \(j_m\), \(m\in [M]\). Here, index \(j_m\in [|\mathcal{B}_m|]\) is associated with \emph{the unit-\(m\) hypothesis}.

    In particular, for unit  $m$, we construct $|\mathcal{B}_m|$ possible hypotheses (i.e., $\mathcal{H}_m^{j_m}$, $j_m\in [|\mathcal{B}_m|]$), each indexed by $j_m$ where for hypothesis $\mathcal{H}^{j_m}_m$, 
    the utility of the bidder for unit $m$ from submitting bid $b^{j_m}_m$---defined as the ${j_m}$'th largest value in $\mathcal{B}_m$---is marginally larger than the remaining bid values in $\mathcal{B}_m$ by perturbing the base distribution $H_m(b^{j_m}_m)$ by $\gamma_m > 0$. More formally, define marginal CDF's $\{H_m^{j_m}\}_{j_m \in [|\mathcal{B}_m|]}$ for each $m \in [M]$, such that $H_m^{j_m}(b) = H_m(b)$ for all $b \in \mathcal{B}$, except at $b^j_m \in \mathcal{B}_m$, and $H_m^{j_m}(b_m^{j_m}) = H_m(b_m^{j_m}) + \gamma_m$. \footnote{The agent does not need to differentiate between every hypothesis in $\mathcal{H} = \{\times_{m=1}^M \mathcal{H}^{j_m}_m\}_{j_m \in |\mathcal{B}_m| \forall m}$, which is exponentially large. As the utility function for bidding $\bm{b}$ can be decomposed into the sum of utilities of each of the $M$ units, the agent only needs to differentiate between each of the $\mathcal{H}^{j}_m$.}
    Here, we require $\gamma_m$ to satisfy the following constraint: $$\gamma_m \leq \min\left(H_m(b_m^{j_m} + \frac{1}{|\mathcal{B}|}) - H_m(b_m^{j_m}), H_m(b_m^{j_m}) - H_{m+1}(b_m^{j_m})\right)\,.$$ The first term $H_m(b_m^{j_m} + \frac{1}{|\mathcal{B}|}) - H_m(b_m^{j_m})$ corresponds to the constraint that the CDF is non-decreasing. The second term $H_m(b_m^{j_m}) - H_{m+1}(b_m^{j_m})$ corresponds to the highest-other-bid monotonicity constraint. Using our precise definitions of $H_m$, we see that: \footnote{We also note another constraint requires that $H_m(b_m^j + \frac{1}{|\mathcal{B}|})$ and $ H_m(b_m^j)$ needs to be unique, specifically, that they are not both equal to 1. Since we defined the set $\mathcal{B}_m$ to be within $(\frac{1}{2} - \frac{m}{4M}, \frac{1}{2} - \frac{m-1}{4M}]$, the largest value in $\mathcal{B}_m$ is $\frac{1}{2} - \frac{m-1}{4M}$. Plugging this into $H_m^j$, we see that $H_m(\max(\mathcal{B}_m)) = \min(1, \frac{\frac{1}{2} - \frac{m}{4M}}{1 - \max(\mathcal{B}_m)} = \min(1, \frac{\frac{1}{2} - \frac{m}{4M}}{\frac{1}{2} + \frac{m-1}{4M}}) < 1$.}
    \begin{align*}H_m(b_m^{j_m} + \frac{1}{|\mathcal{B}|}) - H_m(b_m^{j_m})&= \frac{c_m}{|\mathcal{B}|(1 - b_m^{j_m})(1-b_m^{j_m} - |\mathcal{B}|^{-1})} \\
   & =
   \frac{\frac{1}{2} - \frac{m}{4M}}{|\mathcal{B}|(1 - b_m^{j_m})(1-b_m^{j_m} - |\mathcal{B}|^{-1})}
   \geq \frac{\frac{1}{2} - \frac{m}{4M}}{|\mathcal{B}|(  \frac{1}{2} + \frac{m}{4M})(\frac{1}{2} + \frac{m}{4M} - |\mathcal{B}|^{-1})}\,.
    \end{align*}
    This function is decreasing for any $0 \leq m \leq M$, as the term $\frac{m}{M}$ appears negatively in the numerator and positively in the denominator. Thus, we have
    \[\frac{\frac{1}{2} - \frac{m}{4M}}{|\mathcal{B}|(  \frac{1}{2} + \frac{m}{4M})(\frac{1}{2} + \frac{m}{4M} - |\mathcal{B}|^{-1})} \ge  \frac{\frac{1}{2} - \frac{1}{4}}{|\mathcal{B}|(  \frac{1}{2} + \frac{1}{4})(\frac{1}{2} + \frac{1}{4} - |\mathcal{B}|^{-1})} \ge 4/|\mathcal{B}| \]
    as desired.
Further, 
    $$H_m(b_m^j) - H_{m+1}(b_m^{j_m}) = \frac{c_{m} - c_{m+1}}{1-b_m^{j_m}}  =
    \frac{\frac{1}{4M}}{ \frac{1}{2} + \frac{m}{4M}}
    \geq \frac{\frac{1}{4M}}{ \frac{1}{2} + \frac{1}{4}} \ge \frac{1}{3M}\,.$$
    Combining these, we see that $\gamma_m$ is $O\left(\min(\frac{1}{|\mathcal{B}|}, \frac{1}{M})\right)$. We then see that the expected utility of bidding $b_m^{j_m} $ for unit $m$ is $c_m + \gamma_m(1-b_m^{j_m}) = c_m + O\left(\min(\frac{1}{|\mathcal{B}|}, \frac{1}{M})\right)$  as compared to the expected utility of bidding $b \neq b_m^j$ for unit $m$, which is precisely just $c_m$. This sub-optimality plays a direct role in the regret accumulated by item $m$, which is precisely the number of times that the optimal $b_m^{j_m}$ was \textit{not} selected times the difference between the reward of the optimal and sub-optimal arms, which is $O(\gamma_m)$.
    We will select $\gamma_m = \gamma = \frac{1}{48\max(|\mathcal{B}|, M)}$ for all $m$. Now, we rewrite the regret lower bound as a function of the number of times that the optimal bid $b_m^{j_m}$ was not selected for for the $m$'th bid. Let
    \[
    \mathcal{S}(m, t) = \{m' \in [M]: b_{m'}^t \in \mathcal{B}_m\} \quad \text{and} \quad T_m = \sum_{t=1}^T |\mathcal{S}(m, t)|
    \] 
    where $\mathcal{S}(m, t)$ denotes the set of bids at round $t$ that took on value in $\mathcal{B}_m$ for any unit $m'\in [M]$ (including unit $m$), and $T_m$ denotes the total number of times across all rounds that any bid was selected within $\mathcal{B}_m$. Assuming that all bids fall within one of the $M$ buckets, as bidding outside of these buckets is strictly utility sub-optimal and also cannot provide any information on distinguishing between the unit hypotheses, we have:
     \begin{align*}
        \textsc{Regret}_{\mathcal{B}} &= \sum_{t=1}^T \sum_{m=1}^M 1_{b_m^t \neq b_m^{j_m}} \gamma(1 - b_m^{j_m}) \overset{(a)}{\geq} \frac{\gamma}{2} \sum_{t=1}^T \sum_{m=1}^M 1_{b_m^t \neq b_m^{j_m}} \overset{(b)}{=} \frac{\gamma}{2} \sum_{t=1}^T \sum_{m=1}^M \sum_{i \in \mathcal{S}(m, t)} 1_{b_{i}^t \neq b_i^{j_i}}\\
        &\overset{(c)}{=} \frac{\gamma}{4} \sum_{m=1}^M \left[ \sum_{t=1}^T \sum_{i \in \mathcal{S}(m, t)} 1_{b_{i}^t \neq b_i^{j_i}} + \sum_{t=1}^T \sum_{i \in \mathcal{S}(m, t)} 1_{b_{i}^t \neq b_i^{j_i}} \right]\,.\\
       &{=}\frac{\gamma}{4} \sum_{m=1}^M \left[ \sum_{t=1}^T \sum_{i \in \mathcal{S}(m, t)} 1_{b_{i}^t \neq b_i^{j_i}} + \sum_{t=1}^T \sum_{i\in S(m, t)}(1- 1_{b_i^t = b_i^{j_i}}) \right]\\
       &{=}\frac{\gamma}{4} \sum_{m=1}^M \left[ \sum_{t=1}^T \sum_{i \in \mathcal{S}(m, t)} 1_{b_{i}^t \neq b_i^{j_i}} + \sum_{t=1}^T|S(m,t)|- \sum_{t=1}^T \sum_{i\in S(m,t)} 1_{b_i^t = b_i^{j_i}} \right]\\
        &\overset{(d)}{=}\frac{\gamma}{4} \sum_{m=1}^M \left[ \sum_{t=1}^T \sum_{i \in \mathcal{S}(m, t)} 1_{b_{i}^t \neq b_i^{j_i}} + \sum_{t=1}^T |\mathcal{S}(m, t)| - \sum_{t=1}^T 1_{b_m^t = b_m^{j_m}} \right]\,.
    \end{align*}
    
Here, the inequality (a) holds because  $b_m^{j_m}< 1/2$, and the  equality (b) holds because for any  $t$, $m,m'\in [M]$,  and $m\ne m'$, we have $S(m, t)\cap S(m', t)=\emptyset$.  The   equality (c) is a simple algebraic manipulation. The   equality (d) uses the fact that  any time a bid for item $i \neq m$ takes value in $\mathcal{B}_m$ (see the definition of $S(m, t)$), it cannot possibly be equal to $b_{i}^{j_i} \in \mathcal{B}_i$ since $\mathcal{B}_i$ and $\mathcal{B}_m$ are disjoint.
We then have 
    \begin{align}
        \textsc{Regret}_{\mathcal{B}}
        &\overset{(e)}{\geq} \frac{\gamma}{4} \sum_{m=1}^M \left[ \sum_{t=1}^T \sum_{i \in \mathcal{S}(m, t)} 1_{b_{i}^t \neq b_i^{j_i}} + \max(0, T_m - T) \right] 
        \\
        &\overset{(f)}{\geq} \frac{\gamma}{4} \sum_{m=1}^M \left[ \sum_{t=1}^T \sum_{i \in \mathcal{S}(m, t)} 1_{b_{i}^t \neq b_m^{j_m}} + \max(0, T_m - T) \right] 
        \\
        &\overset{(g)}{\geq} \frac{\gamma}{4} \sum_{m=1}^M \left[ \max(0, T_m - c\gamma T_m^{\frac{3}{2}}\sqrt{\frac{M}{|\mathcal{B}|}}) + \max(0, T_m - T) \right] 
    \end{align}
   Inequality (e) holds since the maximum number of times that the $m$'th bid is equal to its optimal is at most $T$---i.e. $\sum_{t=1}^T 1_{b_m^t = b_m^{j_m}} \leq T$.
   Inequality (f) follows as for any $i \neq m, i \in \mathcal{S}(m, t)$, we have that $1_{b_i^t \neq b_i^{j_i}} = 1 \geq 1_{b_i^t \neq b_m^{j_m}}$ as $b_i^t \in \mathcal{B}_m$ and $b_i^{j_i} \notin \mathcal{B}_m$ for $i \neq m$. 
    To obtain inequality (g), we notice that the term $\sum_{t=1}^T \sum_{i \in \mathcal{S}(m, t)} 1_{b_{i}^t \neq b_m^{j_m}}$ is precisely the number of times that the optimal price $b_m^{j_m}$ was \textit{not} selected out of $\sum_{t=1}^T \sum_{i \in \mathcal{S}(m, t)} 1 = T_m$ times. From Theorem 4.3 and Lemma 4.4 of \cite{DemandCurve2003}, the number of times that the optimal action (price in their context and bid in our context) \textit{was selected} (out of $K$ possible actions) over the course of $T_m$ draws is upper bounded by $c\epsilon T_m^{\frac{3}{2}} K^{-\frac{1}{2}}$ for some absolute constant $c \in (2/3,1)$ (for example, $c = \frac{2\sqrt{2}}{3}$ as in \cite{DemandCurve2003}) and perturbation of size $\epsilon$ (i.e, the cost of not choosing the correct hypothesis). Plugging in $K = |\mathcal{B}_m| = \frac{|\mathcal{B}|}{M}$ and $\epsilon = \gamma_m = \gamma$, we obtain inequality (g).

    In the paradigm with large $T$ and small $M$, more specifically with $T \geq M^2 \to |\mathcal{B}| = M^{\frac{1}{3}}T^{\frac{1}{3}} \geq M$, we have $\gamma = \frac{1}{32\max(|\mathcal{B}|, M)} = \Theta(\frac{1}{|\mathcal{B}|}) = \Theta(M^{-\frac{1}{3}}T^{-\frac{1}{3}})$. Plugging this in, we obtain:
    \begin{align}
        \textsc{Regret}_{\mathcal{B}} \geq \frac{M^{-\frac{1}{3}}T^{-\frac{1}{3}}}{4} \sum_{m=1}^M  T_m\left[ \max(0, 1 - c\sqrt{\frac{T_m}{T}}) + \max(0, T_m - T) \right]\,.\label{eq: tm to t 4}
    \end{align}
    Since the agent seeks to minimize regret, they take the minimum over $\bm{T} = (T_1,\ldots,T_M)$, subject to $MT = \sum_{m=1}^M T_m \in \mathbb{N}^M$ for $T_m \in \mathbb{N}$.
    \begin{align*}
        \textsc{Regret}_{\mathcal{B}} \geq \frac{M^{-\frac{1}{3}}T^{-\frac{1}{3}}}{4} \sum_{m=1}^M  T_m\left[ \max(0, 1 - c\sqrt{\frac{T_m}{T}}) + \max(0, T_m - T) \right]  \quad \text{s.t.} \quad \sum_{m=1}^M T_m = MT
    \end{align*}

\begin{lemma} \label{lem:optimization}
Consider the following optimization problem for any $c\in (0, 1)$: 
\[\min_{\{T_m\}_{m\in [M]}, T_m \ge 0}\sum_{m=1}^M  T_m\left[ \max(0, 1 - c\sqrt{\frac{T_m}{T}}) + \max(0, T_m - T) \right]  \quad \text{s.t.} \quad \sum_{m=1}^M T_m = MT\,.\]
where $T_m\in \in \mathbb{Z} $ are integers. 
Then, the optimal solution happens at $T_m= T$ for any $m\in [M]$.
\end{lemma}

The proof of the theorem is completed by applying Lemma  \ref{lem:optimization}. Plugging in $T_m = T$ into Equation~\ref{eq: tm to t 4}, we obtain a regret lower bound of $\textsc{Regret}_{\mathcal{B}} \geq \frac{1-c}{4} M^{\frac{2}{3}}T^{\frac{2}{3}} = \Omega(M^{\frac{2}{3}}T^{\frac{2}{3}})$.

$\blacksquare$

    \begin{proof}{Proof of Lemma \ref{lem:optimization}}
    To show that the optimal solution is precisely $\bm{T} = (T,\ldots,T)$, we minimize the objective function over the set of $T_m \in \mathbb{R}^+$ subject to $\sum_{m=1}^M T_m \geq MT$, which is subsumes the original equality constraint $\sum_{m=1}^M T_m = MT$ for integral $T_m \in \mathbb{Z}$. We show that $\bm{T} = (T,\ldots,T)$ is a solution to this relaxed problem, and since this satisfies the inequality and integrality constraints, it must be optimal for the original problem. We begin by finding first order conditions for the objective function: 
    \begin{align*}
        R(\bm{T}) \doteq \sum_{m=1}^M  \left[T_m \max(0, 1 - c\sqrt{\frac{T_m}{T}}) + \max(0, T_m - T) \right]\,.
    \end{align*}
    We consider the following cases:
    
    \begin{enumerate}
        \item If $1 \leq \frac{T_m}{T} \leq \frac{1}{c^2}$, the partial derivative of $R(\bm{T})$ with respect to $T_m= T(\frac{4}{3c})^2$  is zero; that is,  
        we have $0 = \delta_{T_m} R(\bm{T}) = 2 - \frac{3}{2}c\sqrt{\frac{T_m}{T}} \to T_m = T(\frac{4}{3c})^2$.
        However, this does not imply that  $T_m = T(\frac{4}{3c})^2$ is a minimizer as 
     it violates the second order optimality conditions. In particular,  for any value of $T_m > 0$, we have 
        \begin{align*}
            \delta_{T_m^2} R(\bm{T}) = -\frac{3}{4}c\sqrt{\frac{1}{T_mT}} < 0\,.
        \end{align*}
        Given that the second derivative is negative and the first derivative is positive for \( T_m < T \left( \frac{4}{3c} \right)^2 \) and negative for \( T_m > T \left( \frac{4}{3c} \right)^2 \), we only need to examine the boundary points \( T \) and \( \frac{T}{c^2} \). 

        \item If $\frac{T_m}{T} \leq 1 \leq \frac{1}{c^2}$, we have $0 = \delta_{T_m} R(\bm{T}) = 1$. This derivative is always positive. We then only need to check for the minimum boundary case, which is $T_m = 0$ as the derivative is positive.
        \item If $1 \leq \frac{1}{c^2} \leq \frac{T_m}{T}$, we have $0 = \delta_{T_m} R(\bm{T}) = 1 - \frac{3}{2}c\sqrt{\frac{T_m}{T}} \to T_m = T(\frac{2}{3c})^2$.  However, this does not imply that  $T_m = T(\frac{2}{3c})^2$ is a minimizer as 
     it violates second order optimality. In particular, for any value of $T_m > 0$, we have:
        \begin{align*}
            \delta_{T_m^2} R(\bm{T}) = -\frac{3}{4}c\sqrt{\frac{1}{T_mT}} < 0
        \end{align*}
        The first derivative is positive when $T_m < T(\frac{2}{3c})^2$ and negative when $T_m > T(\frac{2}{3c})^2$. This implies that we only need to check $T_m$ at the two boundary points, which are $T_m = \frac{T}{c^2}$ and $T_m = \infty$ (which is clearly sub-optimal). 
    \end{enumerate}
    In summary, we only need to check for the cases where $T_m \in \{0, T, \frac{T}{c^2}\}$. We check the value of the $m$'th summand for each of these cases:
    \begin{enumerate}
        \item $T_m = 0$: We have $T_m \max(0, 1 - c\sqrt{\frac{T_m}{T}}) + \max(0, T_m - T) = 0$.
        \item $T_m = T$: We have $T_m \max(0, 1 - c\sqrt{\frac{T_m}{T}}) + \max(0, T_m - T) = (1 - c)T$.
        \item $T_m = \frac{T}{c^2}$: We have $T_m \max(0, 1 - c\sqrt{\frac{T_m}{T}}) + \max(0, T_m - T) = \frac{T}{c^2} - T$.
    \end{enumerate}
   Let $x,y,z$ denote the number of $T_m$ equal to $0, T, \frac{T}{c^2}$ respectively. That is, $x =|\{m \in [M]: T_m = 0\}|$, $y =|\{m\in [M]: T_m = 0\}|$, $z =|\{m\in [M]: T_m = T/c^2\}|$.
    Then, we have that $x + y + z = M$ and $yT + z\frac{T}{c^2} \geq MT$. Rewriting $R(\bm{T})$:
    \begin{align*}
        R(\bm{T}) &= (1-c)Ty + (\frac{1}{c^2} - 1)Tz\,.
    \end{align*}
    We minimize over the set of $(x, y, z)$, yielding the following optimization problem:
    \begin{align*}
        \min_{x, y, z} (1-c)Ty + (\frac{1}{c^2} - 1)Tz \quad \text{s.t} \quad x + y + z = M, \quad x, y, z \geq 0, \quad Ty + \frac{Tz}{c^2} \geq MT\,.
    \end{align*}
    As $x$ does not appear in the objective function, it is effectively a slack variable, and we can rewrite this optimization problem over $(x, y, z)$ as one over $(y, z)$ only:
    \begin{align*}
        \min_{y, z} (1-c)Ty + (\frac{1}{c^2} - 1)Tz \quad \text{s.t} \quad y + z \leq M, \quad y, z \geq 0, \quad Ty + \frac{Tz}{c^2} \geq MT\,,
    \end{align*}
    where in the above optimization, we ignore integrality of $y, z$. If we can show that $(y, z) = (M, 0)$  is the optimal solution to the above optimization problem, then we are finished.  The constraints are linear in $y$ and $z$ and form a triangle with vertices at $(y, z) \in \{(M,0), (0, c^2M), (0, M)\}$. We note that the objective function is linear, thus all it remains to do is to compare the objective function values at each of the vertices. The objective function values corresponding to $(y, z) \in \{(M,0), (0, c^2M), (0, M)\}$ are $(1-c)MT, (1-c^2)MT, (\frac{1}{c}-1)MT$ respectively. Thus, for any $c \in (0, 1)$, the optimal solution is $(M, 0)$ as desired, showing that $\bm{T} = (T,\ldots,T)$ was optimal in the original equality constrained, integer constrained optimization problem.
\end{proof}}

\section{Additional Experiments}
\label{sec: additional experiments}

In this section, we run several additional experiments, including a comparison of the decoupled exponential weights algorithms (Algorithms~\ref{alg: Decoupled Exponential Weights} and \ref{alg: Decoupled Exponential Weights - Path Kernels}) and the OMD algorithms (Algorithm~\ref{alg: OMD}) under both full information and bandit feedback in the stochastic setting. We additionally provide more comprehensive experiments regarding the market dynamics of the PAB and uniform price auctions. In particular, we plot the evolution of the revenue, welfare, and winning bid gaps over time, as well compare the bandit feedback versions of the PAB and uniform price auctions.

\subsection{Stochastic Setting} \label{sec:stochastic}

To compare the regret incurred between our decoupled exponential weights and OMD based algorithms in a non-multi agent setting, we consider the setting where the bidder competes in a stochastic setting. Here, the bidder, endowed with valuation vector $\bm{v} = [1, 1, 1]$, will compete over the course of $T=10^4$ rounds for full information feedback ($T = 10^5$ rounds for bandit feedback) for $\overline{M}=M=3$ items. The competing bids are  $\bm{b}^{-1} = [0.1, 0.1, 0.1]$, $[0.3,0.3, 1.0]$, or $[0.4, 1.0, 1.0]$ with probabilities $\frac{1}{4}, \frac{1}{4}$, and $\frac{1}{2}$, respectively. Assuming that the bidder receives priority in tiebreaks, with $\mathcal{B} = \{\frac{i}{10}\}_{i \in [10]}$, the expected utility $\sum_{m=1}^3 \prob(b_m \geq b^{-1}_m) (1 - b_m)$ maximizing bid vector is given by $\bm{b} = [0.4, 0.3, 0.1]$, which yields utility $(1)(1-0.4) + (0.75)(1 - 0.3) + (0.5)(1-0.1) = 0.6 + 0.525 + 0.45 = 1.575$. We select learning rates $\eta = \sqrt{\tfrac{\log(|\mathcal{B}|)}{MT}}$ and $\eta = \sqrt{\tfrac{\log(|\mathcal{B}|)}{M|\mathcal{B}|T}}$ for the full information and bandit decoupled exponential weights algorithms, respectively. We select learning rates $\eta = \sqrt{\tfrac{\log(|\mathcal{B}|)}{T}}$ and $\eta = \sqrt{\tfrac{\log(|\mathcal{B}|)}{|\mathcal{B}|T}}$ for the full information and bandit OMD algorithms, respectively. For the uniform price algorithms, we select $\eta = \sqrt{\tfrac{\log(|\mathcal{B}|)}{MT}}$ and $\eta = \sqrt{\tfrac{\log(|\mathcal{B}|)}{M^3|\mathcal{B}|^2T}}$ for full information and bandit feedback, respectively.

\begin{figure}
    \centering
    \includegraphics[scale=0.37, trim={0 0 0 0}, clip]{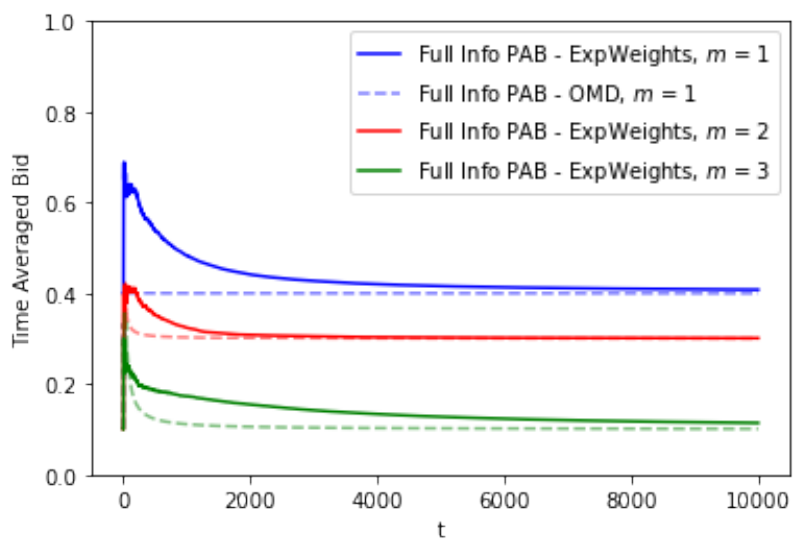}
    \includegraphics[scale=0.37, trim={20 0 0 0}, clip]{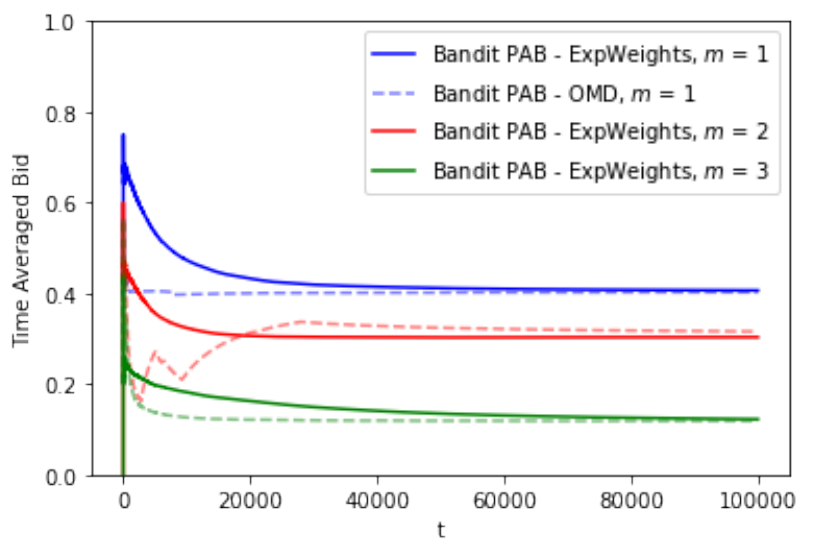}
    \includegraphics[scale=0.37, trim={20 0 0 0}, clip]{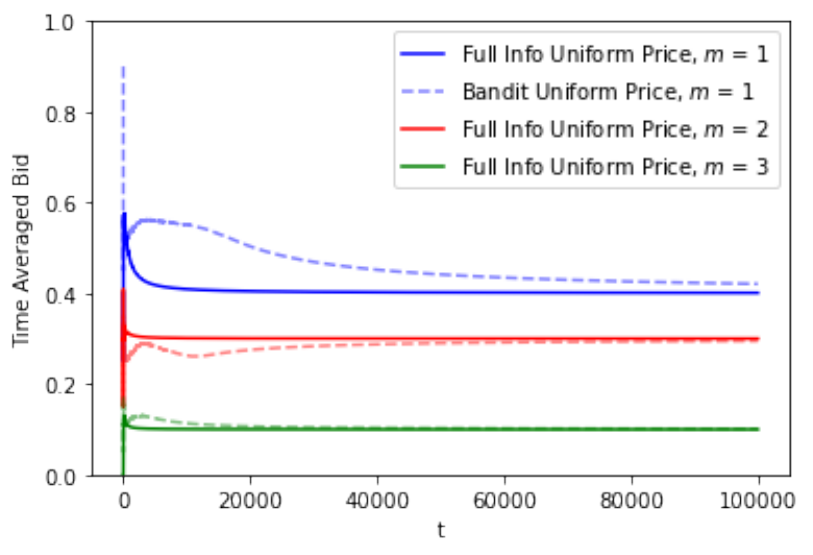}
    \caption{Bid convergence over time under the stochastic setting in Section \ref{sec:stochastic} for PAB with full information (left), PAB with bandit feedback (middle), and the uniform price auction (right). For the PAB auctions, the solid (resp. dashed) line denotes the decoupled exponential weights (resp. OMD). For uniform price, the solid and dashed lines are for full information and bandit feedback respectively. The top, middle, and bottom lines denote the time averaged values of $b_1, b_2$, and $b_3$ respectively.} 
    \label{fig:exp2_bids}
\end{figure}

In Figure \ref{fig:exp2_bids}, we plot the average value of each bid over time. Here, the bidder's objective is to learn the optimal bid vector under each of our algorithms: decoupled exponential weights algorithm (Algorithms~\ref{alg: Decoupled Exponential Weights} and \ref{alg: Decoupled Exponential Weights - Path Kernels}), and the OMD algorithm  (Algorithm \ref{alg: OMD}) for both the full information and bandit settings. In this figure, we further compare the rate of convergence to the optimal bid vector of $[0.4, 0.3, 0.1]$ with our three algorithms. 

We observe that the full information OMD converges the fastest to the optimal bid, followed by full information decoupled exponential weights. Similarly, the bandit OMD converges faster than the bandit exponential weights algorithm, albeit slower than either full information variant. This behavior is consistent with our theoretical findings. 

We repeat this experiment for the algorithms to learn in uniform price auctions described in \cite{brânzei2023online}. Though we do not perform the calculations, the optimal bid vector in the uniform price setting is still $[0.4, 0.3, 0.1]$. We note that it takes noticeably longer for the bandit algorithm to converge as compared to either its full information variant or our Algorithms \ref{alg: Decoupled Exponential Weights - Path Kernels} or \ref{alg: OMD}, as predicted by the looser regret upper bounds:
\begin{theorem}[(Discrete) Regret in Uniform Price Auctions, \cite{brânzei2023online}] \label{thm: uniform price regret full}
    Under full information feedback (resp. bandit feedback), there exists an algorithm which achieves $O(M^{\frac{3}{2}}\sqrt{T \log |\mathcal{B}|})$ (resp. $O(M^{\frac{5}{2}}|\mathcal{B}|T^{\frac{1}{2}}\sqrt{\log |\mathcal{B}|} + M^2 \log |\mathcal{B}|)$) discrete regret.
\end{theorem}

\subsection{Market Dynamics over Time}\label{sec:market_dynamics_additional}

{\color{black} In this section, we run five additional batches of experiments to better understand the market dynamics. We first analyze the evolution of welfare and revenue over time for a fixed $M, N, |\mathcal{B}|, T$. In the second experiment, we look at the evolution of the ratios between the largest winning, smallest winning, and largest losing bids over time. In the last three experiments, we repeat the previous two experiments for the uniform price auction, as well as the first experiment in Section~\ref{sec: experiments} that shows the regret dependence on $M$ and $T$.}

{\color{black}
\textbf{First Experiment: Welfare and Revenue Over Time - PAB.} In this next set of experiments, we compare welfare and revenue over time. We use $M = 5, |\mathcal{B}| = 20, N = 3, T = 10^5$ for these simulations, though we note that our findings are consistent across parameterizations. In Figure \ref{fig:rev_wel_over_time}, we further compare the distribution of welfare and revenue over time showing the 10th, 25th, 50th, 75th, and 90th percentiles in different shades under Algorithms~\ref{alg: Decoupled Exponential Weights} and \ref{alg: Decoupled Exponential Weights - Path Kernels}. We normalize both welfare and revenue by the optimal welfare (sum of the largest $M$ valuations) in each instance. In comparison to the full information version, it takes longer for the bidders to settle to an approximately welfare maximizing steady state in the bandit setting. Furthermore, the revenue at the recovered steady state under bandit feedback is lower than that of the full information setting, albeit with lower variance.

Additionally, we observe that revenues initially converge for intermediate values of $t$ before decreasing and subsequently increasing in variance for larger $t$. To explain this, it's essential to understand the evolution of the ratios between the largest and smallest winning bids, as well as the largest losing bid (See Figure~\ref{fig: bid ratios PAB}). Specifically, the decrease in revenue coincides with when the winning bids are far apart. The sudden drop in revenue indicates that bidders have learned to strategically shade their winning bids down to the clearing price. Eventually, the bids converge to either the clearing price, which naturally has some variance as it is a random function of the valuations and input parameters $N, M, |\mathcal{B}|$, or a best response cycle where bidders undercut and subsequently slightly outbid one another indefinitely (See Figure~\ref{fig: bid cycling main body}).

\begin{figure}
     \centering
     \begin{subfigure}[b]{0.3\textwidth}
         \includegraphics[scale=0.52, trim={0 0 0 25},clip]{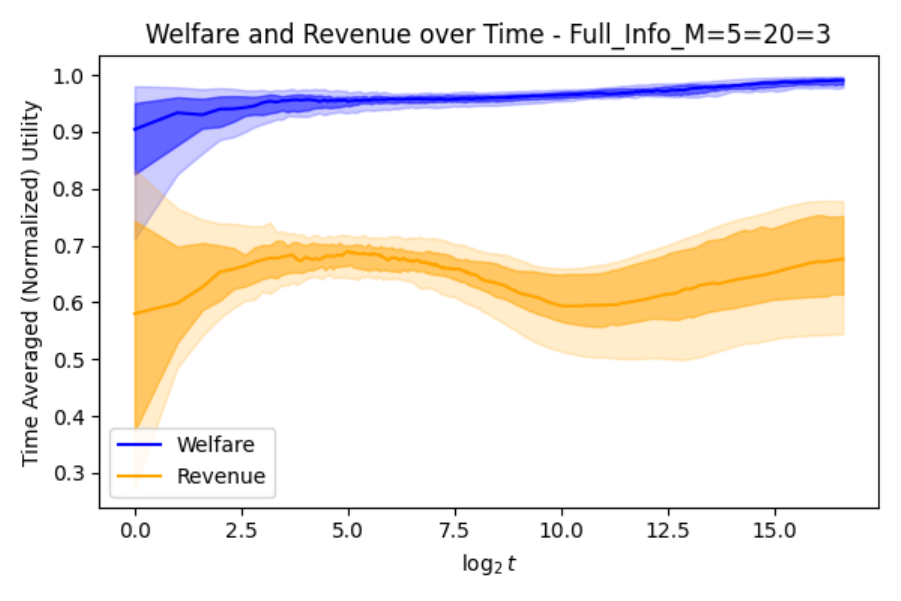}
     \end{subfigure}
     \hfill
     \begin{subfigure}[b]{0.51\textwidth}
         \includegraphics[scale=0.52, trim={0 0 0 25},clip]{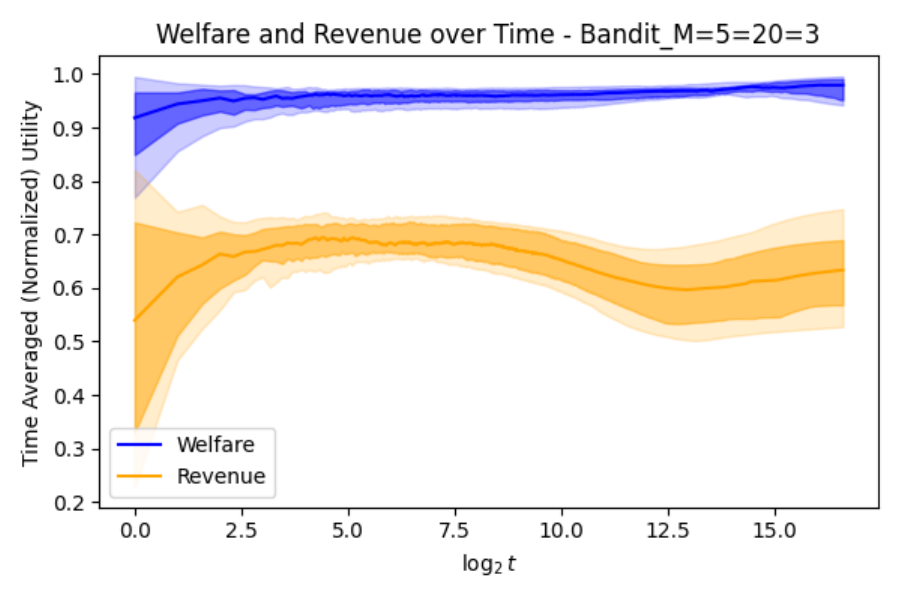}
     \end{subfigure}
    \caption{Welfare (top) and revenue (bottom) over time: The left (resp. right) figure corresponds with the full information setting (resp. bandit setting) for $M = 5$, $|\mathcal{B}| = 20$, $N = 3$, $T = 10^5$.  }. \label{fig:rev_wel_over_time}
\end{figure}

{\color{black}
\textbf{Second Experiment: Convergence of Winning Bids Over Time - PAB.}  In Figure \ref{fig: bid ratios PAB}, we compare the evolution of the ratio of the largest and smallest winning bids (top), as well as the ratio of the smallest winning bid to the largest losing bid (bottom). We find that the ratios both converge to approximately 1 over time, though somewhat faster  under full information feedback. Despite the existence of CCEs supported over joint bid profiles with large gaps between winning bids (see Lemma~\ref{lem: CCE non uniform winning}), the convergence of these ratios suggests near uniformity of the winning bids over the course of the learning dynamics induced by our decoupled exponential weights algorithms.}

\begin{figure}
     \centering
     \begin{subfigure}[b]{0.2\textwidth}
         \includegraphics[scale=0.5, trim={0 0 0 25},clip]{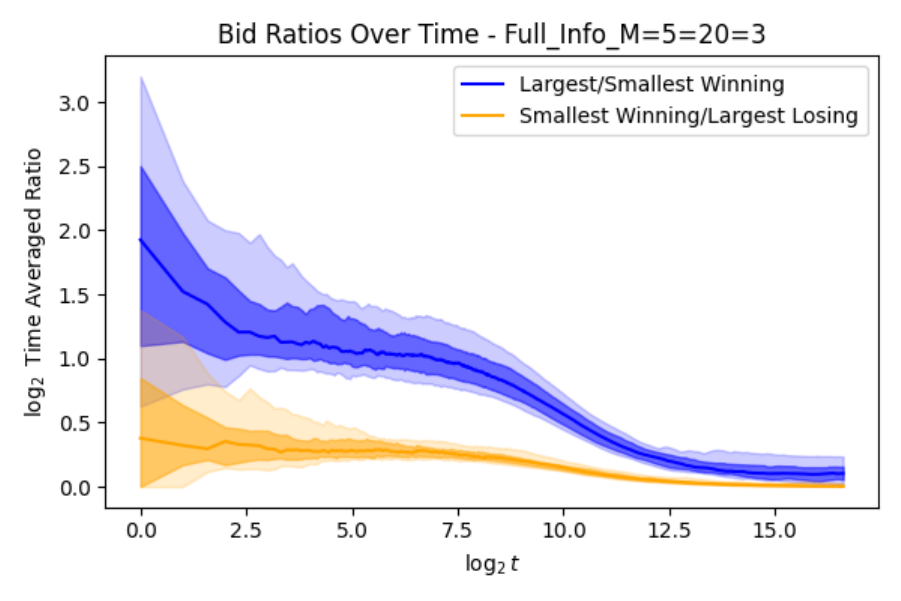}
     \end{subfigure}
     \hfill
     \begin{subfigure}[b]{0.5\textwidth}
         \includegraphics[scale=0.5, trim={0 0 0 25},clip]{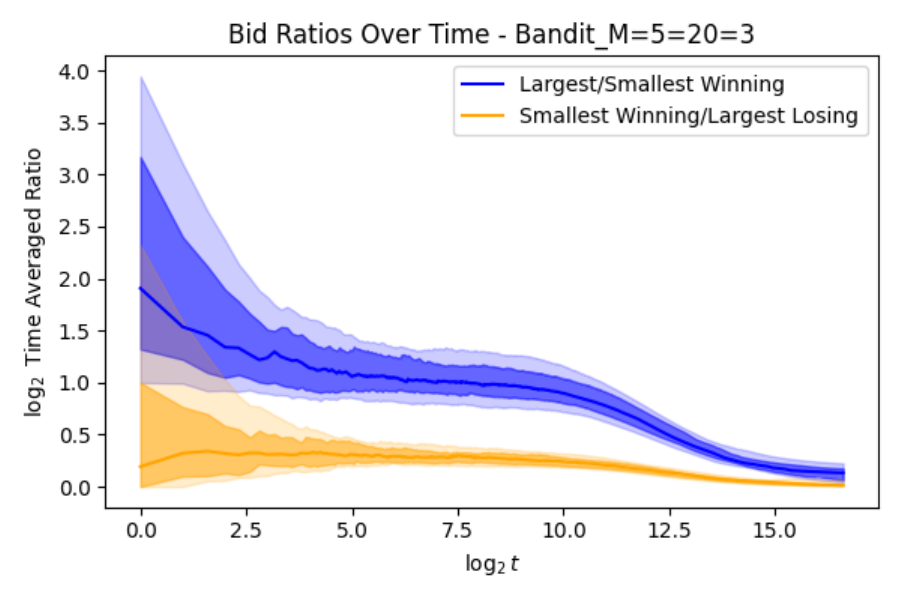}
     \end{subfigure}
    \caption{\textbf{Winning Bid Convergence. } Ratios of the largest winning bid to the smallest winning bid (top) and the smallest winning bid to the largest losing bid (bottom): The left (resp. right) figure corresponds with the PAB dynamics under full information (resp. bandit feedback).}. \label{fig: bid ratios PAB}
\end{figure}
{\color{black}\textbf{Third Experiment: Regret as a Function of $M$ and $T$ - Uniform Price.} To better understand the impact of $M$ and $T$ on the continuous regret, we run the repeated auction setting with varying $T, M$ for the uniform price auction. From \cite{brânzei2023online}, the continuous regret of their no-regret learning algorithm is of order $O(M^\frac{3}{2}T^{\frac{1}{2}} \log T)$ and $O(\min(MT, M^\frac{7}{4}T^{\frac{3}{4}} \log T))$ in the full information and bandit feedback settings respectively. To obtain these regret guarantees, we select $|\mathcal{B}|$ and $\eta$ to balance the discretization error of order $O(\frac{MT}{|\mathcal{B}|})$ and discrete regret: for the full information setting specifically, we set $|\mathcal{B}| = \max(5, \sqrt{\frac{T}{M}})$ and $\eta = \sqrt{\frac{\log T}{MT}}$. For bandit feedback, we set $|\mathcal{B}| = \max(5, M^{-\frac{3}{4}}T^\frac{1}{4})$ and $\eta = \min(\frac{1}{M}, \sqrt{\frac{\log |\mathcal{B}|}{M^3 |\mathcal{B}|^2T}})$. We plot the $\log-\log$ sum of discrete regrets across all $N=3$ bidders, normalized by $NT$ to obtain the per-bidder, per-round average regret; see Figure \ref{fig:regret-plot-unif}. Running a linear regression on the median, we find that the slopes w.r.t. $T$ of the median regret are approximately $-\frac{1}{2}$ and $-\frac{1}{6}$ for the full information and bandit settings respectively. This accurately reflects the theoretical regret bound in the full information setting and similarly, the slope for the bandit setting falls between the predicted $-\frac{1}{4}$ and 0. Interestingly, the slopes of the median w.r.t. $M$ are approximately $\frac{5}{4}$ and $\frac{7}{4}$ for the full information and bandit feedback settings respectively. While the bandit feedback setting is consistent with theory, the full information feedback regret dependence on $M$ is $O(M^{\frac{1}{4}})$ faster than predicted.}

\begin{figure}
    \centering
    \includegraphics[scale=0.5, trim={0 0 0 25},clip]{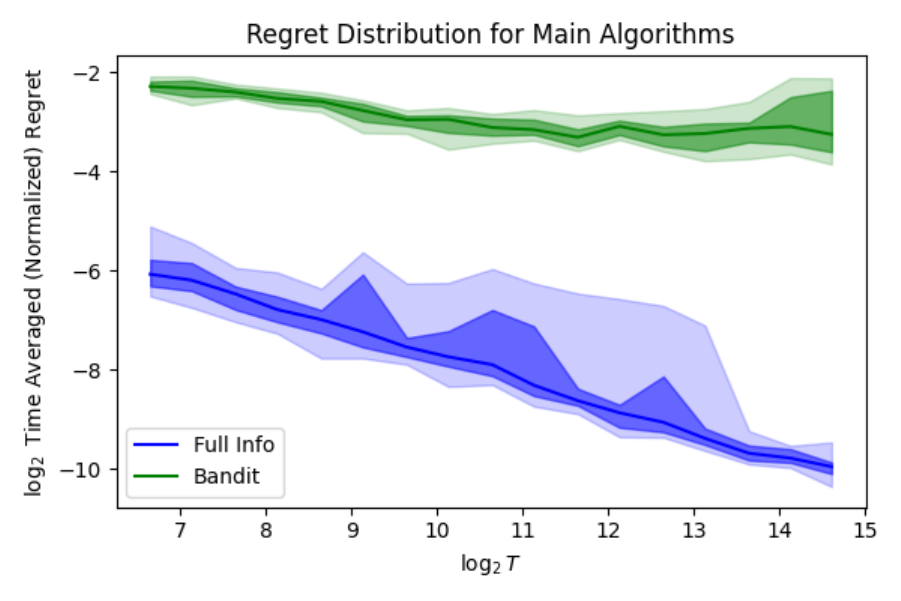}
    \includegraphics[scale=0.5, trim={0 0 0 25},clip]{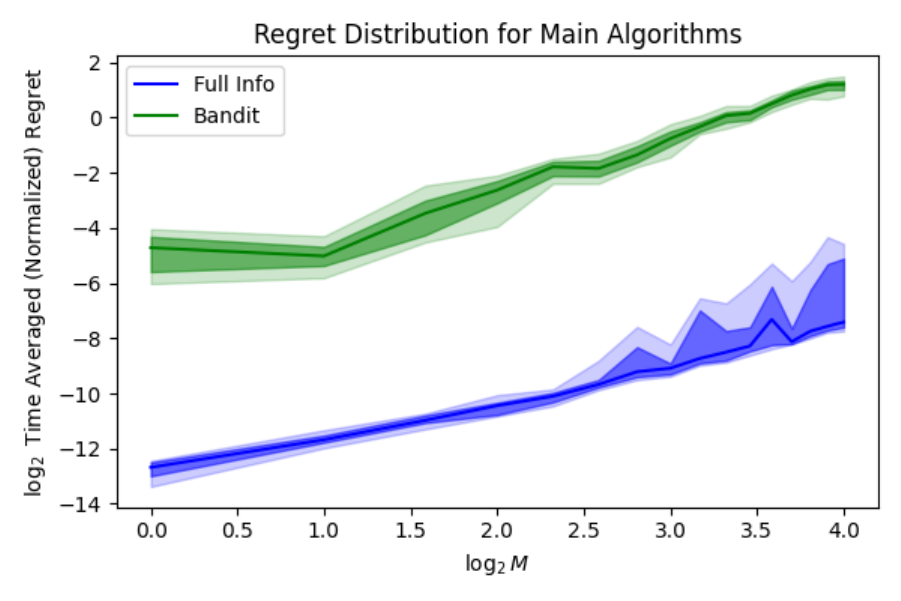}
    \caption{Uniform Price Regret as a Function of $M$ and $T$. We compare the time-averaged aggregate discretized regret across all agents for uniform price, under both full information (bottom) and bandit feedback (top), for both varying $T$ (left) and $M$ (right). When varying $T$ (resp. $M$), we fix $M = 5$ (resp. $T = 25000$) and derive $|\mathcal{B}|$ and $\eta$ according to \cite{brânzei2023online}.}
    \label{fig:regret-plot-unif}
\end{figure}

{\color{black}\textbf{Fourth and Fifth Experiments: Welfare, Revenue, Bid Ratios over Time - Uniform Price.} 
In Figure \ref{fig: uniform price market dynamics}, we present the evolution-over-time of the welfare, max-welfare normalized revenue, and ratios of the largest bids for the uniform price auction, under full information feedback. We use parameterization $N=3, M = 5, |\mathcal{B}| = 20, T = 10^5$ (our insights are consistent across parameterizations, see next section). Under uniform price, the payment is determined by the lowest winning bid, thus, less incentive to manipulate the largest bids. This observation aligns with findings from other variants of discriminatory versus uniform pricing, such as generalized first and second price auctions \cite{kagel1986curse, nyborg1996discvsunif, bergemann2020winnerscurse}. This explains the lower revenue and bid ratio non-convergence, whereas the slightly improved welfare stems from the tie-breaking rule (see Section~\ref{sec: experiments}).}

\begin{figure}
     \centering
     \includegraphics[scale=0.5, trim={0 0 0 25},clip]{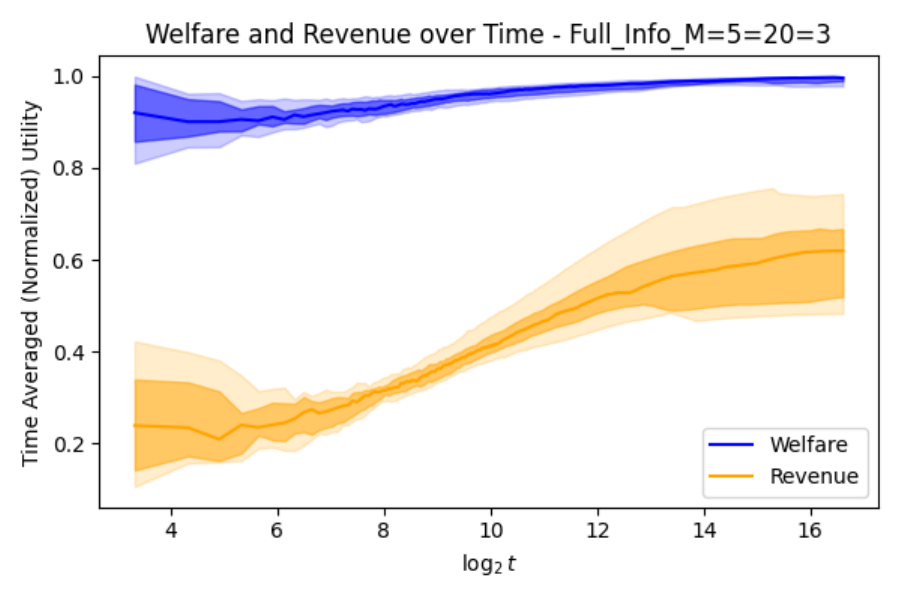}
     \includegraphics[scale=0.5, trim={0 0 0 25},clip]{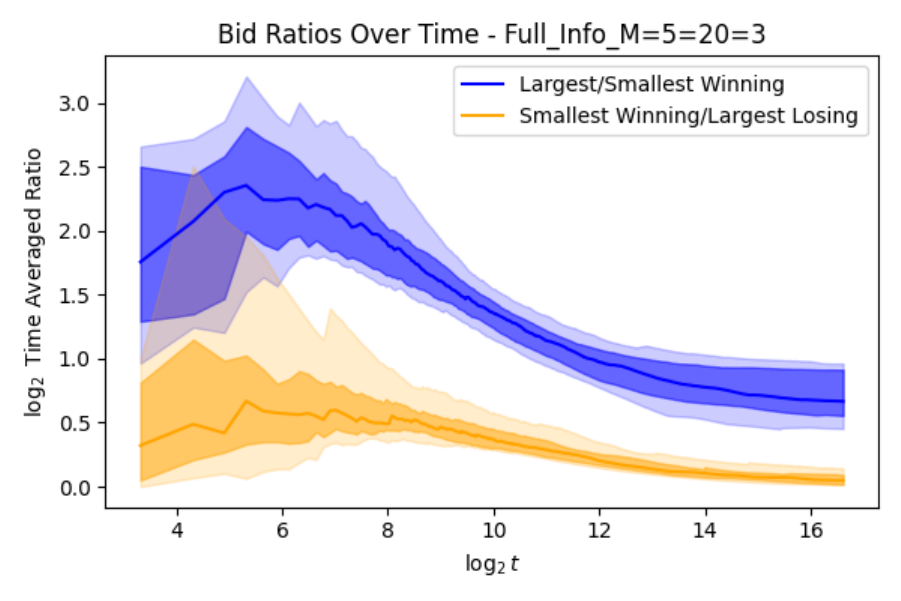}
    \caption{\textbf{Uniform Price Welfare, Revenue, Bid Ratios Over Time.} In the left figure, we plot the welfare (top) and revenue (bottom) of the uniform price auction over time, under full information with $M = 5, |\mathcal{B}| = 20, N = 3, T = 10^5$. In the right figure, we plot the ratios of the largest winning to smallest winning bids (top), and the smallest winning to largest losing bids (bottom).} 
    \label{fig: uniform price market dynamics}
\end{figure}

\subsection{Comparison of PAB and Uniform Price Auctions under Bandit Feedback}\label{sec:compare_additional}

{\color{black} In this section, we compare the regret, welfare, revenue, bid ratios, and competitive ratios of the equilibria under the bandit feedback version of both the PAB and uniform price no-regret learning algorithms. In particular, we run the same parameterizations with $N \in \{3, 5\}, M \in \{1, 5, 10\}, |\mathcal{B}| \in \{10, 20\}, T = 10^5$ as in Tables~\ref{table: learning dynamics full info} and \ref{table: learning dynamics full uniform}, except now for bandit feedback. We use $\eta = \sqrt{\frac{\log |\mathcal{B}|}{M|\mathcal{B}|T}}$ for PAB and $\eta = \sqrt{\frac{\log |\mathcal{B}|}{M^3|\mathcal{B}|^2T}}$ for uniform price. Our results are shown in Tables~\ref{table: learning dynamics bandit info} and \ref{table: learning dynamics bandit uniform}. While most of our findings are consistent with the full information market dynamics, it is clear that the market dynamics have not yet converged for larger $M, |\mathcal{B}|$, especially for the uniform price auction, as indicated by the large regret, bid ratios, competitive ratios, and revenue gaps. Consequently, we analyze the PAB and uniform price equilibria exclusively under full information, as presented in Section~\ref{sec: experiments}.}

\begin{table}[ht]
    \centering
    \begin{minipage}[t]{0.48\linewidth} 
        \caption{Summary of the Bidding Dynamics (Bandit, Algorithm~\ref{alg: Decoupled Exponential Weights - Path Kernels})}
        \label{table: learning dynamics bandit info}
        \footnotesize
        \begin{tabular}{lcccc}
        \toprule
         & \multicolumn{2}{c}{$|\mathcal{B}| = 10$} & \multicolumn{2}{c}{$|\mathcal{B}| = 20$} \\
        \cmidrule(lr){2-3} \cmidrule(lr){4-5}
        Metric & $N = 3$ & $N = 5$ & $N = 3$ & $N = 5$ \\
        \midrule
        \multicolumn{5}{c}{\textbf{For \( M = 1 \)}} \\
        Regret  & .8/2.5 & .5/2.5 & .3/1.9 & .8/2.7 \\
        Welfare Gap  & .40/11 & .50/4 & .86/14 & 1.6/13 \\
        Revenue Gap  & 28/78 & 18/47 & 38/69 & 24/41 \\
        CR Gap  & 1.6/56 & 2.2/25 & 1.6/53 & 8.2/52 \\
        Runtime (sec) & 321 & 511 & 415 & 879 \\
        \midrule
        \multicolumn{5}{c}{\textbf{For \( M = 5 \)}} \\
        Regret  & .4/.7 & .6/1.2 & .5/1.2 & .9/1.4 \\
        Welfare Gap  & 2.5/7.7 & 2.7/5.4 & 2.1/5.8 & 2.6/5.1 \\
        Revenue Gap  & 37/49 & 21/27 & 37/47 & 22/28 \\
        $b_{(1)}/b_{(M)}$ & 1.17/1.26 & 1.12/1.17 & 1.10/1.17 & 1.06/1.11 \\
        CR Gap  & 3.6/23 & 13/20 & 8.7/21 & 12/23 \\
        Runtime (sec) & 1844 & 2843 & 2660 & 5123 \\
        \midrule
        \multicolumn{5}{c}{\textbf{For \( M = 10 \)}} \\
        Regret  & 3.9/16 & 2.3/5.8 & 9.6/20.5 & 4.3/8.6 \\
        Welfare Gap  & 3.2/6.9 & 2.6/4.4 & 3.2/6.6 & 2.6/4.4 \\
        Revenue Gap  & 34/44 & 20/25 & 37/43 & 21/26 \\
        $b_{(1)}/b_{(M)}$ & 1.21/1.26 & 1.14/1.18 & 1.12/1.16 & 1.08/1.1 \\
        CR Gap  & 4.6/18 & 6/15 & 9/23 & 10/20 \\
        Runtime (sec) & 3912 & 5222 & 5098 & 9731 \\
        \bottomrule
        \end{tabular}
    \end{minipage}\hfill
    \begin{minipage}[t]{0.48\linewidth} 
        \caption{Summary of the Bidding Dynamics (Uniform Price, Bandit \cite{brânzei2023online})}
        \label{table: learning dynamics bandit uniform}
        \footnotesize
        \begin{tabular}{lcccc}
        \toprule
         & \multicolumn{2}{c}{$|\mathcal{B}| = 10$} & \multicolumn{2}{c}{$|\mathcal{B}| = 20$} \\
        \cmidrule(lr){2-3} \cmidrule(lr){4-5}
        Metric & $N = 3$ & $N = 5$ & $N = 3$ & $N = 5$ \\
        \midrule
        \multicolumn{5}{c}{\textbf{For \( M = 1 \)}} \\
        Regret&	5.1/13.3&	3.7/6.3&	6.3/12&	5.4/7.4\\
        Welfare Gap&	1.5/26&	.8/15&	3.1/33&	17/32\\
        Revenue Gap&	45/70&	38/50&	49/67&	38/45\\
        CR Gap&	39/72&	42/64&	35/72&	66/77\\
        Runtime (sec) & 166 & 241 & 317 & 466\\
        \midrule
        \multicolumn{5}{c}{\textbf{For \( M = 5 \)}} \\
        Regret&	22/35&	9.1/11&	12/13&	11/12\\
        Welfare Gap&	2.3/7.8&	3.5/5.7&	4.4/6.1&	5.6/7.3\\
        Revenue Gap&	80/85&	47/58&	60/66&	43/48\\
        $b_{(1)}/b_{(M)}$ &	7.1/16&	2.2/4.5&	3/3.9&	1.7/2\\
        CR Gap&	19/24&	19/23&	14/16&	26/30\\
        Runtime (sec) & 1531 & 2612 & 5323 & 8777\\
        \midrule
        \multicolumn{5}{c}{\textbf{For \( M = 10 \)}} \\
        Regret&	98/147&	36/39&	53/60&	40/43\\
        Welfare Gap&	4.6/9.3&	3.9/6.9&	5.3/9.2&	5.4/6.5\\
        Revenue Gap&	98/98&	59/63&	89/90&	56/60\\
        $b_{(1)}/b_{(M)}$ &	43/46&	4.8/5.5&	21/24&	3.4/3.6\\
        CR Gap&	29/34&	26/30&	19/21&	30/32\\
        Runtime (sec) & 2643 & 5000 & 9783 & 14532\\
        \bottomrule
        \end{tabular}
    \end{minipage}
\end{table}

\newpage

\TITLE{Learning in Repeated Multi-Unit Pay-As-Bid Auctions}

\ARTICLEAUTHORS{
\AUTHOR{Rigel Galgana}
\AFF{Operations Research Center, Massachusetts Institute of Technology, \EMAIL{rgalgana@mit.edu}, \URL{}}
\AUTHOR{Negin Golrezaei}
\AFF{Sloan School of Management, Massachusetts Institute of Technology,  \EMAIL{golrezaei@mit.edu}, \URL{}}
} 

\begin{center}
    {\LARGE{Learning in Repeated Multi-Unit Pay-As-Bid Auctions}}

    {\small{Rigel Galgana}}

    {\scriptsize{Operations Research Center, Massachusetts Institute of Technology, rgalgana@mit.edu}}

    {\small{Negin Golrezaei}}

    {\scriptsize{Sloan School of Management, Massachusetts Institute of Technology, golrezaei@mit.edu}}
\end{center}

\maketitle

In this online accompaniment, we include all of the omitted proofs of several of our key results as well as additional discussion and experiments regarding algorithm implementation and market dynamics. We conclude with a discussion regarding both the feasibility and practicality of the PAB versus Uniform Price auction formats.

\section{Online Appendix - Missing Proofs}

In this section, we include proofs of several of our key results. We first precisely characterize the PAB equilibrium under our $c = \lfloor v_{(M)}\rfloor_\delta = c_{-n}$ competitiveness assumption. Second, we flesh out the details of our OMD algorithm as well as its full proof, regret analysis, and complexity analysis. Lastly, we prove our time-varying valuations generalization of our decoupled exponential weights algorithm.

\subsection{Proof of Theorem \ref{thm: PNE existence}: Existence of an Approximately Efficient PNE}
\label{sec: equilibrium proof}

\textsc{Theorem 1} \textbf{(Existence of an Apprixxmately Efficient PNE)}

\emph{Define the clearing price $c = \lfloor v_{(M)} \rfloor_{\delta}$ to be the $M$'th largest valuation among all bidders, denoted by $v_{(M)}$,  rounded down to the nearest multiple of $\delta$, and similarly, define $c_{-n} = \lfloor v_{(-n, M)} \rfloor_{\delta}$ to be the rounded $M$'th largest valuation among all bidders except $n$. If $c = c_{-n}$ for all $n\in [N]$ and ties are broken in favor of higher indexed bidders, then there exists a PNE $(\bm{b}_1,\ldots,\bm{b}_N)$ where each bidder $n\in [N]$:  
    \begin{enumerate}
        \item submits bids of either all $c$ or all $c + \delta$ for units such that $v_{n,m} \geq c + \delta$,
        \item submit bids of $c$ for all units such that $v_{n,m} \in [c, c + \delta)$,
        \item submit bids smaller than $c$ for all other units.
    \end{enumerate}
    Moreover, this PNE is $M\delta$-approximately welfare optimal:
    \begin{align*}
        \sum_{m=1}^M v_{(m)} - \sum_{n=1}^N \sum_{m=1}^{x(\bm{b}_n, \bm{b}_{-n})} v_{n,m} \leq M\delta\,.
    \end{align*}}

{\color{black}
\begin{proof}{Proof of Theorem~\ref{thm: PNE existence}}

To prove Theorem~\ref{thm: PNE existence}, we proceed in two steps. First, we demonstrate that if all other bidders adhere to the three properties outlined in Theorem~\ref{thm: PNE existence}, then bidder $n$'s optimal bids must also satisfy these properties. Afterwards, using a monotonicity argument, we show that under our particular deterministic tie-breaking rule, there exists a specific configuration of bids—either $c$ or $c + \delta$—among each bidder’s winning bids that constitutes a PNE.

\textbf{First Part of the Proof.} {We first show the three properties. The first property implies that if for every bidder $i\ne n$, we have  bids of either all $c$ or all $c + \delta$ for units such that $v_{i,m} \geq c + \delta$, then bidder $n$ also submits bids of  either all $c$ or all $c + \delta$ for units such that $v_{n,m} \geq c + \delta$. This 
follows immediately from Lemma~\ref{lem: PNE uniform bidding} and recalling the definition of $c$.

 To demonstrate the second property, we argue that there is no incentive for bidder $n$ to decrease their winning bids below $c$, specifically for any unit $m$ with $v_{n, m} \in [c, c+\delta)$. Given that there are at least $M$ other valuations equal to or greater than $c = c_{-n}$, and considering the PNE characterization which states that for any bidder $i \neq n$, all units with a valuation at or above the clearing price (i.e., $v_{i, m} \geq c$, $i \in [N], i\ne n$) must be accompanied by a bid of either $c$ or $c + \delta$, it follows that there are at least $M$ bids of at least $c = c_{-n}$ submitted by other bidders. Therefore, reducing any of bidder's $n$ winning bids below $c$ would not result in winning a unit, making it sub-optimal.

Lastly, to show the final property, we show that there exists no incentive for bidder $n$ to increase their bid for the remaining items (i.e., any item $m$ with $v_{n, m}< c$) to at least $c$, we recall the definition of $c= c_{-n}$ being the $M$'th largest valuation among all bidders except bidder $n$. Thus, bidding above $c$ for these units violates the no-overbidding assumption.}

\textbf{Second Part of the Proof.} Now that we have shown the three properties, we fully characterize the PNE w.r.t. each bidders' largest bids. In particular, consider the perspective of bidder $n$, who has tie-break priority over bidders $1,\ldots,n-1$ but behind $n+1,\ldots,N$. Thus, any bids of $c+\delta$ submitted by the first $n-1$ bidders, and any bids of $c$ or $c + \delta$ submitted by bidders $n+1,\ldots,N$, take priority over any bids of $c$ submitted by bidder $n$. Let $M_n$ denote the number of bids of $c+\delta$ submitted by bidder $n$, which by our PNE characterization, is between $0$ and $\sum_{m=1}^M 1_{v_{n, m} \geq c + \delta}$. Fixing $M_{1:n-1}$ and $M_{n+1:N}$, we claim that the optimal $M_n$ is precisely either 0 or $\sum_{m=1}^M 1_{v_{n, m} \geq c + \delta}$---they either submit bids of all $c$ or all $c+\delta$ for all items they value at least $c + \delta$. 

To show this, notice that bidder $n$ can win all of the items they value at least $c + \delta$ by bidding at $c + \delta$, as there are fewer than $M$ values at least $c+\delta$ across all bidders. In addition to this, they can win some number of items by bidding at $c$ for all items with value at least $c$ via tie-break. To be more specific, 
there are $M$ items with $\sum_{n' \leq n} M_{n'}$ bids of $c + \delta$ and $\sum_{n' > n} \sum_{m=1}^M 1_{v_{n',m} \geq c}$ bids of at least $c$ (by the PNE characterization) that take priority over any bids of $c$ submitted by bidder $n$, where we note that the number of bids of $c$ submitted by bidder $n$ is at most $\sum_{m=1}^M 1_{v_{n,m} \geq c}$. Thus, bidder $n$'s allocation and utility as a function of $M_n$ for any fixed $\bm{M}_{-n} = (M_{-1},\ldots,M_{-(n-1)},M_{-(n+1)},\ldots,M_{-N})$ is given by:

\begin{align*}
    x_n(M_n | \bm{M}_{-n}) &= \min\left(\sum_{m=1}^M 1_{v_{n,m} \geq c}~,~ M_n + \max\big(0, \theta(\bm{M}_{-n}) - M_n\big)\right)\\
    \hspace{1mm} \mu_n(M_n | \bm{M}_{-n}) &= - M_n\delta +\sum_{m=1}^{x_n(M_n | \bm{M}_{-n})} (v_{n,m} - c)\,,
\end{align*}
where $\theta(\bm{M}_{-n}) = M - \sum_{n' < n} M_{n'} - \sum_{n' > n} \sum_{m=1}^M 1_{v_{n',m} \geq c}$. Here, $\min\left(\sum_{m=1}^M 1_{v_{n,m} \geq c}, \theta(\bm{M}_{-n})\right)$ denotes the number of units bidder $n$ would have won in tie-break by submitting $M_n = 0$ bids of $c + \delta$.
Now, let's consider the general case where $M_n\ge 0$. For  all $M_n \geq \theta(\bm{M}_{-n})$, the $\max(0, \theta(\bm{M}_{-n}) - M_n)$ term in the allocation function is 0, and thus, the allocation function  increases linearly in $M_n$ until $M_n = \sum_{m=1}^M 1_{v_{n,m} \geq c+\delta}$. In contrast, for $M_n <\theta(\bm{M}_{-n})$, the second term (i.e., $\max(0, \theta(\bm{M}_{-n}) - M_n)$) is non-zero, and the $M_n$ within the summation of the second term cancels out with the first term of $M_n$. Thus, the allocation function $x_n(M_n | \bm{M}_{-n}) = \theta(\bm{M}_{-n})$ is constant for all $M_n \leq \theta(\bm{M}_{-n})$, at which point it becomes precisely $x_n(M_n | \bm{M}_{-n}) = \theta(\bm{M}_{-n}) = M_n$ until $M_n = \sum_{m=1}^M 1_{v_{n, m} \geq c + \delta}$. This reflects the fact that each additional bid at $c + \delta$ submitted by bidder $n$ consumes an item that could have been won in tie-break at price $c$, which we illustrate in Figure~\ref{fig: allocation visualization PNE}. From Figure~\ref{fig: allocation utility vs mn}, we see that the optimal $M_{-n}$ is always achieved at $M_{-n} \in \{0, \sum_{m=1}^M 1_{v_{n,m} \geq c + \delta}\}$.

\begin{figure}
    \centering
    \includegraphics[scale=0.5, trim={0 0 240 0}, clip]{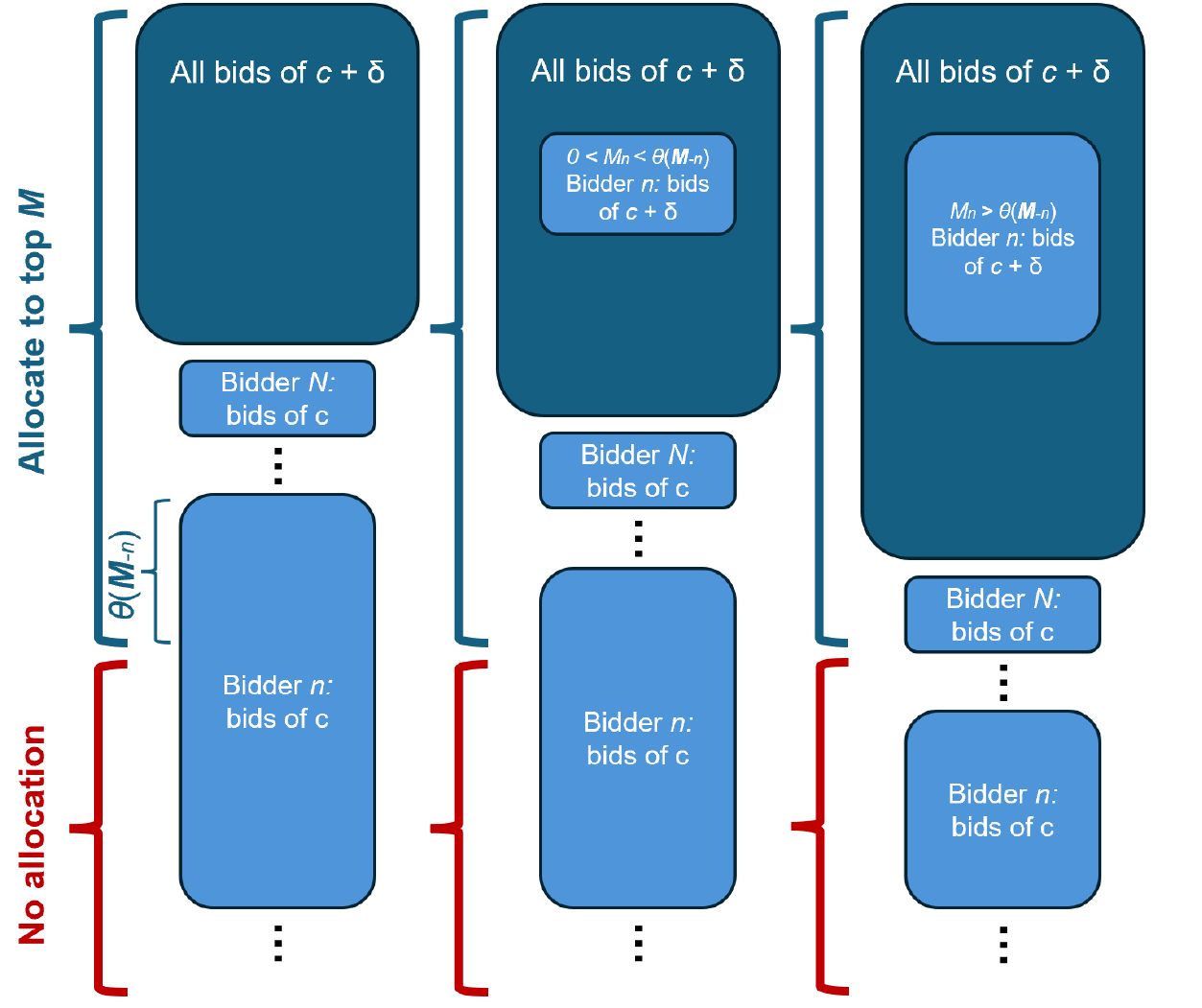}
    \includegraphics[scale=0.5, trim={600 0 0 0}, clip]{figs/Figure10.pdf}
    \caption{\textbf{Allocation as a Function of $M_{n}$. } We visualize the allocation of bidder $n$ when behaving according to the PNE as prescribed by Theorem~\ref{thm: PNE existence} and submitting $M_n$ bids of $c + \delta$. The left (resp. right) figure assumes that $\theta(\bm{M}_{-n})$ is at least (resp. at most) $\sum_{m=1}^M 1_{v_{n,m} \geq c}$.}
    \label{fig: allocation visualization PNE}
\end{figure}

\begin{figure}
    \centering
    \begin{subfigure}[t]{0.45\textwidth}
        \centering
        \begin{tikzpicture}[scale=0.74]
            \begin{axis}[
                title={Allocation and Utility w.r.t. $M_n$},
                xlabel={},
                ylabel={},
                legend style={at={(0.5,-0.2)}, anchor=north, legend columns=-1},
                domain=0:5,
                samples=100,
                axis x line=middle,
                axis y line=middle,
                enlarge x limits=false,
                enlarge y limits=false,
                ymin=0, ymax=9,
                xmin=0, xmax=10.5,
                xtick={0,2,5},
                xticklabels={0, $\theta(\bm{M}_{-n})$, $\sum_{m=1}^M 1_{v_{n,m} \geq c+\delta}$},
                ytick={0,1, 1.5, 2 ,5, 6.7},
                yticklabels={0, $\mu_1$, $\mu_2$, $\theta(\bm{M}_{-n})$, $\sum_{m=1}^M 1_{v_{n,m} \geq c+\delta}$, $\sum_{m=1}^M 1_{v_{n,m} \geq c}$},
                legend cell align={left}
            ]

            \pgfmathsetmacro{\c}{2}
            \pgfmathsetmacro{\muzero}{1}
            
            \addplot[domain=0:\c, thick, blue] {2};
            \addplot[domain=\c:5, thick, blue] {x} node[right,pos=1] {$x_n(M_n | \bm{M}_{-n})$};
            
            \addplot[domain=0:\c, thick, red] {1 - 0.4 * x};
            \addplot[domain=0:\c, dashed, red] {1.5 - 0.4 * x};
            \addplot[domain=\c:5, thick, red, samples=100] {1 - 0.4 * x + 1.6*(sqrt(x-1) - sqrt(\c - 1))} node[right,pos=1] {$\mu^1_n(M_n | \bm{M}_{-n})$};
            \addplot[domain=\c:5, dashed, red, samples=100] {4*(1.5 - 0.4 * x + 1.6*(sqrt(x-1) - sqrt(\c - 1))) - 3*(1.5 - 0.4 * \c + 1.6*(sqrt(\c-1) - sqrt(\c - 1)))} node[right,pos=1] {$\mu^2_n(M_n | \bm{M}_{-n})$};
            
            \addlegendimage{line legend,blue,thick}
            \addlegendimage{line legend,red,thick}
            
            \end{axis}
        \end{tikzpicture}
    \end{subfigure}
    \begin{subfigure}[t]{0.45\textwidth}
        \centering
        \begin{tikzpicture}[scale=0.74]
            \begin{axis}[
                title={Allocation and Utility w.r.t. $M_n$},
                xlabel={},
                ylabel={},
                legend style={at={(0.5,-0.2)}, anchor=north, legend columns=-1},
                domain=0:5,
                samples=100,
                axis x line=middle,
                axis y line=middle,
                enlarge x limits=false,
                enlarge y limits=false,
                ymin=0, ymax=6.8,
                xmin=0, xmax=7.8,
                xtick={0,3.7},
                xticklabels={0, $\sum_{m=1}^M 1_{v_{n,m} \geq c+\delta}$},
                ytick={0,1.5, 2.5,3.75,5, 6.2},
                yticklabels={0, $\mu^1_0$, $\mu^2_0$, $\sum_{m=1}^M 1_{v_{n,m} \geq c+\delta}$, $\sum_{m=1}^M 1_{v_{n,m} \geq c}$, $\theta(\bm{M}_{-n})$},
                legend cell align={left}
            ]

            \pgfmathsetmacro{\c}{2}
            \pgfmathsetmacro{\muzero}{1}
            
            \addplot[domain=0:5, thick, blue] {5}  node[right,pos=1] {$x_n(M_n | \bm{M}_{-n})$};
            
            \addplot[domain=0:3.7, thick, red] {1.5 - 0.4 * x} node[above right,pos=1] {$\mu^1_n(M_n | \bm{M}_{-n})$};
            \addplot[domain=0:3.7, thick, red] {2.5 - 0.4 * x} node[above right,pos=1] {$\mu^2_n(M_n | \bm{M}_{-n})$};
            
            \addlegendimage{line legend,blue,thick}
            \addlegendimage{line legend,red,thick}
            
            \end{axis}
        \end{tikzpicture}
    \end{subfigure}
    \caption{\textbf{Allocation $x_n$ and Utility $\mu_n$ as a function of $M_n$. } Here, we plot the allocation and utility functions of bidder $n$ when fixing $\bm{M}_{-n}$ for the cases where $\theta(\bm{M}_{-n}) < \sum_{m=1}^M 1_{v_{n,m} \geq c}$ (left figure) and $\theta(\bm{M}_{-n}) \geq \sum_{m=1}^M 1_{v_{n,m} \geq c}$ (right figure) respectively. In the left figure, we let $\mu_n^1, \mu_n^2$ denote the utility functions w.r.t. two different valuation vectors such that the optima occurs at $M_n = 0$ and $M_n = \sum_{m=1}^M 1_{v_{n,m} \geq c+\delta}$ respectively. In the right figure, the allocation function is constant yielding a  slope $-\delta$ for the utility function, which is trivially maximized at $M_n = 0$.}
    \label{fig: allocation utility vs mn}
\end{figure}

    Now that we have shown that the optimal number of bids to submit at $c + \delta$ of each bidder is either $M_n = 0$ or $M_n = \sum_{m=1}^M 1_{v_{n, m} \geq c + \delta}$, we finish the argument by claiming that the utility corresponding to $M_n = 0$ is weakly decreasing in $\bm{M}_{-n}$. This is because $x_n(\cdot | \bm{M}_{-n})$, and similarly $\mu_n(\cdot | \bm{M}_{-n})$, are weakly decreasing in these quantities. Because of this monotonicity, agents can run best response dynamics beginning with $(M_1,\ldots,M_N) = (0, \ldots, 0)$ and converge to a PNE w.r.t. $M_1,\ldots,M_N$. That is, each agent that switches from $M_n = 0$ to $M_n = \sum_{m=1}^M 1_{v_{n, m} \geq c + \delta}$ can only incentivize other bidders to also switch away from 0 and never towards 0.
    
\end{proof}

}

\subsection{Lemmas \ref{lem: QSpace Equivalence} and \ref{lem: Online Linear Optimization}, and their Proofs} \label{sec:QSpace Equivalence} 

Now, we complete the description and analysis of our OMD based algorithm. In Algorithm~\ref{alg: OMD}, we require that the space of all possible node probability measure $\bm{q}$ encompasses the set of node probability measures that correspond to any policy $\bm{\pi}$ over our DP graph.

\begin{lemma}[$\mathcal{Q}$-Space Equivalence]
    \label{lem: QSpace Equivalence}
    Let \[\Pi = \Big\{{\pi} \in [0,1]^{M\times |\mathcal B|\times |\mathcal B|}: \pi((m, b), b') = 0 ~~\forall b' > b, m\in [M], \sum_{b' \leq b} \pi((m, b), b') = 1, m\in [M]\Big\}\] denote the space of policies on our DP graph. With a slight abuse of notation, for any $\pi\in  \Pi$, define 
    \[q(\pi) = \{\mathbf{q} \in [0, 1]^{M\times |\mathcal B|}: \forall b \in \mathcal{B},  q_1(b) =\pi((0, b_0), b), q_{m+1}(b) = \sum_{b' \in \mathcal B} q_m(b')\pi((m, b'), b), m\in[M-1]\}\,\] as the node probabilities induced by $\pi$. Here, $b_0= \max \mathcal B$. Let $\mathcal{Q}_{\Pi} = \cup_{\pi \in \Pi} q(\pi)$. Then,   $\mathcal{Q}_{\Pi}$ is equivalent to the set $\mathcal{Q}$ where $\mathcal{Q}$ is defined in Equation \eqref{eq:Q}. 
\end{lemma}

Lemma \ref{lem: QSpace Equivalence} establishes that during the execution of Algorithm \ref{alg: OMD}, we can focus on the node probabilities in set $\mathcal{Q}$ without loss of generality. We recall that within $\mathcal{Q}$, the stochastic dominance conditions are enforced solely over node probabilities across layers. 
We now argue that we only need to consider optimizing over $\mathcal{Q}$ as opposed to the space of policies $\Pi$, as the regret can be rewritten strictly in terms of $\bm{q}$, independently of the corresponding $\bm{\pi}$.

\begin{lemma}
    \label{lem: Online Linear Optimization}
     Any sequence of policies $\bm{\pi}^1,\ldots,\bm{\pi}^\Nround$ over our DP graph with associated node probability measures $\bm{q}^1,\ldots,\bm{q}^\Nround$ has discretized regret $\textsc{Regret}_{\mathcal{B}} = \max_{\bm{q} \in \mathcal{Q}} \sum_{\nround=1}^\Nround \langle \bm{q} - \bm{q}^\nround, \bm{w}^\nround\rangle$. Here, $\bm{w}^\nround = \{w^\nround_m(b)\}_{m \in [M], b \in \mathcal{B}}$ represents vector of the round $\nround$ rewards for all possible $(m, b)$ unit-bid value pairs.
 \end{lemma}
      
\subsubsection{Proof of Lemma \ref{lem: QSpace Equivalence}}

In order to show equivalence, we show that (1) for any $\pi \in \Pi$, that $q(\pi) \in \mathcal{Q}$ and (2) for any $\bm{q} \in \mathcal{Q}$, there exists a $\pi \in \Pi$ such that $q(\pi) = \bm{q}$. We first prove (1). To do this, we simply need to check that for a given $\pi \in \Pi$, that $q^\pi = q(\pi)$ satisfies the constraints prescribed by $\mathcal{Q}$.

    The non-negativity constraint holds trivially as each $\pi((m, b), b')$ is non-negative. Since all $q^\pi_1(b) = \pi((0, \max\mathcal{B}), b) \geq 0$ for all $b \in \mathcal{B}$, by induction, $q^\pi_{m+1}(b) = \sum_{b" \geq b} q^\pi_m(b") \pi((m, b"), b)$ is also non-negative.
    
    Now we prove that each layer $m$ sums to 1, i.e., $\sum_{b \in \mathcal{B}} q^\pi_m(b) = 1$. Since $\sum_{b \in \mathcal{B}} q^\pi_1(b)$, the policy has total node probability 1 in the first layer, we can prove $\sum_{b \in \mathcal{B}} q^\pi_m(b) = 1$, that the policy has total node probability 1 in the $m$'th layer, via induction. This follows immediately from the fact that the DP graph is layered, i.e., edges exist only from nodes in layer $m$ to nodes in layer $m+1$, thus the only edges leading to layer $m+1$ are from layer $m$, in which there are no other edges. Hence, the total node probability in layer $m+1$ must be exactly that of layer $m$. More formally, we have:
    \begin{align*}
        \sum_{b \in \mathcal{B}} q^\pi_{m+1}(b) = \sum_{b \in \mathcal{B}} \sum_{b" \geq b} q^\pi_m(b") \pi((m, b"), b) = \sum_{b" \in \mathcal{B}} q^\pi_m(b") \sum_{b \leq b"} \pi((m, b"), b) = \sum_{b" \in \mathcal{B}} q^\pi_m(b")\,.
    \end{align*}

     To show the stochastic domination constraint $\sum_{b \leq b'} q^\pi_{m+1}(b) \geq \sum_{b \leq b'} q^\pi_m(b)$, we use the bid monotonicity constraint; i.e., the fact that the edges between layers are only from larger bids to (weakly) smaller bids. Recall that $\pi((m,b'), b")$ is the probability of transitioning from unit-bid value pair $(m, b')$ to $(m+1, b")$ and that the only edges leading to $(m+1, b")$ come from nodes $(m, b')$ for $b' \geq b"$. Then, we have:
     \begin{align*}
         \sum_{b \leq b'} q^\pi_{m+1}(b) &= \sum_{b \leq b'} \sum_{b" \geq b} q^\pi_m(b")\pi((m, b"), b) \\
         &=\sum_{b" > b'} q^\pi_m(b") \sum_{b \leq b'} \pi((m, b"), b) + \sum_{b" \leq b'} q^\pi_m(b") \sum_{b \leq b"} \pi((m, b"), b)\\
         &= \sum_{b" > b'} q^\pi_m(b") \sum_{b \leq b'} \pi((m, b"), b) + \sum_{b \leq b'} q^\pi_m(b)\\
         &\geq \sum_{b \leq b'} q^\pi_m(b)\,.
     \end{align*}
     Hence, we have shown that for any $\pi \in \Pi$, that the corresponding $q(\pi) \in \mathcal{Q}$.
     
     Now we show the other direction (2), that for any $\bm{q} \in \mathcal{Q}$, there exists a $\pi \in \Pi$ such that $q(\pi) = \bm{q}$. We proceed by showing that for all $m, b^*$, there exists $\{\pi((m, b), b')\}_{b, b' \in \mathcal{B}}$ such that the following conditions hold:
     \begin{enumerate}
         \item $\pi((m, b), b') \geq 0$ for all $b, b' \geq b^*$.
         \item $\pi((m, b), b') = 0$ for all $b' > b \geq b^*$.
         \item $\sum_{b' \leq b, b' \geq b^*} \pi((m, b), b') \leq 1$ for all $b^* \in \mathcal{B}$, with equality if and only if $b^* =  b_{\min}$ where $b_{\min} = \min \mathcal{B}$.
         \item $\sum_{b' \geq b^*} \sum_{b \geq b'} q_m(b)\pi((m, b), b') = \sum_{b' \geq b^*} q_{m+1}(b')$.
     \end{enumerate}
Let $\Pi(b^*; \mathbf{q})$, $b^*\in \mathcal B$, be the set of all policies under which the four conditions hold at $b^*$ and ${\mathbf q} \in \mathcal Q$. 
     
    These conditions trivially hold for $m = 0$, as we can set $\pi((0, \max \mathcal{B}), b) = q_1(b)$ and $\pi((0, b), b') = \textbf{1}_{b = b'}$. To solve for general $m$, we must show that there exists $\{\pi((m, b), b')\}_{b, b' \in \mathcal{B}}$ that satisfies the constraints prescribed by $\Pi$ and that $\sum_{b \geq b'} q_m(b)\pi((m, b), b') = q_{m+1}(b')$ for all $b' \in \mathcal{B}$. In order to do this, we show that conditions (1), (2), (3), and (4) for each $b^* \in \mathcal{B}$. In particular, if we show conditions (1) and (2) for $b^* = b_{\min}$, then we have already satisfied the first two conditions of $\Pi$. If we show that (3) holds for $b^* = b_{\min}$, then by condition (1), then (3) holds for all $b^* \in \mathcal{B}$ as well, as the summation only includes fewer terms as $b^*$ increases. Similarly, if we show condition (4) holds for two adjacent values of $b_-^* < b^*$, then we have that $\sum_{b \geq b' \geq b^*} q_m(b)\pi((m, b), b') = q_{m+1}(b^*)$. Thus, if condition (4) holds for all possible pairs of adjacent bid values, then we have that $\sum_{b \geq b'} q_m(b)\pi((m, b), b') = q_{m+1}(b')$ for all $b'$. These observations suggest use of induction over $b^*$, and indeed, we begin by showing that these conditions hold for $b^* = b_{\min}$. We then show that this implies that the conditions hold for the next smallest value of $b^*$, which would complete the induction proof.

    \textbf{Base Case}: Recall $b^* = b_{\min}$. We now show that there exists $\{\pi((m, b), b')\}_{b, b' \in \mathcal{B}}$ satisfying all four conditions. For any $m\in[M]$, let we set $\pi((m, b), b') = \textbf{1}_{b = b'}$. Then, 
 condition (4) is clearly satisfied: 
    \begin{align*}
        \sum_{b' \geq b^*} \sum_{b \geq b'} q_m(b)\pi((m, b), b') = \sum_{b' \geq b^*} q_{m+1}(b') \leftrightarrow \sum_{b \in \mathcal{B}} q_m(b) \sum_{b' \leq b} \pi((m, b), b') = \sum_{b \in \mathcal{B}} q_{m+1}(b) = 1\,.
    \end{align*}
    It is also easy to check that conditions (1)-(3) are also satisfied when we set $\pi((m, b), b') = \textbf{1}_{b = b'}$ for any $m$. This shows that $\Pi(b^*; {\mathbf{q}})$ is non-empty, as desired.

    \textbf{Recursive Case}:
    For any $b\in \mathcal B$, let $b_-$ be the largest $b' \in \mathcal B$, which is strictly smaller than $b$. Here, we assume that 
    $\Pi(b^*_-; \bf{q})$ is not empty, and under this assumption, we show that set $\Pi(b^*; \bf{q})$ is not empty, where $\Pi(b^*; {\bf{q}}) \subseteq \Pi(b^*_-; \bf{q})$. 
    Let us start with condition (4). 
    We would like to show that there exists a $\bm{\pi}$ that satisfies condition (4) at $b^*$ along with the other three conditions. By the induction assumption, we have  
    \begin{align*}
        &\sum_{b' \geq b_-^*} \sum_{b \geq b'} q_m(b)\pi((m, b), b') = \sum_{b' \geq b_-^*} q_{m+1}(b') \to\\
        &\sum_{b' \geq b^*} \sum_{b \geq b'} q_m(b)\pi((m, b), b') + \sum_{b \geq b_-^*} q_m(b)\pi((m, b), b_-^*) = \sum_{b' \geq b^*} q_{m+1}(b') + q_{m+1}(b_-^*) \to\\
        &\sum_{b' \geq b^*} \sum_{b \geq b'} q_m(b)\pi((m, b), b') = \sum_{b' \geq b^*} q_{m+1}(b') + \left[q_{m+1}(b_-^*) - \sum_{b \geq b_-^*} q_m(b)\pi((m, b), b_-^*)\right] \to\\
        &\sum_{b \geq b^*} q_m(b) \sum_{b' \leq b; b' \geq b^*} \pi((m, b), b') = \sum_{b' \geq b^*} q_{m+1}(b') + \left[q_{m+1}(b_-^*) - \sum_{b \geq b_-^*} q_m(b)\pi((m, b), b_-^*)\right] 
    \end{align*}
    Thus, we can satisfy condition (4) if $q_{m+1}(b_-^*) = \sum_{b \geq b_-^*} q_m(b)\pi((m, b), b_-^*)$. We now observe that the latter summation depends linearly (and hence, continuously) in the values of $\pi((m, b), b_-^*)$. If we can show that there exists an assignment of these variables that satisfy $q_{m+1}(b_-^*) \geq \sum_{b \geq b_-^*} q_m(b)\pi((m, b), b_-^*)$ and also $q_{m+1}(b_-^*) \leq \sum_{b \geq b_-^*} q_m(b)\pi((m, b), b_-^*)$, then by the intermediate value theorem, there must be some assignment that achieves exact equality.

    In order to show the first inequality, notice that if we set $\pi((m, b), b_-^*) = 1 - \sum_{b' < b_-^*} \pi((m, b), b')$ for all $b \geq b_-^*$ (this is required in order to guarantee conditions (1) and (3) are satisfied), then:
    \begin{align*}
        \sum_{b \geq b_-^*} q_m(b)\pi((m, b), b_-^*) &= \sum_{b \geq b_-^*} q_m(b) - \sum_{b \geq b_-^*} q_m(b)\sum_{b' < b_-^*} \pi((m, b), b') \\
        &= \sum_{b \geq b_-^*} q_{m}(b) - \sum_{b \in \mathcal{B}} q_m(b) \sum_{b' < b_-^*} \pi((m, b), b') + \sum_{b < b_-^*} q_m(b) \sum_{b' < b_-^*} \pi((m, b), b')\\
        &= \sum_{b \geq b_-^*} q_{m}(b) - \sum_{b' < b_-^*} q_{m+1}(b) + \sum_{b < b_-^*} q_m(b) \sum_{b' < b_-^*} \pi((m, b), b')\\
        &= \sum_{b \geq b_-^*} q_{m}(b) - \sum_{b' < b_-^*} q_{m+1}(b) + \sum_{b < b_-^*} q_m(b)\\
        &\geq \sum_{b \geq b_-^*} q_{m+1}(b) - \sum_{b' < b_-^*} q_{m+1}(b) + \sum_{b < b_-^*} q_{m+1}(b)\\
        &= \sum_{b \geq b_-^*} q_{m+1}(b)\\
        &\geq q_{m+1}(b_-^*)\,.
    \end{align*}
    Here, the third equality follows from the (strong) inductive hypothesis, and the first inequality is a result of the stochastic domination constraint in $\mathcal{Q}$. We also note that the values $\sum_{b' < b_-^*} \pi((m, b), b')$ have already been fixed
    as these were required to satisfy condition (4) in the previous iterates, and as condition (3) holds for $b^*_-$ by the inductive hypothesis, then $1 - \sum_{b' < b_-^*} \pi((m, b), b') \geq 0$. Conversely, if we set $\pi((m, b), b_-^*) = 0$ for all $b \geq b_-^*$, then:
    \begin{align*}
        \sum_{b \geq b_-^*} q_m(b)\pi((m, b), b_-^*) = 0 \leq q_{m+1}(b_-^*)\,.
    \end{align*}
    As the sum $\sum_{b \geq b_-^*} q_m(b)\pi((m, b), b_-^*)$ linearly (thus, continuously) depends on the values of $\pi((m, b), b_-^*)$, by the intermediate value theorem, there exists an assignment of $\{\pi((m, b), b_-^*)\}_{b \geq b_-^*}$ with each $\pi((m, b), b_-^*) \in [0, 1 - \sum_{b' < b_-^*} \pi((m, b), b')]$ such that the sum is precisely equal to $q_{m+1}(b_-^*) \in [0, 1]$. Now we observe that these values of $\pi((m, b), b_-^*) \in [0, 1 - \sum_{b' < b_-^*} \pi((m, b), b')]$ do not violate conditions (1), (2), or (3). Furthermore, note that any $\bm{\pi} \in \Pi$ also satisfied conditions (1), (2), and (3) under $b^*_-$ for $\{\pi((m, b), b')\}_{b \geq b_-^*, b' \leq b_-^*}$, then the assignment to $\{\pi((m, b), b')\}_{b \geq b_-^*, b' < b_-^*}$ will not violate these conditions as our new constraint on the variables $\{\pi((m, b), b_-^*)\}_{b \geq b_-^*}$ is independent of the values of $\{\pi((m, b), b')\}_{b \geq b_-^*, b' < b_-^*}$. Thus, the set $\Pi(b^*)$ is non-empty:
    \begin{align*}
        \Pi(b^*) = \{\{\pi((m, b), b')\}_{b, b' \in \mathcal{B}} \in \Pi(b_-^*): \sum_{b \geq b_-^*} q_m(b)\pi((m, b), b_-^*) = q_{m+1}(b_-^*)\} \neq \emptyset
    \end{align*}    
    With this, we have proven via induction that our four conditions hold for all $b^* \in \mathcal{B}$, implying that for a fixed $m$, every constraint in $\Pi$ pertaining to variables $\pi((m, b), b')$ is satisfied, as well as the node-measure constraints $\sum_{b \geq b'} q_m(b)\pi((m, b), b') = q_{m+1}(b')$ for all $b'$. By induction, this works for all $m \in [M]$, which concludes the proof.

\subsubsection{Proof of Lemma \ref{lem: Online Linear Optimization}}
We have by the definition of discretized regret:
\begin{align*}
    \textsc{Regret}_\mathcal{B} &= \max_{\bm{b} \in \mathcal{B}} \sum_{\nround=1}^\Nround \mu^\nround_n(\bm{b}) - \mathbb{E}\left[\sum_{\nround=1}^\Nround \mu^\nround_n(\bm{b}^\nround)\right] = \max_{\bm{q} \in \mathcal{Q}} \mathbb{E}\left[\sum_{\nround=1}^\Nround \langle \bm{q}, \bm{w}^\nround\rangle - \sum_{\nround=1}^\Nround \langle \bm{q}^\nround, \bm{w}^\nround\rangle\right]\,, 
\end{align*}
where in the first equality, we applied Equation~\eqref{eq: Loss of policy} which equated the dot product of utilities $\bm{w}^\nround$ and node probability weights $\bm{q}$ to the expected utility of bid vector $\bm{b} \sim \bm{\pi}$ with utilities $\{w_m^t(b)\}_{m \in [M], b \in \mathcal{B}} = \bm{w}^\nround$. Combining the two summations yields the desired result.

\subsection{Proof of Theorem~\ref{thm: OMD}: Online Mirror Descent Algorithm }

\label{sec: Proof of OMD}

\begin{proof}{Proof of Theorem~\ref{thm: OMD}: Online Mirror Descent Algorithm}

    The proof is divided into four parts, similar to the analysis of Algorithm~\ref{alg: Decoupled Exponential Weights}. In the first part, we rigorously show how our algorithm achieves the stated regret. In the second, we verify correctness of our procedure that recovers a policy $\bm{\pi}^\nround$ from $\bm{q}^\nround$. Then, we show the corresponding time and space complexity of our algorithm. Afterwards, we optimize over discretization error to obtain the continuous regret.

    \textbf{Part 1: Regret of Online Linear Optimization.} Recall that from Lemma~\ref{lem: Online Linear Optimization}, we have
    \begin{align}
        \textsc{Regret}_\mathcal{B} = \max_{\bm{q} \in \mathcal{Q}} \mathbb{E}\left[ \sum_{\nround=1}^\Nround \langle \bm{q} - \bm{q}^\nround, \bm{w}^\nround \rangle\right] = \max_{\bm{q} \in \mathcal{Q}} \mathbb{E}\left[\sum_{\nround=1}^\Nround \langle  \bm{q}^\nround - \bm{q}, -\bm{w}^\nround \rangle\right]\, ,
    \end{align}
    where we negate the utility function into a loss function to be consistent with the OLO convention. We follow a standard analysis of OMD, which shows that the optimization step can be solved efficiently and the resulting iterates have bounded regret. For the former, we show that solution to the $\bm{q}$ optimization step in our algorithm $\bm{q}^{\nround} = \text{argmin}_{\bm{q} \in \mathcal{Q}} \eta\langle \bm{q}, -\bm{w}^\nround\rangle + D(\bm{q} || \bm{q}^{\nround-1})$ can be obtained as the projection of the unconstrained minimizer of \[\tilde{q}^{\nround}= \text{argmin}_{\bm{q} \in [0, 1]^{M \times |\mathcal{B}|}} \eta\langle \bm{q}, -\bm{w}^\nround\rangle + D(\bm{q} || \bm{q}^{\nround-1})\] to the space $\mathcal{Q}$ (See Projection Lemma, Lemma 8.6 of \cite{BartokLecNotes2011}). 
    Having characterized the exact form of the OMD iterates, all that remains is to upper bound the regret of OMD with the regret of Be-the-regularized-leader.
    \begin{lemma}[Lemma 9.2 of \cite{BartokLecNotes2011}]
        \label{lem: Be Regularized leader regret}
        Letting $D(\bm{q} || \bm{q}')$ denote the unnormalized KL divergence between $\bm{q}$ and $\bm{q}'$, we have:
        \begin{align*}
            \textsc{Regret}_\mathcal{B} \leq \max_{\bm{q} \in\mathcal{Q}} \mathbb{E}\Big[\eta^{-1} D(\bm{q} || \bm{q}^1) + \sum_{\nround=1}^\Nround \langle \bm{q}^\nround - \tilde{\bm{q}}^{\nround+1}, \bm{w}^\nround \rangle\Big]\,.
        \end{align*}
    \end{lemma} 
     The remainder of the regret analysis closely follows that of Theorem 1 in \cite{OREPS2013}. At a high level, we want to bound the regret of Online Mirror Descent by the regret of the unconstrained Be the
     (Negentropy) Regularized leader, via Lemma~\ref{lem: Be Regularized leader regret} (see Lemma 13 of \cite{LectureNotes2009} for the more general statement and proof of this lemma). 
     We then upper the contribution of the summation term by using the specific definition of the node weight estimators. Similarly, we upper bound the divergence term as a function of the dimension of the space $\mathcal{Q}$.

    To begin, note that our node utility estimators $\widehat{w}_\nitem^\nround(b)$ are unbiased:
    \begin{align}
        \mathbb{E}_{\bm{b} \sim \bm{\pi}^\nround}[\widehat{w}_\nitem^\nround(b)] = \mathbb{E}_{\bm{b} \sim \bm{\pi}^\nround}[\frac{w_\nitem^\nround(b)}{q^{\nround}_\nitem(b)} \textbf{1}_{b = b^{\nround}_\nitem}] = \frac{w_\nitem^\nround(b)}{q^{\nround}_\nitem(b)} \prob_{\bm{b} \sim \bm{\pi}^\nround} (b = b^\nround_\nitem) = \frac{w_\nitem^\nround(b)}{q^{\nround}_\nitem(b)} q^{\nround-1}_\nitem(b) = w_\nitem^\nround(b)\ .
        \label{proof: part1}
    \end{align}
    Now, consider the right hand side of the inequality in  Lemma \ref{lem: Be Regularized leader regret}. As the node utility estimators are unbiased, so we can replace $\bm{w}^\nround$ with $\widehat{\bm{w}}^\nround$. 
    Now, as per Lemma \ref{lem: Be Regularized leader regret}, we can upper bound the expected estimated regret as a function of the unconstrained optimizer $\tilde{\bm{q}}^{\nround+1}$ and the unregularized relative entropy  with respect to the initial state-edge occupancy measure $\bm{q}^1$. Applying the aforementioned lemma to Equation \eqref{proof: part1}, we obtain:
    \begin{align}
        \textsc{Regret}_\mathcal{B} = \max_{\bm{q} \in \mathcal{Q}} \mathbb{E}\left[\sum_{\nround=1}^\Nround \langle \bm{q}^{\nround} - \bm{q}, -\widehat{\bm{w}}^{\nround} \rangle \right] \leq \max_{\bm{q} \in \mathcal{Q}}\mathbb{E}\left[\sum_{\nround=1}^\Nround \langle \bm{q}^{\nround} - \tilde{\bm{q}}^{ \nround+1}, -\widehat{\bm{w}}^{\nround} \rangle + \eta^{-1}D(\bm{q} || \bm{q}^{1}) \right]
    \end{align}
    Applying $\exp(x) \geq 1 + x$ for $x = \exp(\eta \widehat{\bm{w}}^{\nround})$, we obtain $\tilde{\bm{q}}^{ \nround+1} = \bm{q}^\nround \exp(\eta \widehat{\bm{w}}^{\nround}) \geq \bm{q}^{\nround} + \eta \bm{q}^{\nround} \widehat{\bm{w}}^{\nround}$, which yields $\bm{q}^t - \bm{q}^t\exp(\eta\widehat{\bm{w}}^t) \ge  -\eta \bm{q}^t \widehat{\bm{w}}^t$. Plugging this back in:
    \begin{align}
        \textsc{Regret}_\mathcal{B} &\leq \max_{\bm{q} \in \mathcal{Q}} \mathbb{E}\left[\sum_{\nround=1}^\Nround \langle \bm{q}^{\nround} - \bm{q}^{\nround} \exp(\eta \widehat{\bm{w}}^{\nround}), -\widehat{\bm{w}}^{\nround} \rangle + \eta^{-1}D(\bm{q} || \bm{q}^{1}) \right]\\
        &\le  \max_{\bm{q} \in \mathcal{Q}} \mathbb{E}\left[\eta \sum_{\nround=1}^\Nround \sum_{\nitem = 1}^\Nitem \sum_{b \in \mathcal{B}} q^{\nround}_\nitem(b) \widehat{w}^{\nround}_\nitem(b)^2 + \eta^{-1}D(\bm{q} || \bm{q}^{1}) \right] \,.\label{eq: node diff}
    \end{align}
    Note that $\widehat{w}^{\nround}_\nitem(b) = \frac{w_\nitem^\nround(b)}{q^{\nround-1}_{\nitem}(b)} \textbf{1}_{b = b^{\nround}_{\nitem}}$ for all $\nitem \in [\Nitem]$ and $b \in \mathcal{B}$ by definition. Since $w^{\nround}_\nitem(b) \leq 1$ and $\textbf{1}_{b = b^{\nround}_{\nitem}} \leq 1$ we have $\widehat{w}^{\nround}_\nitem(b) \leq \frac{1}{q^{\nround}_\nitem(b)}$ and we continue the above chain of inequalities with:
    \begin{align}
        \textsc{Regret}_\mathcal{B} &\leq \max_{\bm{q} \in \mathcal{Q}} \mathbb{E}\left[\eta \sum_{\nround=1}^\Nround \sum_{\nitem = 1}^\Nitem \sum_{b \in \mathcal{B}} q^{\nround}_{\nitem}(b) \widehat{w}^{\nround}_\nitem(b) \frac{1}{q^{\nround}_{\nitem}(b)}  + \eta^{-1}D(\bm{q} || q^{1}) \right] \label{eq: full info difference appendix}\\
        &= \max_{\bm{q} \in \mathcal{Q}} \mathbb{E}\left[\eta \sum_{\nround=1}^\Nround \sum_{\nitem = 1}^\Nitem \sum_{b \in \mathcal{B}} \widehat{w}^{\nround}_\nitem(b)  + \eta^{-1}D(\bm{q} || \bm{q}^{1}) \right] \, .
    \end{align} 
    
    Recalling that $D(\bm{q} || \bm{q}^1) = \sum_{\nitem \in [\Nitem], b \in \mathcal{B}} q_\nitem(b)\log\frac{ q_\nitem(b)}{q^1_\nitem(b)} - (q_\nitem(b) - q^1_\nitem(b))$, we note that:
    
    \begin{align*}
        D(\bm{q} || \bm{q}^1) &= \sum_{m = 1}^M \sum_{b \in \mathcal{B}} q_m(b)\frac{\log q_m(b)}{\log q^1_m(b)} - q_m(b) + q^1_m(b) \\
        &= \sum_{\nitem=1}^\Nitem \sum_{b \in \mathcal{B}} q_{\nitem}(b)\log q_m(b) - q_m(b)\log q^1_m(b)\,,
    \end{align*} 
    where in the second equality, we used the fact that the elements both $\bm{q}$ and $\bm{q}^1$ all sum to $M$. Selecting $\bm{q}^1_m(\cdot)$ to be the uniform distribution over all $b \in \mathcal{B}$ and using the fact that the entropy of a discrete distribution over $|\mathcal{B}|$ items is $\log |\mathcal{B}|$, we obtain:
    \begin{align*}
        D(\bm{q} || \bm{q}^1) &= -\sum_{\nitem=1}^\Nitem H(\bm{q}_m) + \log |\mathcal{B}|\sum_{\nitem=1}^\Nitem \sum_{b \in \mathcal{B}} q_{\nitem}(b) \\
        &\leq \sum_{\nitem=1}^\Nitem \log |\mathcal{B}| + \log |\mathcal{B}|\sum_{\nitem=1}^\Nitem \sum_{b \in \mathcal{B}} q_{\nitem}(b) = \Theta(M\log|\mathcal{B}|)\,,
    \end{align*}
    where $H(\bm{x}) = -\sum_{x \in \bm{x}} x \log x $ denotes the discrete entropy function.  
    Plugging this back in:
    \begin{align*}
        \textsc{Regret}_\mathcal{B} \leq \mathbb{E}\left[\eta \sum_{\nround=1}^\Nround \sum_{\nitem = 1}^\Nitem \sum_{b \in \mathcal{B}} \widehat{w}^{\nround}_\nitem(b)  + \eta^{-1}\Nitem \log |\mathcal{B}| \right] \leq \eta \sum_{\nround=1}^\Nround \sum_{\nitem=1} \sum_{b \in \mathcal{B}} w^{\nround}_\nitem(b) + \eta^{-1}\Nitem \log |\mathcal{B}| = \eta \Nround \Nitem |\mathcal{B}| + \eta^{-1}\Nitem \log |\mathcal{B}|\,,
    \end{align*}
    where we used unbiasedness of $\widehat{\bm{w}}^{\nround}$. Setting $\eta = \sqrt{\frac{\log |\mathcal{B}|}{|\mathcal{B}|T}}$, we obtain $\textsc{Regret}_\mathcal{B}(\Nround) \leq \Nitem \sqrt{|\mathcal{B}| \Nround \log |\mathcal{B}|}$.

    \textbf{Part 2: Determining Policy $\bm{\pi}$ from Node Probability Measures $\bm{q}$.} 
    Notice that in our regret analysis for both the bandit and full information setting, we do not require explicit knowledge of the policy $\bm{\pi}^t$, so long as it generates the desired node occupancy measure $\bm{q}^t$. In particular, we require a method of converting $\bm{q}^\nround$ to policy $\bm{\pi}^\nround$ which, in turn, is required in order to sample $\bm{b}^\nround$. Recall from Lemma~\ref{lem: QSpace Equivalence} that the mapping from the space of policies $\Pi$ to the space of node weight measures $\mathcal{Q}_\Pi = \mathcal{Q}$ is injective. Thus, for any $\bm{q} \in \mathcal{Q}$, there must exists a $\bm{\pi} \in \Pi$ such that $q(\bm{\pi}) = \bm{q}$. Moreover, the set $\Pi(\bm{q})$ of such $\bm{\pi}$ can be written as the intersection of two polyhedrons, and hence a polyhedron, from which a feasible solution can be computed efficiently (e.g., ellipsoid method), where  $ \Pi(\bm{q})$ is the set of policies $\pi \in [0,1]^{M\times |\mathcal B|\times |\mathcal B|}$ such that 
    \begin{itemize}
        \item $\pi((0, \max \mathcal{B}), b) = q_1(b)$, for any $b \in \mathcal{B}$;
        \item $\pi((0, b), b') = \textbf{1}_{b = b'}$ for any $b, b' < \max \mathcal{B}$;
        \item $q_{m+1}(b') = \sum_{b \in \mathcal{B}} q_m(b) \pi((m, b), b')\}$ for any $b'\in \mathcal B$ and $m \in [M-1]$.
    \end{itemize}

    \textbf{Part 3: Complexity analysis.} One may wonder how to efficiently update the state occupancy measures by computing the minimizer of $\eta\langle \bm{q}, -\widehat{\bm{w}}^\nround\rangle + D(\bm{q} || \bm{q}^{\nround-1})$. The idea is to first solve the unconstrained entropy regularized minimizer with $\tilde{\bm{q}}^{ \nround+1} = \bm{q}^{\nround} \exp(\eta \widehat{\bm{w}}^{\nround})$. We then project this unconstrained minimizer to $\mathcal{Q}$ with:
    \begin{align}
        \bm{q}^{\nround + 1} = \text{argmin}_{\bm{q} \in \mathcal{Q}} D(\bm{q}||\tilde{\bm{q}}^{\nround + 1})\,.
    \end{align}
    Relegating the details to \cite{OREPS2013}, the above constrained optimization problem can be solved as the minimizer of an equivalent unconstrained convex optimization problem with a polynomial (in $\Nitem$ and $|\mathcal{B}|$) number of variables, and therefore, can be computed efficiently. Combining with finding an initial feasible solution to $\Pi(\bm{q})$ as well as the optimization step, we achieve polynomial in $\Nitem, |\mathcal{B}|, \Nround$ total time complexity. For the space complexity, we only need store the values of $\bm{\pi}^\nround$, $\bm{q}^\nround$, and $\widehat{\bm{w}}^\nround$, for a total space complexity of $O(\Nitem |\mathcal{B}|^2)$.   
    
    \textbf{Part 4: Continuous Regret.} To obtain the continuous regret, recall that the discretization error is $O(\frac{\Nitem \Nround}{|\mathcal{B}|})$. As the discretized regret is $O\left(\Nitem \sqrt{|\mathcal{B}| \Nround \log |\mathcal{B}|}\right)$ in the bandit feedback setting, the optimal choice of $|\mathcal{B}|$ is $\Theta(\Nround^{\frac{1}{3}})$, which achieves continuous regret $\textsc{Regret} = O(\Nitem \Nround^{\frac{2}{3}} \sqrt{\log \Nround})$.

\end{proof}

\subsubsection{Proof of Corollary \ref{cor}}
\label{sec: Proof of cor}

We can straightforwardly extend Algorithm~\ref{alg: OMD} to the full information setting. To do this, we note that we can improve Equation~\eqref{eq: full info difference appendix} by instead replacing $\widehat{\bm{w}}^{\nround}$ with $\bm{w}^{\nround}$ in Equation~\eqref{eq: node diff} to obtain:
\begin{align*}
    \sum_{\nround=1}^\Nround \sum_{\nitem=1}^\Nitem \sum_{b \in \mathcal{B}} q^{\nround}_\nitem(b) \widehat{w}^{\nround}_\nitem(b)^2 = \sum_{\nround=1}^\Nround \sum_{\nitem=1}^\Nitem \sum_{b \in \mathcal{B}} q^{\nround}_\nitem(b) w^{\nround}_\nitem(b)^2 \leq \sum_{\nround=1}^\Nround \sum_{\nitem=1}^\Nitem \sum_{b \in \mathcal{B}} q^{\nround}_\nitem(b) = \sum_{\nround=1}^\Nround \sum_{\nitem=1}^\Nitem 1 = \Nround \Nitem\,.
\end{align*}
Setting $\eta = \sqrt{\frac{ \log |\mathcal{B}|}{T}}$, we obtain in the full information setting $\textsc{Regret}_\mathcal{B} = O(\Nitem \sqrt{\Nround \log |\mathcal{B}|})$. We can also compute the optimal choice of $|\mathcal{B}|$ to obtain optimal continuous regret. Using the optimal choice of $|\mathcal{B}|$ being $\Theta(\sqrt{T})$, we achieve continuous regret of $\textsc{Regret} = O(\Nitem \sqrt{\Nround \log \Nround})$. Note that due to the complexity of the optimization sub-routine in the projection step of OMD, for the full information setting, it is preferable to use Algorithm~\ref{alg: Decoupled Exponential Weights} instead.

{\color{black}
\subsection{Proof of Theorem \ref{thm: time varying known finite}}
We now prove the regret bounds of the contextualized version of our decoupled hedge algorithm (Algorithm~\ref{alg: Decoupled Exponential Weights - Time Varying Known Finite}) to handle time-varying valuations.
To begin, we can once again `decouple' the utility per unit-bid pair, but this time conditional on the valuation vector context. In particular, we have:
\begin{align*}
    \mu_n^t(\bm{b}; \bm{v}) =  \sum_{m=1}^M w_m^t(b_m; \bm{v}) = \sum_{m=1}^M (v_m - b_m)1_{b_m \geq b^t_{-m}} \quad \text{and} \quad \widehat{\mu}_n^t(\bm{b}; \bm{v}) =  \sum_{m=1}^M \widehat{w}_m^t(b_m; \bm{v})\,.
\end{align*}
As stated earlier, we define reward estimates based on Equation (6) of \cite{ContextBanditsCrossLearning2019} and our Algorithm~\ref{alg: Decoupled Exponential Weights - Path Kernels}: 
\begin{align*}
    \widehat{w}_m^t(b; \bm{v}) = 1 - \frac{1 - w_m^t(b; \bm{v})}{\sum_{\bm{v} \in \mathcal{V}} \prob(\bm{v}^t = \bm{v}) q_m^t(b; \bm{v})} \textbf{1}_{b_m^t = b} = 1 - \frac{1 - w_m^t(b; \bm{v})}{Q_m^t(b)} \textbf{1}_{b_m^t = b}\,.
\end{align*}
Here, $q_m^t(b; \bm{v}) = \prob(b^t_m = b | \bm{v}^t = \bm{v}) = \sum_{\bm{b}: b^t_m = b} \prob(\bm{b}^t = \bm{b} | \bm{v}^t = \bm{v})$ is the probability of selecting bid $b$ in slot $m$ with valuation $\bm v$. Similarly, $Q_m^t(b)$ is the probability of selecting bid $b$ for unit $m$, averaged across all possible valuations. One can verify unbiasedness of this estimator $\mathbb{E}[\hat{w}_m^t(b; \bm{v})] = w_m^t(b; \bm{v})$ for all $m \in [M], b \in \mathcal{B}, \bm{v} \in \mathcal{V}$. The second moment can similarly be computed as:
\begin{align*}
    \mathbb{E}[\widehat{w}_m^t(b; \bm{v})^2] = \mathbb{E}\left[\left( 1 - \frac{1-w_m^t(b; \bm{v})}{Q_m^t(b)} \textbf{1}_{b_m^t = b} \right)^2 \right] = 1 - 2\mathbb{E}\left[\frac{1-w_m^t(b; \bm{v})}{Q_m^t(b)}\textbf{1}_{b_m^t=b}\right] + \mathbb{E}\left[\left(\frac{1-w_m^t(b; \bm{v})}{Q_m^t(b)}\right)^2\textbf{1}_{b_m^t=b}\right]\,.
\end{align*}

Evaluating the expectations and recalling that $\mathbb{E}[\textbf{1}_{b_m^t = b}] = Q_m^t(b)$, we have:
\begin{align*}
    \mathbb{E}[\widehat{w}_m^t(b; \bm{v})^2] = 1 - \left[2 - 2w_m^t(b; \bm{v})\right] + \left[\frac{(1 - w_m^t(b; \bm{v}))^2}{Q_m^t(b)}\right] = 2w_m^t(b; \bm{v}) - 1 + \frac{1}{Q_m^t(b)} \leq 1 + \frac{1}{Q_m^t(b)} \leq \frac{2}{Q_m^t(b)}\,.
\end{align*}

Using this, the proof largely follows that of Algorithm~\ref{alg: Decoupled Exponential Weights - Path Kernels} up until Equation (\ref{eq: full info difference}). In particular, we have that the contextual regret can be written as:
\begin{align*}
    \textsc{Regret}_\mathcal{B}(F_{\bm{v}}) &= \mathbb{E}_{F_{\bm{v}}}\left[\sum_{t=1}^{T} \mu_n^t(\bm{b}'; \bm{v}^t) - \sum_{t=1}^{T} \mathbb{E}[\mu^t(\bm{b}^\nround; \bm{v}^t)]\right]\\
    &\lesssim \eta^{-1}M\log|\mathcal{B}| + \eta \mathbb{E}_{F_{\bm{v}}}\left[\sum_{\nround=1}^\Nround \sum_{\bm{b}} \prob(\bm{b}^\nround =\bm{b}| \bm{v}^t = \bm{v}) \mathbb{E}[(\sum_{m=1}^M \widehat{w}^\nround_m(b_m; \bm{v}))^2]\right]\\
    &= \eta^{-1}M\log|\mathcal{B}| + \eta \left[\sum_{\nround=1}^\Nround \sum_{\bm{v} \in \mathcal{V}} \prob(\bm{v}^t = \bm{v})\sum_{\bm{b}} \prob(\bm{b}^\nround =\bm{b}| \bm{v}^t =  \bm{v}) \mathbb{E}[(\sum_{m=1}^M \widehat{w}^\nround_m(b_m; \bm{v}))^2]\right]\\
    &= \eta^{-1}M\log|\mathcal{B}| + \eta M\left[\sum_{\nround=1}^\Nround \sum_{m=1}^M \sum_{\bm{v} \in \mathcal{V}} \prob(\bm{v}^t = \bm{v})\sum_{b \in \mathcal{B}} \mathbb{E}[\widehat{w}_m^t(b; \bm{v})^2]  \sum_{\bm{b}: b_m = b} \prob(\bm{b}^\nround =\bm{b}| \bm{v}^t = \bm{v})\right]\\
    &= \eta^{-1}M\log|\mathcal{B}| + \eta M\left[\sum_{\nround=1}^\Nround \sum_{m=1}^M \sum_{\bm{v} \in \mathcal{V}} \prob(\bm{v}^t = \bm{v})\sum_{b \in \mathcal{B}} \mathbb{E}[\widehat{w}_m^t(b; \bm{v})^2]  q_m^t(b; \bm{v})\right]\\
    &= \eta^{-1}M\log|\mathcal{B}| + 2\eta M\left[\sum_{\nround=1}^\Nround \sum_{m=1}^M \sum_{\bm{v} \in \mathcal{V}} \prob(\bm{v}^t = \bm{v})\sum_{b \in \mathcal{B}} \frac{1}{Q_m^t(b)}  q_m^t(b; \bm{v})\right]\\
    &= \eta^{-1}M\log|\mathcal{B}| + 2\eta M\left[\sum_{\nround=1}^\Nround \sum_{m=1}^M \sum_{b \in \mathcal{B}} \frac{1}{Q_m^t(b)} \sum_{\bm{v} \in \mathcal{V}} \prob(\bm{v}^t = \bm{v}) q_m^t(b; \bm{v})\right]\\
    &= \eta^{-1}M\log|\mathcal{B}| + 2\eta M\left[\sum_{\nround=1}^\Nround \sum_{m=1}^M \sum_{b \in \mathcal{B}} \frac{1}{Q_m^t(b)} Q_m^t(b)\right]\\
    &\leq \eta^{-1}M\log|\mathcal{B}| + \eta M^2|\mathcal{B}|T\,.
\end{align*}
(We will show the first inequality shortly.)
With $\eta = \Theta(\sqrt{\frac{\log |\mathcal{B}|}{M|\mathcal{B}|T}})$,  this yields  the  discretized contextual regret upper bounds of $O(M^{\frac{3}{2}}\sqrt{|\mathcal{B}| T \log |\mathcal{B}|})$ under the bandit setting. Accounting for the rounding error of order $O(\frac{MT}{|\mathcal{B}|})$, we obtain the stated continuous contextual regret upper bounds by optimizing with $|\mathcal{B}| = M^{-\frac{1}{3}}T^{\frac{1}{3}}$. 
To obtain the full information results, we simply replace $\widehat{w}_m^t(b_m; {\bm v}^t)$ with $w_m^t(b_m; {\bm v}^t)$ in the second line of the above equations, which leads to  the discretized contextual regret upper bounds of $O(M^{\frac{3}{2}}\sqrt{ T \log |\mathcal{B}|})$, as desired.

Next, following the proof of Algorithm \ref{alg: Decoupled Exponential Weights - Path Kernels}, we show   the first inequality. 
We define the potentials with respect to a fixed valuation vector $\bm{v}$: $\Phi^t(\bm{v}) = \sum_{\bm{b} \in \mathcal{B}^{+\Nitem}} \exp(\eta \sum_{\tau=1}^{t} \widehat{\mu}^\tau(\bm{b}; \bm{v}^\tau))$. Taking the ratio of adjacent terms, we obtain:
\begin{align*}
    \frac{\Phi^t(\bm{v})}{\Phi^{t-1}(\bm{v})} = \sum_{\bm{b} \in \mathcal{B}^{+M}} \frac{\exp(\eta \sum_{\tau=1}^{t-1} \widehat{\mu}^\tau(\bm{b}; \bm{v}^\tau))}{\Phi^{t-1}(\bm{v})} \exp(\eta \widehat{\mu}^t(\bm{b}; \bm{v}^t)) = \sum_{\bm{b} \in \mathcal{B}^{+M}} \prob(\bm{b}^\nround = \bm{b} | \bm{v}^\nround = \bm{v}) \exp(\eta \widehat{\mu}^t(\bm{b}; \bm{v}^\nround))\,,
\end{align*}
Where in the last equality, we used the condition that our algorithm samples bid vector $\bm{b}$ with probability proportional to $\exp(\eta \sum_{\tau=1}^{t-1} \widehat{\mu}^\nround(\bm{b}; \bm{v}))$ at round $\nround$ with valuations $\bm{v}^t = \bm{v}$. Combining this with inequalities $\exp(x) \leq 1 + x + x^2$ and $1 + x \leq \exp(x)$ for all $x \leq 1$, we obtain:
\begin{align*}
    \frac{\Phi^t(\bm{v})}{\Phi^{t-1}(\bm{v})} \leq \sum_{\bm{b} \in \mathcal{B}^{+M}} \prob(\bm{b}^\nround = \bm{b} | \bm{v}^\nround = \bm{v}) \exp(\eta \widehat{\mu}^t(\bm{b}; \bm{v})) \leq \exp(\sum_{\bm{b} \in \mathcal{B}^{+M}} \prob(\bm{b}^\nround = \bm{b} | \bm{v}^\nround = \bm{v}) \left[\eta\widehat{\mu}^t(\bm{b}; \bm{v}) + \eta^2 \widehat{\mu}^t(\bm{b}; \bm{v})^2 \right])\,.
\end{align*}
Combining this with Equations~\eqref{eq: Potentials} and the fact that $\Phi^0(\bm{v}) = M\log |\mathcal{B}|$, for any fixed bid vector $\bm{b}'$, we have:
\begin{align*}
    \sum_{t=1}^{T} \widehat{\mu}^t(\bm{b}'; \bm{v}) - \sum_{t=1}^{T} \sum_{\bm{b}} \prob(\bm{b}^\nround = \bm{b}| \bm{v}^t = \bm{v})  \widehat{\mu}^t(\bm{b}; \bm{v}) &\lesssim \eta^{-1} \Nitem \log |\mathcal{B}| + \eta \sum_{t=1}^T \sum_{\bm{b}}\prob(\bm{b}^\nround = \bm{b} | \bm{v}^\nround = \bm{v}) \widehat{\mu}^t(\bm{b}; \bm{v})^2\\
    &= \eta^{-1} \Nitem \log |\mathcal{B}| + \eta \sum_{t=1}^T \sum_{\bm{b}}\prob(\bm{b}^\nround = \bm{b} | \bm{v}^\nround = \bm{v}) (\sum_{m=1}^M \widehat{w}_m^t(b_m; \bm{v}))^2\,.
\end{align*}
Taking expectations over $\bm{b}$ and the supremum over all $\bm{b}'$ yields the desired first crucial regret inequality.

As for the time and space complexity, notice that the only algorithmic difference between Algorithm~\ref{alg: Decoupled Exponential Weights - Time Varying Known Finite} and Algorithm~\ref{alg: Decoupled Exponential Weights - Path Kernels} is precisely in computing the estimator, which in the former, requires having to compute the weights $Q_m^t(b)$ by iterating over all $\bm{v} \in \mathcal{V}$. As we also have to store reward estimates for each possible valuations, both the time complexity and space complexity of Algorithm~\ref{alg: Decoupled Exponential Weights - Time Varying Known Finite} are a factor $|\mathcal{V}|$ larger than in Algorithm~\ref{alg: Decoupled Exponential Weights - Path Kernels}, which are $O(M|\mathcal{B}| |\mathcal{V}| T)$ and $O(M|\mathcal{B}| |\mathcal{V}|)$ respectively.

}

\section{Online Appendix - Additional Discussion and Experiments}

In this section, we run several additional experiments and include further discussion of our previous results as well as these new experiments. First, we discuss the interaction between the underlying parameters $N, M, |\mathcal{B}|$ and the $c = c_{-n}$ condition required to guarantee existence of a PNE in Theorem~\ref{thm: PNE existence}. We then justify our use of the $\textsc{EXP3-IX}$ weight estimator rather than the unbiased estimator in the implementation of our algorithms in Section~\ref{sec: experiments}. We also discuss a method that achieves uniform exploration per item as per the $\textsc{EXP3.P}$ algorithm described in \cite{Lattimore2020}. Thirdly, we include several experiments omitted in the main body regarding faster convergence with a larger, more competitive market. We then conclude with a discussion of the practicality of the PAB vs. Uniform Price auction formats from the market design perspective.

\subsection{Nash Equilibrium Existence}

\label{sec: Nash equilibrium existence further discussion}

The assumption of competitiveness $c = c_{-n}$ used in Theorem~\ref{thm: PNE existence} can be relaxed. The key idea behind this assumption is that no individual bidder $n$ cannot lower their bids such that their decreased payment offsets their decreased allocation. To see this in effect without requiring the $c = c_{-n}$ condition, consider the following example. 
\begin{enumerate}
    \item Let $N=3, M=3$, and $\mathcal{B} = \{\frac{i}{10}\}_{i \in [10]}$.
    \item Let $v_1 = [1, 0, 0]$, $v_2 = [1, 0.7 - \epsilon, 0], v_3 = [1, 0.4, 0]$ for some small $\epsilon > 0$.
    \item Here, $c = 1$ but $c_{-1} = c_{-3} = 0.7 - \epsilon \neq c_{-2} = 0.4$. Despite $c \neq c_{-n}$, we have that a PNE exists in the form $[0.6, 0, 0], [0.5, 0.5, 0], [0.5, 0.4, 0]$. To verify that this is indeed a PNE, each of the bidders obtains an allocation of 1, with corresponding utilities of 0.4, 0.5, and 0.5 respectively. Considering only uniform winning bids, which contain the set of optimal responses as per Lemma~\ref{lem: near-uniform optimal bidding}, we see the bidders are bidding at Nash:
    \begin{enumerate}
        \item Bidder 1 only demands one item. Here is $b_{2,2} = 0.5$ is bidder 1's $M$'th largest competing bid. Since this belongs to a higher tie-break priority bidder, they must bid strictly higher, yielding a clearing price of 0.6 for bidder 1. Thus, bidder 1 cannot increase or decrease their bids and therefore must be playing their optimal response.
        \item Bidder 2 can win two items at a price of 0.6, which yields utility $0.5-\epsilon < 0.5$. Similarly, in order to win one item, their $M$'th highest competing bid is $b_{3,2} = 0.4$. Since this belongs to a higher tie-break priority bidder, they must bid strictly higher, yielding a clearing price of 0.5. Thus, bidder 2 cannot increase or decrease their bids and therefore must be playing their optimal response.
        \item Bidder 3 can win two items at a price of 0.6, which yields utility $0.2 < 0.5$. Similarly, in order to win one item, their $M$'th highest competing bid is $b_{2,2}=0.5$. Since this belongs to a lower tie-break priority bidder, this yields a clearing price of $0.5$ for bidder 3. Thus, bidder 3 cannot increase or decrease their bids and therefore must be playing their optimal response.
    \end{enumerate}
\end{enumerate}

We have shown that the $c = c_{-n}$ condition is not necessary for PNE existence. Unfortunately, giving a simple characterization of conditions that guarantee PNE existence is non-trivial. However, as mentioned at the beginning of this section, any equilibrium bids $(\bm{b}_1^*,\ldots,\bm{b}_N^*)$ requires that individual bidders cannot i) profitably sacrifice allocation in order to reduce costs and ii) profitably increase their bids to increase allocation. Let $\tilde{b}_{-n,m}$ be as in Lemma~\ref{lem: near-uniform optimal bidding}: $\tilde{b}_{-n,m}$ denotes bidder $n$'s $m$'th smallest competing bid rounded up to the next multiple of $\frac{1}{|\mathcal{B}|}$ if belonging to bidder $n' > n$ due to tie-breaking. In other words, submitting $\tilde{b}_{-n,m}$ for the first $m$ items minimizes the cost required to win $m$ items. As such, we can construct an efficient frontier $\bm{\mu}_n(\bm{b}_{-n}^*) \doteq \{\sum_{i=1}^m v_{n,i} - m\tilde{b}_{-n,m}\}_{m \in [M]}$ of bidder $n$'s utility by minimizing the cost required to win $m$ items. If this frontier is maximized at $\bm{b}_n^*$ for all $n$, then $(\bm{b}_1^*,\ldots,\bm{b}_N^*)$ constitutes a PNE. More formally, a PNE exists at $(\bm{b}_1^*,\ldots,\bm{b}_N^*)$ if:
\begin{align}
    b_{n,m}^* = \min\left(\tilde{b}_{-n,m^*}, \lfloor v_{n,m} \rfloor_{\delta}\right) \quad \text{where} \quad m^* = \text{argmax}_{m \in [M]} \bm{\mu}_n(\bm{b}_{-n}^*) \quad \forall m \in [M], n \in [N]\,.
\end{align}

With this, we can interpret the relaxed condition as saying that bidders prefer to win as many units as possible (subject to positive marginal per-unit utility) over winning fewer units at a reduced cost. This can equivalently interpreted as bidders being price-takers and having little manipulative power over the market price. While even this weaker assumption can be easily violated in artificial examples, e.g., the example in our bid convergence remark, we claim that it is a reasonable assumption in many of the PAB auction's real world applications and markets.
 
\begin{enumerate}
    \item Bidders only have high utility for a small number of units, i.e., the number of units $M$ each bidder $n$ demands is much smaller than the supply $\overline{M}$.
    This is reasonable in many relevant markets, e.g. electricity or emissions, as the total supply far exceeds any single firm's power usage or pollution capacity. 
    \item The total supply is smaller than the aggregate demand, i.e., $\sum_{n \in [N]} \overline{M}_n \gg M$.
    For example, electricity supply in the US in the early 2020's has been strained due to increasing power requirements across many industries, due in part to the rise of energy intensive AI technologies, cloud computing, and cryptocurrency mining. This is also the case in carbon markets, the supply $M$ is artificially limited so as to keep the prices from falling too low and becoming too weak of a disincentive (e.g. carbon markets) to pollute.
\end{enumerate}
In conjunction, these conditions imply that no individual bidder can sufficiently impact prices as their demand is too small compared to the aggregate demand and supply, thus, their reduced costs from underbidding is outweighed by their reduction in allocation. We do note that this result differs from existing characterizations of efficient Pure Nash Equilibria in PAB auctions \cite{inefficiency2013}. However, the Nash equilibria described in the latter require overbidding (bidding above one's marginal valuations) for units, so long as they are guaranteed not to win one of these units. In particular, if every bidder submits the same vector of $[b,\ldots,b]$ where $b$ is the $M$'th largest valuation among all bidders, and ties are broken in favor of bidders with the highest value, then this is also an efficient Pure Strategy Nash equilibrium. 

To illustrate this more quantitatively, we analyze the effects of changing $M, |\mathcal{B}|, N$ on the probability that a PNE (assuming no overbidding) exists in Figure~\ref{fig: PNE existence vs M B N}. Of course as $|\mathcal{B}|$ gets smaller and the number of agents $N$ becomes larger, the probability that the rounded-down $M$'th highest-other-valuation is equal for all $n$ increases. Curiously, the probability is also increasing, albeit slowly, as a function of $M$.

\begin{figure}
    \centering
    \includegraphics[scale=0.43, trim={0 33 0 0},clip]{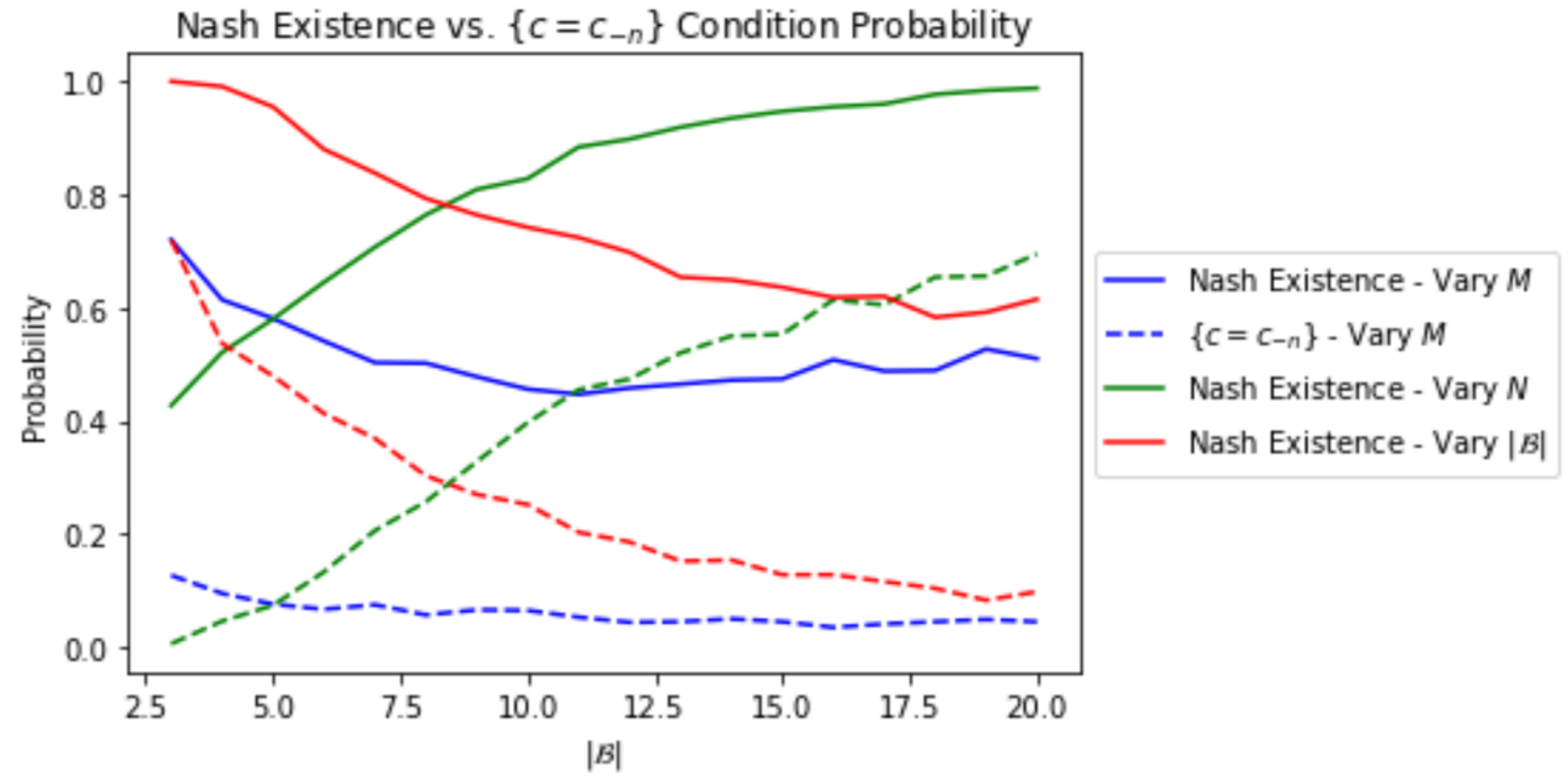}
    \caption{We plot the probability of the existence of a PNE versus the probability of satisfying the condition $c = c_{-n}$ as in Theorem~\ref{thm: PNE existence} as a function of $M, |\mathcal{B}|, N$ which is varied over the horizontal axis. Here, we fix two of $M = 5, N = 5, |\mathcal{B}| = 10$ when varying the remaining parameter between $3$ and $20$. The valuation vectors are constructed by drawing and sorting $M$ valuations from a Unif(0, 1) distribution.}
    \label{fig: PNE existence vs M B N}
\end{figure}

\subsection{$\textsc{EXP3-IX}$ vs. Unbiased Reward Estimator}

\label{sec: IX}

In the experiments section, we ran a slightly modified version of our existing algorithms in the bandit feedback setting. We do this as the variance of the accumulated regret of our algorithms are high, as the node weight estimators normalize over vanishingly small probabilities $q^t_m(b)$. To mitigate the effect of such normalization, we use the $\textsc{EXP3-IX}$ estimator as described in \cite{Lattimore2020}. Under this estimator, rather than normalizing the probability of selecting bid $b^t_m$ for unit $m$ at time $t$ by $q^t_m(b_m^t)$, we instead normalize it by $q^t_m(b_m^t) + \gamma$ for some constant $\gamma > 0$. In the standard $K$-armed bandit setting, despite being a biased estimator, still achieves the same sublinear expected regret guarantee  with a smaller variance. This smaller variance indeed allows for stronger high probability guarantees on the magnitude of our regret; i.e., for $\delta > 0$ and $\gamma = \sqrt{\frac{\log(K) + \log(\frac{K+1}{\delta})}{4K\Nround}}$, the $\textsc{EXP3-IX}$ algorithm guarantees with probability at least $1 - \delta$ that the regret is upper bounded by $C\sqrt{KT\log K}$ for some absolute constant $c > 0$. We extend this algorithm to the multi-unit PAB setting algorithms, where for each node $(m, b)$, we set $\gamma = \sqrt{\frac{\log(K) + \log(\frac{K+1}{\delta})}{4K\Nround}}$ and $K = |\{b \in \mathcal{B}: b \leq v_m\}|$, for $\delta = 0.05$. Aside from the change in node weight estimators, the $\textsc{EXP3-IX}$ versions of Algorithms \ref{alg: Decoupled Exponential Weights - Path Kernels} and \ref{alg: OMD} are exactly the same.\\

\subsubsection{Empirical Performance of Original Algorithms vs. $\textsc{EXP3-IX}$ Variants}

In this section, we empirically analyze the modified variants of our algorithms which use the biased, but lower variance $\textsc{EXP3-IX}$ node-weight estimators (see Appendix~\ref{sec: IX}). We compare the distribution of the regret recovered by these modified algorithms versus the non-modified versions when the number of units is one. The bidder, endowed with valuation vector $\bm{v} = [1]$, will compete against a single adversary over the course of $T$ rounds for $\overline{M}=M=1$ item. This is the standard first price auction (FPA). Here, we compare performance when the adversary is stochastic (bids drawn uniformly random from $[0, 1]$) versus adaptive adversary (running the same algorithm, with a valuation drawn uniformly random from $[0, 1]$). 

We plot the regret of the bidder against the stochastic and adversarial competitors for moderate $T \in \{100, 500, 2000, 10000\}$. The stochastic adversary setting is shown in Figure~\ref{fig:eval_IX} (a) and the adversarial setting is shown in Figure~\ref{fig:eval_IX} (b). 
We observe that while the $\textsc{EXP3-IX}$ variants marginally worsens regret for small values of $T \in \{100, 500\}$ for both the stochastic and adaptive settings, it significantly mitigates the heavy tailed distribution of regret for large $T \in \{2000, 10000\}$, especially in the adversarial setting.

\begin{figure}
     \centering
     \begin{subfigure}[b]{0.3\textwidth}
         \includegraphics[scale=0.55]{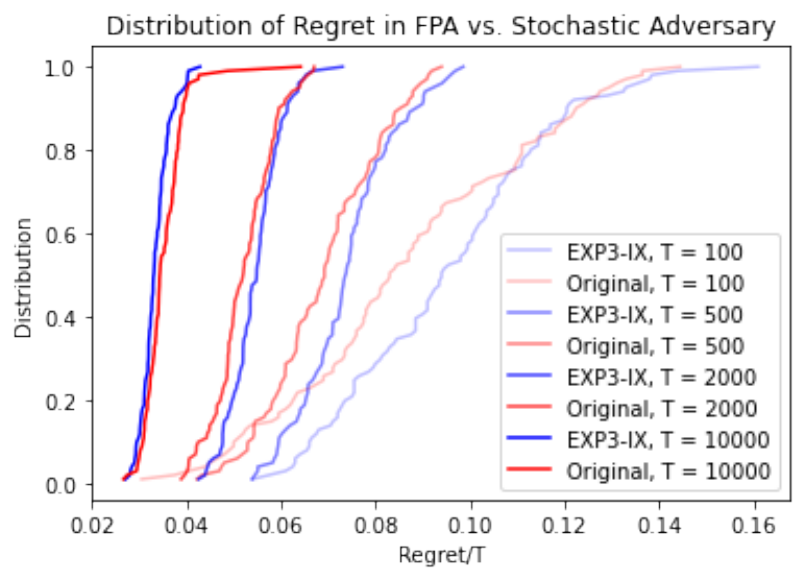}
     \end{subfigure}
     \hfill
     \begin{subfigure}[b]{0.51\textwidth}
         \includegraphics[scale=0.55]{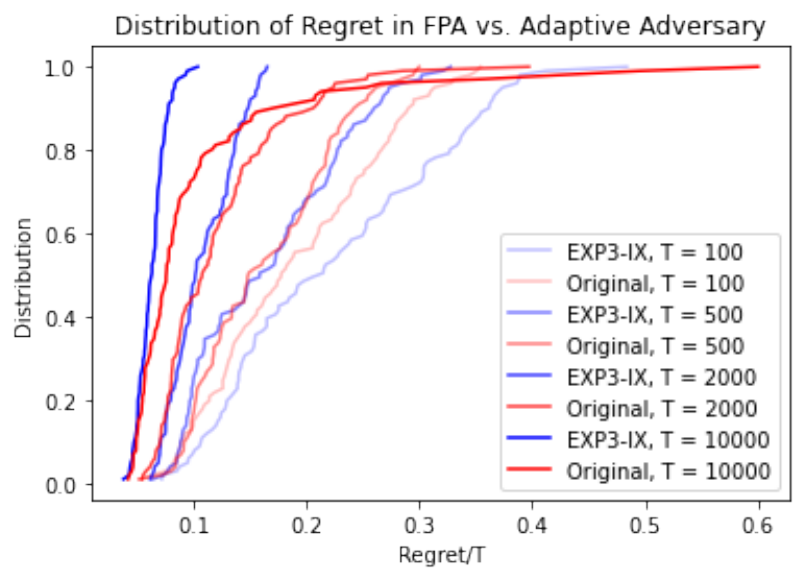}
     \end{subfigure}
    \caption{Distribution of regret when using OMD vs its $\textsc{EXP3-IX}$ variant against stochastic (left) and adaptive adversaries for varying $T$ (right).}
    \label{fig:eval_IX}
    \label{fig:exp1_reg}
\end{figure}

\subsubsection{$\textsc{EXP3.P}$ vs. Unbiased Reward Estimator}

\label{sec: EXP3P}

One downside of the $\textsc{EXP3-IX}$ node weight estimator approach is that the added exploration aggregates over layers $m \in [M]$ in an uneven manner, as bid vectors must stay monotonic. As such, this per-node re-weighting does not guarantee convergence of the regret distribution in probability and only maintains the weaker guarantee over expectation convergence. To mitigate the effect of such uneven exploration aggregation, we instead modify the $\textsc{EXP3.P}$ algorithm as described in \cite{ Lattimore2020}. This algorithm explicitly mixes in uniform noise into the decisions made by the algorithm, and then normalizes accordingly. How does one sample uniformly from the exponentially large space of all monotone, individually rational bid vectors? We claim that the following procedure straightforwardly achieves such uniform mixing:
\begin{enumerate}
    \item Select $b_1$ uniformly at random from $\mathcal{B}$ subject to $b_1 < v_1$.
    \item If $b_m < v_{m+1}$, then set $b_{m+1} = b_m$. Otherwise, select $b_{m+1}$ uniformly at random from $\mathcal{B}$ subject to $b_{m+1} < v_{m+1}$.
\end{enumerate}
With a random exploration probability of $\gamma \in (0, 1)$, the $\textsc{EXP3.P}$ variant of our algorithm performs the above uniform exploration with probability $\gamma$ and follows the procedure in Algorithm~\ref{alg: Decoupled Exponential Weights - Path Kernels} otherwise. Now, we must also account for this in our node weight estimators. In particular, under the $\textsc{EXP3.P}$ variant, the probability of selecting bid $b^t_m$ for unit $m$ is given by $\hat{q}^t_m(b_m^t) = (1-\gamma)q^t_m(b_m^t) + \frac{\gamma}{|b \in \mathcal{B}: b < v_{m}|} = (1-\gamma)q^t_m(b_m^t) + \frac{\gamma \delta}{\lfloor v_{m}\rfloor_\delta}$. We then update the rewards of all bids $b < v_m$ in layer $m$ as:
\begin{align*}
    \widehat{W}^{\nround+1}_{\nitem}(b) \gets \widehat{W}^{\nround}_{\nitem}(b) + \frac{(v_m - b)\textbf{1}_{b \geq b^t_{-m}} + \beta_m}{\hat{q}^t_m(b)} \textbf{1}_{b^t_m = b})
\end{align*}
where $\beta_m = \Theta(\sqrt{\frac{\log(|\mathcal{B}| T/\delta)}{|\mathcal{B}|T}})$ for some high probability parameter $\delta \in (0, 1)$. In the standard $K$-armed bandits problem, the $\textsc{EXP3.P}$ algorithm guarantees that the regret is bounded above by $C\sqrt{KT\log(K/\delta)}$ with probability at least $1-\delta$ for some universal constant $C$. This high probability bound follows immediately from bounding the variance of the weight estimators. Of course, we cannot blindly apply this bound in our multi-unit setting, as there are an exponentially large number of possible bid vectors. Fortunately, we may apply the same utility decoupling trick as in the analyses of Algorithms~\ref{alg: Decoupled Exponential Weights - Path Kernels} and ~\ref{alg: OMD}. That is, we upper bound the variance of the reward estimate of a bid vector by the sum of the variances of each of its constituent bids, yielding a $1-\delta$ high probability bound of $CM^{\frac{3}{2}}\sqrt{|\mathcal{B}|T\log(|\mathcal{B}|/\delta)}$ in the bandit setting.

\subsection{Experiments with Larger $N$}

In this section, we empirically show the faster convergence of and superior welfare and revenue of more competitive PAB markets. More specifically, we evaluate the impact of competition by running the same revenue and welfare over time experiments for both the PAB and uniform price auctions with varying values of $N$. In Figure \ref{fig:comp}, we compare the distribution of welfare and revenue over time showing the 10th, 25th, 50th, 75th, and 90th percentiles for the bandit setting for a varying number of market participants $N \in \{2, 6, 12, 48\}$ with $T = 10^4$. Compared with the previous revenue and welfare plots with $N = 3$ bidders (Figure~\ref{fig:rev_wel_over_time}), we see that as $N$ grows larger, the welfare and revenue more smoothly, and with lower variance, increase towards their equilibrium values. In addition to the increased competition reducing the incentive for agents to strategically shade their bids, the final revenue is higher as expected.

\begin{figure}

     \centering
     \begin{subfigure}[b]{0.3\textwidth}
         \includegraphics[scale=0.52, trim={0 0 0 22},clip]{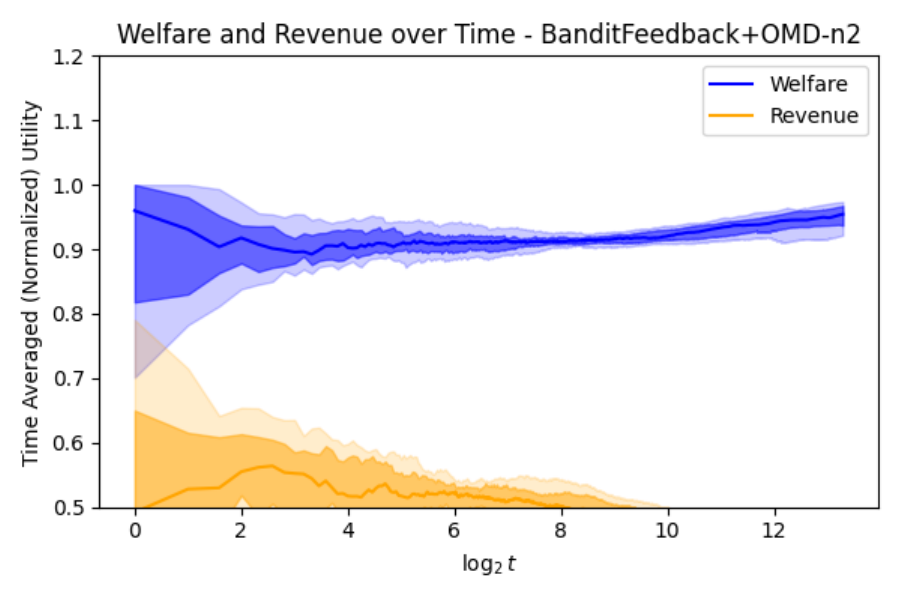}
     \end{subfigure}
     \hfill
     \begin{subfigure}[b]{0.51\textwidth}
         \includegraphics[scale=0.52, trim={0 0 0 22},clip]{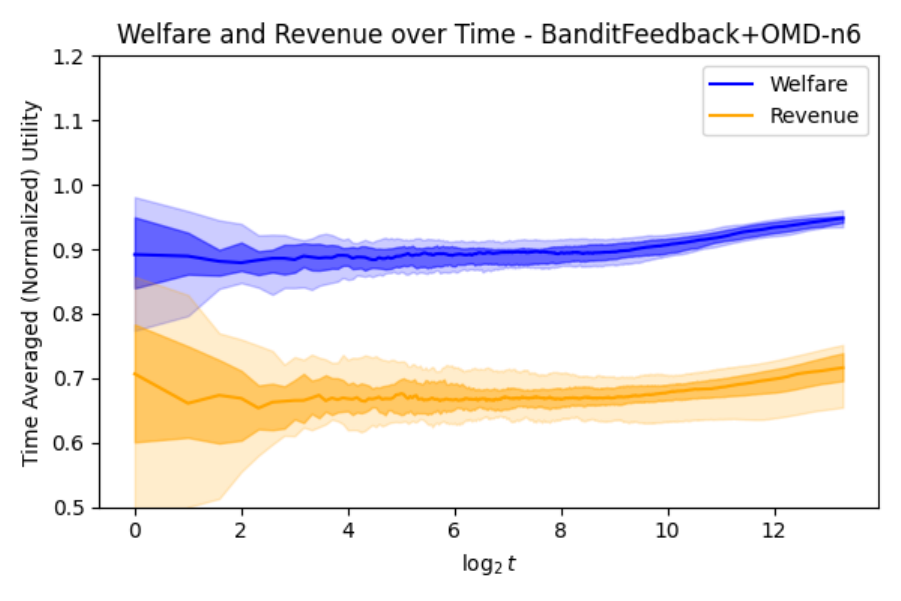}
     \end{subfigure}
     \begin{subfigure}[b]{0.3\textwidth}
         \includegraphics[scale=0.52, trim={0 0 0 22},clip]{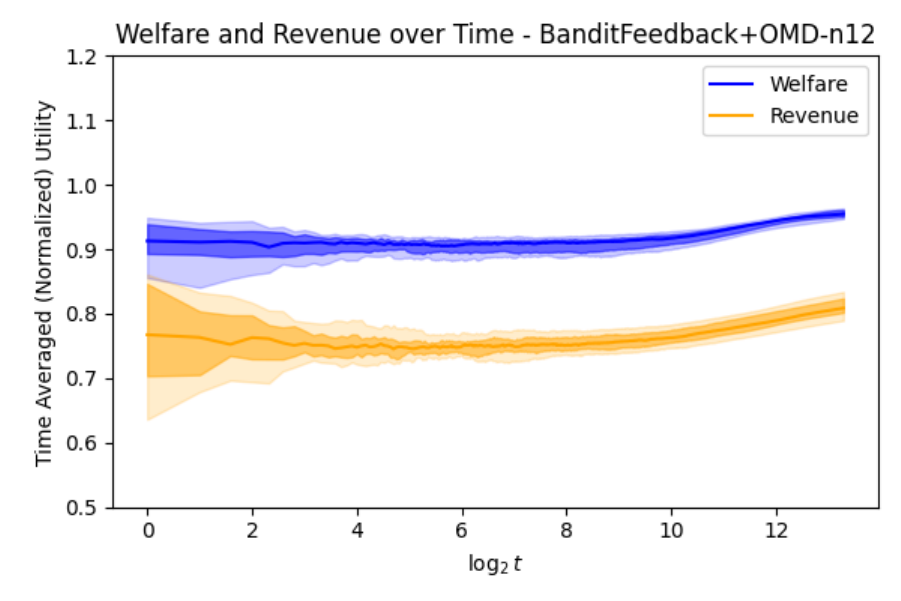}
     \end{subfigure}
     \hfill
     \begin{subfigure}[b]{0.51\textwidth}
         \includegraphics[scale=0.52, trim={0 0 0 22},clip]{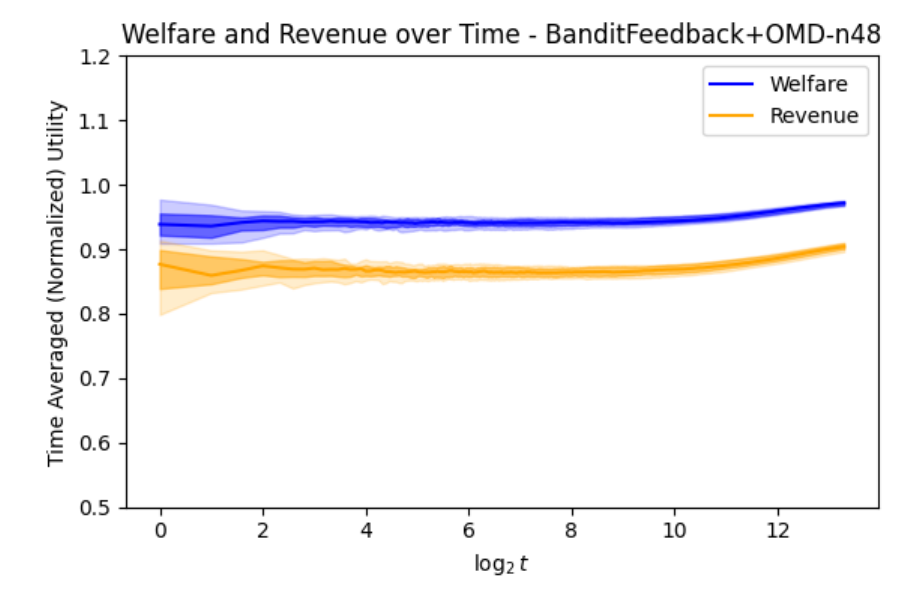}
     \end{subfigure}

    \caption{{Effect of competition: For $M = 5$, $|\mathcal{B}| = 20$, $N = 3$, $T = 10^5$, we show the evolution of the welfare and revenue over time---normalized by the sum of the largest $M$ valuations. We use decoupled exponential weights with bandit feedback (Algorithm~\ref{alg: Decoupled Exponential Weights - Path Kernels}) with $N=2$ (top left), $N=6$ (top right), $N=12$ (bottom left), and $N=48$ (bottom right).} \label{fig:comp}}.
\end{figure}

\subsection{Concluding Remarks}

In this work, we have shown that pay-as-bid (PAB) auctions have several advantages over its more commonly used uniform price counterpart. These advantages include simpler Nash bidding structure requiring only uniform bids (see Lemma~\ref{lem: PNE uniform bidding} and \ref{lem: near-uniform optimal bidding}) where this structure empirically extends to the dynamic setting. Additionally, we show that our utility decoupling insight enables online learning algorithms for PAB that are computationally and regret dominant compared to learning in uniform price auctions. Lastly, we showed that regardless of the parameterizations of $M, N, |\mathcal{B}|$, the PAB auction routinely out-performed uniform price in terms of revenue whilst remaining competitive in terms of welfare (see Table~\ref{table: learning dynamics full info} and ~\ref{table: learning dynamics full uniform}). This superior revenue and welfare trade-off can be seen more closely in Figure~\ref{fig: welfare_revenue_comparison_box_plot_intro}.

\begin{figure}
     \centering
     \begin{subfigure}[b]{0.3\textwidth}
         \includegraphics[scale=0.66, trim={0 0 0 0},clip]{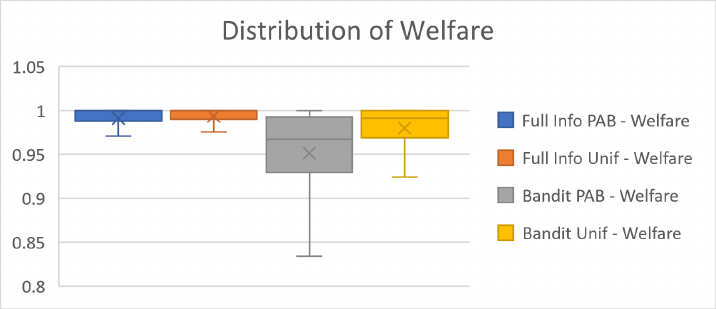}
     \end{subfigure}
     \hfill
     \begin{subfigure}[b]{0.51\textwidth}
         \includegraphics[scale=0.66, trim={0 0 0 0},clip]{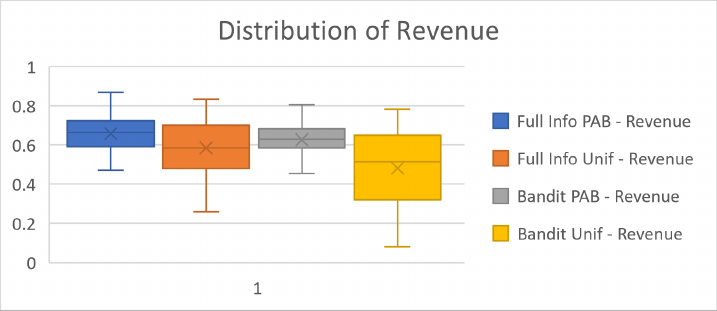}
     \end{subfigure}
    \caption{Welfare (left) and revenue (right) comparisons between the PAB and Uniform Price auctions, under both full information and bandit feedback. We use parameterization $N=3, M = 5, |\mathcal{B}| = 20, T = 10^5$ run over 200 trials with valuations drawn and sorted from a Unif(0, 1) distribution}\label{fig: welfare_revenue_comparison_box_plot_intro}
\end{figure}

Despite these computational and economic advantages of PAB over uniform price, we must still be careful in practice. For example, as PAB achieves higher revenue for the auctioneer, this conversely implies lower consumer surplus. Thus, switching from a uniform price auction to PAB may drive away participants from the platform and join a competing platform. While this is a non-issue for markets where the auctioneer holds a monopoly over a product---e.g., government issued pollution licenses or treasury bills--- this makes switching impractical in certain settings, such as electricity markets. As such, we leave the analysis of learning dynamics in such markets operating under uniform price or markets with multiple platforms for important future work.
\end{document}